\documentclass[conference]{IEEEtran}

\usepackage[noend]{algpseudocode}
\usepackage[linesnumbered,ruled,boxed,vlined]{algorithm2e}
\usepackage{booktabs}
\usepackage{graphicx}
\usepackage{amsmath}
\newtheorem{definition}{Definition}
\usepackage{varioref}
\usepackage[hidelinks]{hyperref}
\usepackage[capitalise,nameinlink]{cleveref}
\usepackage{xcolor}
\usepackage{multirow}

\usepackage{sistyle}
\SIthousandsep{,}

\crefname{lemma}{Lemma}{Lemmas}
\Crefname{lemma}{Lemma}{Lemmas}

\newtheorem{theorem}{Theorem}
\newtheorem{lemma}{Lemma}

\usepackage{subfig}

\newcommand{\vahid}[1]{{\color{brown}#1}}
\newcommand{\clrb}[1]{{{\small\emph{\color{blue}#1}}}}

\newcommand{\route}{{\small\textbf{\texttt{Route-views}}}\xspace}
\newcommand{\bible}{{\small\textbf{\texttt{Bible}}}\xspace}
\newcommand{\advo}{{\small\textbf{\texttt{Advogato}}}\xspace}
\newcommand{\slsh}{{\small\textbf{\texttt{Slash}}}\xspace}
\newcommand{\hyves}{{\small\textbf{\texttt{Hyves}}}\xspace}
\newcommand{\live}{{\small\textbf{\texttt{Live-mocha}}}\xspace}
\newcommand{\youtube}{{\small\textbf{\texttt{Youtube}}}\xspace}
\newcommand{\mnsz}{\ensuremath{minsize}\xspace}
\newcommand{\pl}{\textsc{Baseline}\xspace}
\newcommand{\kqc}{\textsc{KernelQC}\xspace}
\newcommand{\kmax}{\textsc{TopkMaximalQC}\xspace}

\newcommand{\topKqce}{\textsc{Top-$k$ QCE}\xspace}

\newcommand{\qc}{\text{{\ensuremath{\gamma}}-quasi-clique{}}{}\xspace}

\newcommand{\quick}{\textsc{Quick}\xspace}
\newcommand{\quickm}{\textsc{QuickM}\xspace}

\newtheorem{problem}{Problem}
\newcommand{\remove}[1]{}

\begin{document}
\title{Enumerating Top-k Quasi-Cliques}

\author{\IEEEauthorblockN{Seyed-Vahid Sanei-Mehri}
\IEEEauthorblockA{
Iowa State University\\
vas@iastate.edu}
\and
\IEEEauthorblockN{Apurba Das}
\IEEEauthorblockA{
Iowa State University\\
adas@iastate.edu}
\and
\IEEEauthorblockN{Srikanta Tirthapura}
\IEEEauthorblockA{
Iowa State University\\
snt@iastate.edu}
}

\maketitle

\begin{abstract}
	Quasi-cliques are dense incomplete subgraphs of a graph that generalize the notion of cliques. Enumerating quasi-cliques from a graph is a robust way to detect densely connected structures with applications to bio-informatics and social network analysis. However, enumerating quasi-cliques in a graph is a challenging problem, even harder than the problem of enumerating cliques. We consider the enumeration of top-$k$ degree-based quasi-cliques, and make the following contributions: (1)~We show that even the problem of detecting if a given quasi-clique is maximal (i.e. not contained within another quasi-clique) is NP-hard (2)~We present a novel heuristic algorithm \kqc{} to enumerate the $k$ largest quasi-cliques in a graph. Our method is based on identifying {\em kernels} of extremely dense subgraphs within a graph, following by growing subgraphs around these kernels, to arrive at quasi-cliques with the required densities (3)~Experimental results show that our algorithm accurately enumerates quasi-cliques from a graph, is much faster than current state-of-the-art methods for quasi-clique enumeration (often more than three orders of magnitude faster), and can scale to larger graphs than current methods.
\end{abstract}
\vspace{-1ex}

\section{Introduction}
\label{sec:intro}
Finding dense subgraphs within a large graph is a foundational problem in graph mining, with wide applications in bioinformatics, social network mining, and security. Much attention has been paid to the problem of enumerating cliques, which are complete dense structures in a graph, e.g.~\cite{BK73,TTT06,MU04,ELS10,MXT17,SMT15}. Usually, however, dense subgraphs are not cliques. The requirement of complete connectivity among vertices of the graph is often too strict, and there maybe edges missing among some pairs of vertices, or the existence of some edges may not be captured during observation. For example, cliques were found to be overly restrictive in identifying cohesive subgroups in social network analysis~\cite{Alba73,Freeman92}, and instead, dense subgraph models were preferred that did not require complete connectivity. A similar need was found in the analysis of protein interaction networks~\cite{SM03}. This has led to the definition of ``incomplete dense structures'' or ``clique relaxations'' that are dense subgraphs where not every pair of vertices is connected. Such definitions can lead to more robust methods for identifying dense structures in a graph. In addition to being of great practical importance, the study of clique relaxations is of fundamental importance in graph analysis. 

In this work, we consider a type of clique relaxation called a {\em degree-based} quasi-clique in a graph. For a parameter $0 < \gamma \le 1$, a $m$-vertex subgraph $H$ of a graph $G=(V,E)$ is said to be a degree-based $\gamma$-quasi-clique (henceforth called as ``$\gamma$-quasi-clique) if the degree of each vertex in  $H$ is at least $\gamma \cdot (m-1)$. Note that if $\gamma=1$, the definition required $H$ to be a clique. By increasing $\gamma$, it is possible to make a stricter threshold for a subgraph to be admitted as a quasi-clique. If $\gamma < 1$, it is possible for the subgraph to be missing some edges among its vertices and still be admitted as a $\gamma$-quasi-clique.
Quasi-clique mining has been applied in many areas such as biological, social, and telecommunication networks. Specific examples include: detecting co-functional protein modules from a protein interaction network~\cite{BB09}, clustering in a multilayer network~\cite{GF+10,BG+12}, and exploring correlated patterns from an attributed graph~\cite{SMZ12}. A $\gamma$-quasi-clique is said to be {\em maximal} if it is not a proper subgraph of any other larger $\gamma$-quasi-clique. We consider enumerating maximal quasi-cliques. This formulation reduces redundancy in the output by ensuring that if a quasi-clique $Q$ is output, then no other quasi-clique that is contained in $Q$ is also output. Note that a maximal quasi-clique may not be the largest (maximum) quasi-clique in the graph.

We consider top-$k$ maximal quasi-clique enumeration, where it is required to enumerate the $k$ largest maximal quasi-cliques in the graph\footnote{Our methods can also be adapted to enumerate only those quasi-cliques whose size is greater than a given threshold}. There are a few reasons why enumerating top-$k$ maximal quasi-cliques is better than enumerating all maximal quasi-cliques.
(1)~if we focus on the top-$k$, then the output size is no more than $k$ quasi-cliques. Compare this with enumerating all maximal quasi-cliques in a graph, whose output size can be exponential in the size of the input graph. For instance, it is known that there can be as much as $\Omega(3^{n/3})$ maximal cliques in a graph, and hence there can be at least as many maximal quasi-cliques, since each clique is a $\gamma$-quasi-clique with $\gamma=1$. 
(2)~the largest quasi-cliques in a graph are often the most interesting among all the quasi-cliques.
(3)~the time required for enumerating top-$k$ can potentially be smaller than the time for enumerating all maximal quasi-cliques.

A straightforward approach to enumerate top-$k$ maximal quasi-cliques is to first enumerate all maximal quasi-cliques in $G$ using an existing algorithm for quasi-clique enumeration such as \quick~\cite{LW08}, followed by extracting the $k$ largest among them. This approach has the problem of depending on an expensive enumeration of all maximal quasi-cliques. If the number of maximal quasi-cliques is much larger than $k$, then most of the enumerated quasi-cliques are discarded, and the resulting computation is wasteful. It is interesting to know if there is a more efficient way to enumerate the largest maximal quasi-cliques in $G$. In this work, we present progress towards this goal. We make the following contributions:\\

\noindent{\bf NP-hardness of Maximality:} First, we prove that even the problem of detecting whether a given quasi-clique in a graph is a maximal quasi-clique is an NP-hard problem. This is unlike the case of cliques -- detecting whether a given clique is a maximal clique can be done in polynomial time, through simply checking if it is possible to add one more vertex to the clique. Note that our result is not about checking maximum sized quasi-cliques -- it was already known~\cite{PV+18} that finding the maximum sized $\gamma$-quasi-clique in a graph is NP-complete, for any value of $\gamma$. Instead, our result is about checking maximality of a quasi-clique.\\
	
\noindent{\bf Algorithm for Top-$k$ $\gamma$-quasi-cliques:} We present a novel heuristic algorithm \kqc for enumerating top-$k$ maximal quasi-cliques without enumerating all maximal quasi-cliques in $G$. Our algorithm is based on the observation that a $\gamma$-quasi-clique typically contains a smaller but denser subgraph, a $\gamma'$-quasi-clique, for a value $\gamma' > \gamma$. \kqc exploits this fact by first detecting ``kernels'' of extremely dense subgraphs, followed by expanding these kernels into $\gamma$-quasi-cliques in a systematic manner.  \kqc uses the observation that for $\gamma' > \gamma$, it is (typically) much faster to enumerate $\gamma'$-quasi-cliques than it is to enumerate $\gamma$-quasi-cliques. Further, the resulting set of $\gamma'$-quasi-cliques can be expanded into $\gamma$-quasi-cliques more easily than it is to construct the set of $\gamma$-quasi-cliques starting from scratch.\\

\noindent{\bf Experimental Evaluation:} We empirically evaluate our algorithm on large real-world graphs and show that \kqc enumerates top-$k$ maximal quasi-cliques with high accuracy, and is orders of magnitude faster than the baseline, which uses a state-of-the-art algorithm for quasi-clique enumeration. For instance, on the graph \advo\footnote{details of the graphs used in the experiments are presented in Section~\ref{sec:expts}}, \kqc yields a nearly 1000 fold speedup for enumerating the top-100 0.7-quasi-cliques, when compared with a baseline based on the \quick algorithm~\cite{LW08}. 

While \kqc is not guaranteed to return exactly the set of top-$k$ maximal quasi-cliques, it is very accurate in practice. Note that, given that the problem of even checking maximality of a quasi-clique is NP-hard, the cost of exact enumeration of maximal quasi-cliques is necessarily high. In many of the cases that we considered, the output of \kqc exactly matched the output of the exact algorithm that used exhaustive search. Usually, the error in the output, when compared with the output of the exact algorithm, was less than one tenth of one percent. See \cref{sec:expts} for more details on the metrics used to measure the accuracy and performance of \kqc over the baseline algorithm. {\bf Significantly, \kqc was able to scale to much larger graphs than current methods.}

\vspace{-1ex}
\begin{figure*}
	\captionsetup[subfigure]{justification=centering}
	\centering
	\subfloat[Graph $G$]{%
		\includegraphics[width=.24\textwidth]{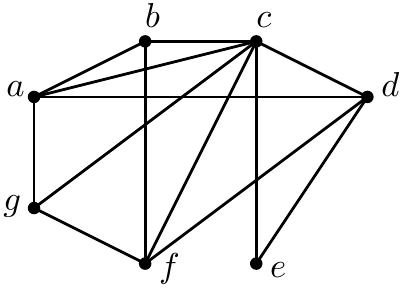}}
	\qquad\qquad
	\subfloat[Quasi clique]{%
		\includegraphics[width=.24\textwidth]{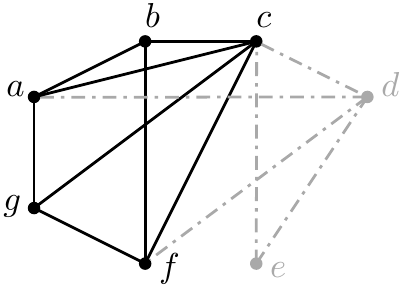}}
	\qquad\qquad
	\subfloat[Maximal Quasi-clique]{%
		\includegraphics[width=.24\textwidth]{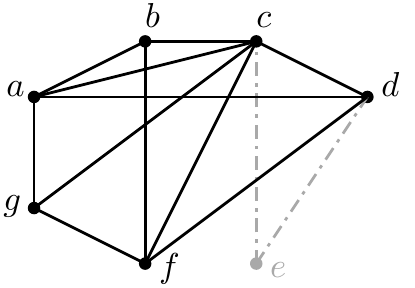}}
	\caption{$\gamma$-quasi-clique with $\gamma = 0.6$ and \mnsz{} = 5. (b) vertices $\{a,b,c,f,g\}$ form a \qc{}. (c) vertices $\{a,b,c,d,f,g\}$ form a maximal \qc{}}
	\label{fig:qc-example}
	\vspace{-4ex}
\end{figure*}
\subsection{Related Works}
\label{sec:related}
\textbf{Degree-based Quasi-Clique: }Motivated by a study on protein sequences, Matsuda et al.~\cite{MIH99} first defined the degree-based $\gamma$-quasi-clique in the context of a protein sequence clustering problem. The degree based $\gamma$-quasi-clique has also been referred to as a $\gamma$-complete-graph in the literature~\cite{K16}. Pei et al.~\cite{PJZ05} study the problem of enumerating those degree-based $\gamma$-quasi-cliques from a graph database that occur in the every graph of the graph database. Zeng et al.~\cite{ZW+07} studied the same problem as Pei et al. but generalize in a sense that their algorithm enumerates degree based $\gamma$-quasi-cliques that occur in at least a certain number of graphs in the database. Note that the algorithms discussed so far can also enumerate all maximal $\gamma$-quasi-cliques. Liu and Wong~\cite{LW08} propose the \quick~algorithm for enumerating all maximal $\gamma$-quasi-cliques from a simple undirected graph that uses a number of pruning techniques, some from prior works, and some newly developed. Lee and Lakshmanan~\cite{LL16} study the problem of finding a maximum $\gamma$-quasi-clique containing a given subset of vertices $S$ of the original graph, and propose a heuristic algorithm. Recently, Pastukhov et al.~\cite{PV+18} study the maximum degree-based $\gamma$-quasi-clique problem. First they prove that finding a maximum $\gamma$-quasi-clique is an NP-Hard problem, and present algorithms for a $\gamma$-quasi-clique of maximum cardinality. Note that while this work focuses on finding a single quasi-clique of the largest size, our goal is not just to find a single large quasi-clique, but to enumerate the $k$ largest maximal quasi-cliques. Further, the NP-hardness result in~\cite{PV+18} is for finding the maximum $\gamma$-quasi-clique, while our NP-hardness result is for finding if a quasi-clique is maximal. 


Abello et al.~\cite{ARS02} first study the problem of finding a density-based $\delta$-quasi-clique, defined as a subgraph $Q$ of the original graph with the ratio of the edges in $Q$ to the total number of edges in a complete subgraph of size $Q$ is at least $\delta$. Note that a degree-based quasi-clique is also a density-based quasi-clique, but the converse is not true. They propose a heuristic algorithm for finding a large $\delta$-quasi-clique. Uno~\cite{uno2007efficient} considered density-based quasi-cliques, and proposed an algorithm for enumerating all $\delta$-quasi-cliques, with polynomial delay. In another study, Pattillo et al.~\cite{PV+13} prove that deciding whether there exists a $\delta$-quasi-clique of size at least $\theta$ is an NP-Complete problem. Brunato et al.~\cite{BHB07} defines a $(\gamma,\delta)$-quasi-clique combining the minimum degree requirement of degree based $\gamma$-quasi-clique and minimum edge requirement of density based $\delta$-quasi-clique. They propose a heuristic algorithm for finding a maximum $(\gamma,\delta)$-quasi-clique. Recently, Balister et al.~\cite{BB+18} derive the concentration bound on the size of the density based maximum $\delta$-quasi-clique following the work of Veremyev et al.~\cite{VB+12}.

\textbf{Other Works on Dense Subgraphs: }The study of dense subgraphs has attracted a wide spectrum of research for many decades. There have been many works on complete dense subgraphs such as maximal cliques~\cite{BK73,CN85,MU04,TTT06,ELS10,CG+16,DST16,MXT17}, maximal bicliques~\cite{AA+04,LSL06,MT17}. There are many different types of incomplete dense subgraphs other than quasi-clique such as $k$-core~\cite{BZ03,CK+11,DDZ14,KB+15}, $k$-truss~\cite{C08} etc. A $k$-core is a maximal connected subgraph such that each vertex in that subgraph has degree at least $k$ and this subgraph is quite different from quasi-clique in the sense that the degree threshold in the $k$-core is an absolute threshold whereas the threshold in the quasi-clique (either degree threshold or density threshold) are relative thresholds, equal to a certain factor ($\gamma$) times the size of the subgraph. In a $k$-truss subgraph, each edge is contained in at least $(k-2)$ triangles. $k$-truss is different from the quasi-clique because the threshold in the $k$-truss is an absolute threshold. Other works on dense subgraphs different from quasi-clique subgraph includes densest subgraph~\cite{G84,C00,KS09,AC09}, triangle densest subgraph~\cite{T14}, $k$-clique densest subgraph~\cite{T15,MP+15} etc. Similar to $k$-core, these subgraphs are based on an absolute threshold for the degree, rather than a relative threshold. Prior work on top-$k$ dense subgraph discovery includes the work of Zou et al.~\cite{ZL+10} on enumeration of top-$k$ maximal cliques from an uncertain graph, defined as the set of cliques with the $k$ largest clique probabilities and the works of Balalau et al.~\cite{BD+15} and Galbrun et al.~\cite{GGT16} on the enumeration of top-$k$ densest subgraphs.

\vspace{-1ex}
\section{Preliminaries and Problem Definition}
\label{sec:prelims}
Let $G=(V,E)$ be a simple undirected graph. 
Let $V(G)$ denote the set of vertices and $E(G)$ denote the set of edges of $G$. 
Let $d^{G}(u)$ denote the degree of vertex $u$ in $G$. When the context is clear, we use $d(u)$ to mean $d^{G}(u)$.
We use the following definition of degree-based quasi-cliques. 

\begin{definition}[$\gamma$-quasi-clique]
For parameter $0 < \gamma \leq 1$, a vertex-induced subgraph $Q$ of $G$ is called a $\gamma$-quasi-clique if $Q$ is connected and, for every vertex $v\in V(Q)$, $d^{Q}(v) \geq \lceil\gamma(|Q|-1)\rceil$.
\end{definition}

Note that when $\gamma=1$, the above definition reduces to a clique. For a $\gamma$-quasi-clique $Q$, by the phrase ``size of $Q$'' and notation $|Q|$, we mean the number of vertices in $Q$. A $\gamma$-quasi-clique $Q$ is called {\em maximal} if there does not exist another $\gamma$-quasi clique $Q'$ such that $V(Q) \subset V(Q')$ and $|Q| < |Q'|$. See Figure~\ref{fig:qc-example} for an example of the above definition.

\begin{problem}[Top-$k$ $\gamma$-QCE]
Given integer $k > 0$, a parameter $0 < \gamma \le 1$, a simple undirected graph $G=(V,E)$, enumerate $k$ maximal $\gamma$-quasi cliques from $G$ that have the largest sizes, among all maximal $\gamma$-quasi-cliques in $G$.
\end{problem}

Given $0 < \gamma \le 1$, a simple undirected graph $G=(V,E)$, the $\gamma$-quasi-clique enumeration ($\gamma$-QCE) problem asks to enumerate all maximal $\gamma$-quasi cliques from $G$. If the value of $\gamma$ is clear from the context we sometimes use ``QCE'' to mean $\gamma$-QCE. 

\textbf{\quick algorithm for QCE:} The current state-of-the-art algorithm for QCE is \quick~\cite{LW08}, which takes as input a set of vertices $X$, degree threshold $\gamma$, and enumerates all maximal $\gamma$-quasi-cliques that contain $X$. By setting $X$ to an empty set, one can enumerate all maximal $\gamma$-quasi-cliques of $G$. Note that \quick may also enumerate non-maximal quasi-cliques which need to be filtered out in a post-processing step. We modify \quick such that it omits the check for maximality in emitting quasi-cliques i.e. it enumerates all $\gamma$-quasi-cliques instead of only maximal ones -- we call this version of the \quick algorithm as \quickm. Non-maximal quasi-cliques are filtered out at a later step, while enumerating top-$k$-quasi-cliques.

\remove{
	\begin{table}[!t]
		\small
		\centering
		\caption{\textbf{Summary of Notations}}
		\label{table:notations}
		\begin{tabular}{| c | c |}
			
			\toprule
			\textbf{Notation} & \textbf{Meaning} \\
			\midrule
			$G$ & Input Graph\\
			\hline
			$V(G)$ &  Vertex set of $G$\\
			\hline
			$E(G)$ & Edge set of $G$\\
			\hline
			$\Gamma_{G}(u)$ & Set of vertices adjacent to $u$ in $G$\\
			\hline
			$d^{G}(u)$ & Number of vertices adjacent to $u$ in $G$\\
			\hline
			$|G|$ & Number of vertices in $G$, also called size of $G$\\
			\bottomrule
		\end{tabular}
	\end{table}
}

\vspace{-1ex}
\newcommand{\NPsz}{\ensuremath{r}}
\section{Hardness of Checking Maximality of a Quasi-Clique}
\label{sec:hardness}
It is easy to deduce that $\gamma$-QCE is an NP-hard problem, since the problem of enumerating maximal cliques is a special case when $\gamma=1$. However, QCE presents an even more severe challenge. We now prove that even determining if a given quasi-clique is maximal is an NP-hard problem. This is very different from the case of maximal cliques -- checking a given clique is maximal can be done in polynomial time, by simply checking if there exists a vertex outside the clique that is connected to all vertices within the clique. If there exists such a vertex, then the given clique is not maximal, otherwise it is maximal.

\begin{problem}[Maximality of a Quasi-Clique]
Given a graph $G=(V,E)$, a $\gamma$-quasi-clique $X \subseteq V$, determine whether or not $X$ is a maximal quasi-clique in $G$.
\end{problem} 

\begin{theorem}
\label{thm:mqc-hardness}
Maximality of a Quasi-clique is NP-hard. 
\end{theorem}

\begin{IEEEproof}
We prove NP-hardness by reducing the $r$-clique problem, that asks whether a given graph $G'=(V',E')$ contains a clique of size $r$, to the problem of checking maximality of a quasi-clique. This $r$-clique problem is NP-complete~\cite{GJ02}. Given graph $G'$ on which we have to solve the $r$-clique problem, construct a graph $G=(V,E)$ as follows (See Figure~\ref{fig:NP}). Let $V = V'\cup X$ where $X$ is a set of $2r^2 + r$ additional vertices.  $X$ consists of three parts -- two sets $A_1$ and $A_2$, each of size $r^2$, and $B$, of size $r$. We construct edges in $G$ as follows:
\begin{itemize}
\item All edges $E'$ in $G'$ are retained in $G$
\item Add edges within $A_1$, within $A_2$ and within $B$ such that $A_1$ is a clique, $A_2$ is a clique, and $B$ is a clique.
\item Add edges connecting each vertex in $A_1$ with each vertex in $A_2$ and $B$. Thus, $A_1 \cup A_2$ is a clique and $A_1 \cup B$ is a clique, but $X=A_1 \cup A_2 \cup B$ is not a clique.
\item Add edges connecting each vertex of $A_1$ to each vertex of $V'$.
\end{itemize}
\begin{figure}[t!]
	\captionsetup[subfigure]{justification=centering}
	\centering
	\includegraphics[width=0.45\textwidth] {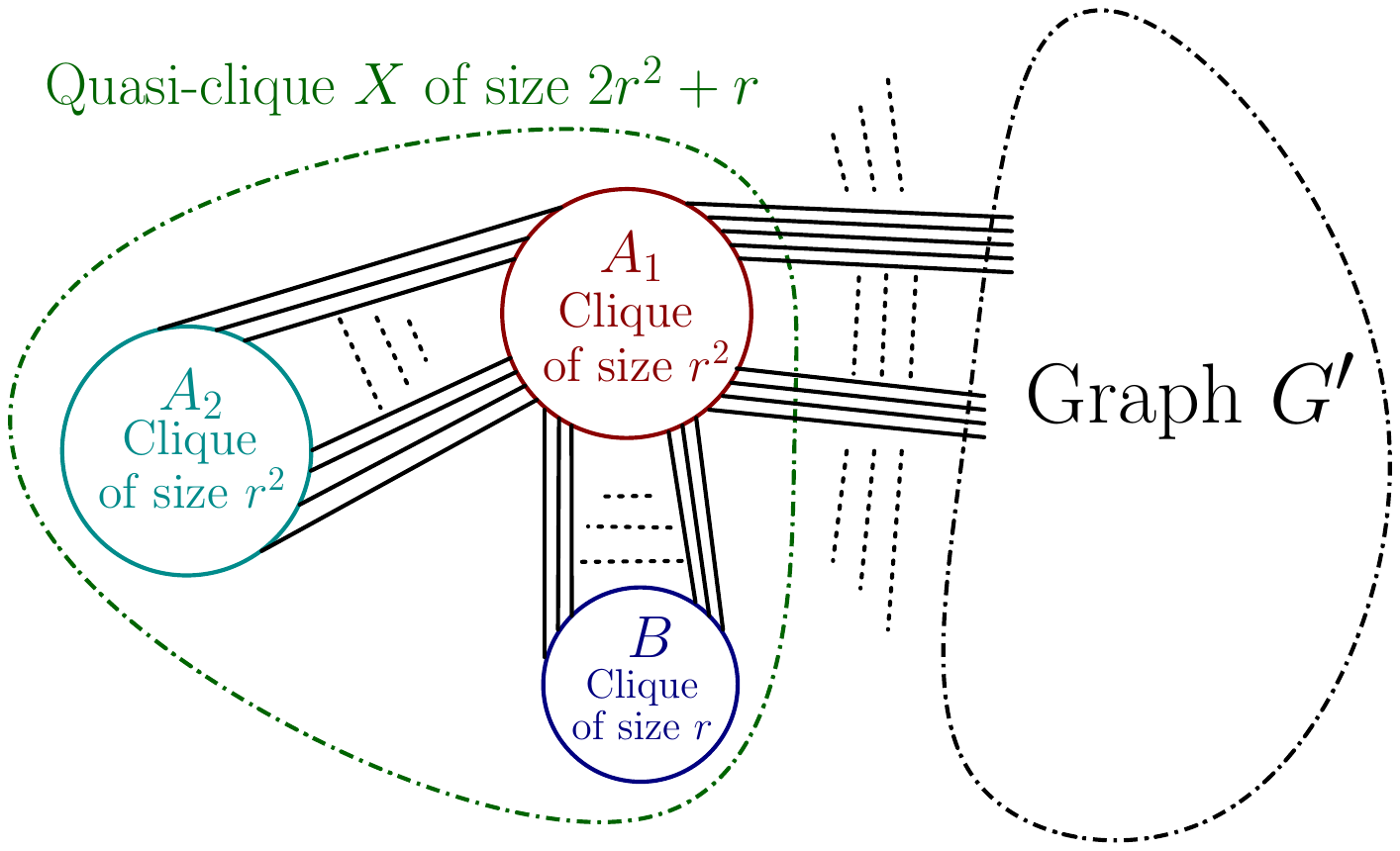}
	\caption{Construction of graph $G$ for Proof of Theorem~\ref{thm:mqc-hardness}.}
	\label{fig:NP}
\end{figure}

Set $\gamma= \frac{r^2 + r - 1}{2r^2 + 2r - 1}$. 
We first show that $X$ is a $\gamma$-quasi-clique. To see this, consider that the total number of vertices in $X$ is $(2r^2 + r)$.
For $X$ to be a $\gamma$-quasi-clique, each vertex should have a degree of at least $\lceil \gamma \cdot (2r^2 + r-1) \rceil = \left\lceil \frac{(r^2 + r - 1)(2r^2+r-1)}{2r^2 + 2r - 1} \right\rceil \le (r^2+r-1)$.
We can verify that every vertex in $X$ has at least a degree of $(r^2 + r - 1)$.

We now claim that $X$ is not a maximal $\gamma$-quasi-clique in $G$ if and only if $G'$ contains an $r$-clique.\\
(1)~Suppose that $G'$ contains an $r$-clique. There exists a set of vertices $L \subset V'$ such that $L$ is a clique and $|L| = r$.
Consider the set $Q = X \cup L$. We show that $Q$ is a $\gamma$-quasi-clique. Since $X \cap L = \emptyset$, we have $|Q| = |X| + |L| = 2r^2 + 2r$. Therefore, $\lceil{\gamma \cdot (|Q| - 1)}\rceil = \left\lceil {\frac{(r^2 + r - 1)(2r^2+2r-1)}{2r^2 + 2r - 1}} \right\rceil = r^2 + r - 1$. It can be verified that every vertex in $Q$ has at least a degree of $r^2 + r - 1$. Thus, $Q$ is a $\gamma$-quasi-clique.

(2)~Suppose that $X$ is not a maximal $\gamma$-quasi-clique in $G$.
Then, there must be a non-empty set $M \subset V'$ such that $R = M \cup X$ is a $\gamma$-quasi-clique in $G$. 
We note that it is not possible that $|M| > r$. If this was the case, then the minimum degree threshold for a vertex in $R$ is $\lceil{\gamma \cdot (|R| - 1)}\rceil = \left\lceil {\frac{(r^2 + r - 1)(|R|-1)}{2r^2 + 2r - 1}} \right\rceil = \left\lceil {\frac{(r^2 + r - 1)(2r^2 + r + |M| -1)}{2r^2 + 2r - 1}} \right\rceil > r^2 + r - 1$, since $|M| > r$. However, the minimum degree of vertices in $R$ is $r^2 + r - 1$ (consider a vertex from the set $B \subset R$). 

Similarly, it is not possible that $|M| < r$. Let assume that was the case. Then, the minimum degree threshold for a vertex in $R$ is $\lceil{\gamma \cdot (|R| - 1)}\rceil = \left\lceil {\gamma \cdot (2r^2 + r + |M| - 1)} \right\rceil$. However, the minimum degree of a vertex in $R$ is $(r^2 + |M| - 1)$ (consider a vertex from $M$). It can be verified that since $|M| < r$,  $ \left\lceil {\gamma \cdot (2r^2 + r + |M| -1)} \right\rceil > r^2 + |M| - 1$. Then, $R$ cannot be a quasi-clique. 

Therefore, it must be that $|M| = r$. In this case, $M$ must be a clique of size $r$. In the case that $M$ is not a clique, the minimum degree of a vertex in $R$ is $r^2 + r - 2$ (consider a vertex from $M$). However, the minimum degree threshold for a vertex in $R$ is $\lceil{\gamma \cdot (|R| - 1)}\rceil = \left\lceil {\frac{(r^2 + r - 1)(2r^2+2r-1)}{2r^2 + 2r - 1}} \right\rceil = r^2 + r - 1 > r^2 + r - 2$. This completes the proof.
\end{IEEEproof}
\vspace{-1ex}
\begin{figure*}[t!]
	\captionsetup[subfigure]{justification=centering}
	\centering
	\subfloat[$\gamma'=0.85$]{%
		\includegraphics[width=.25\textwidth] {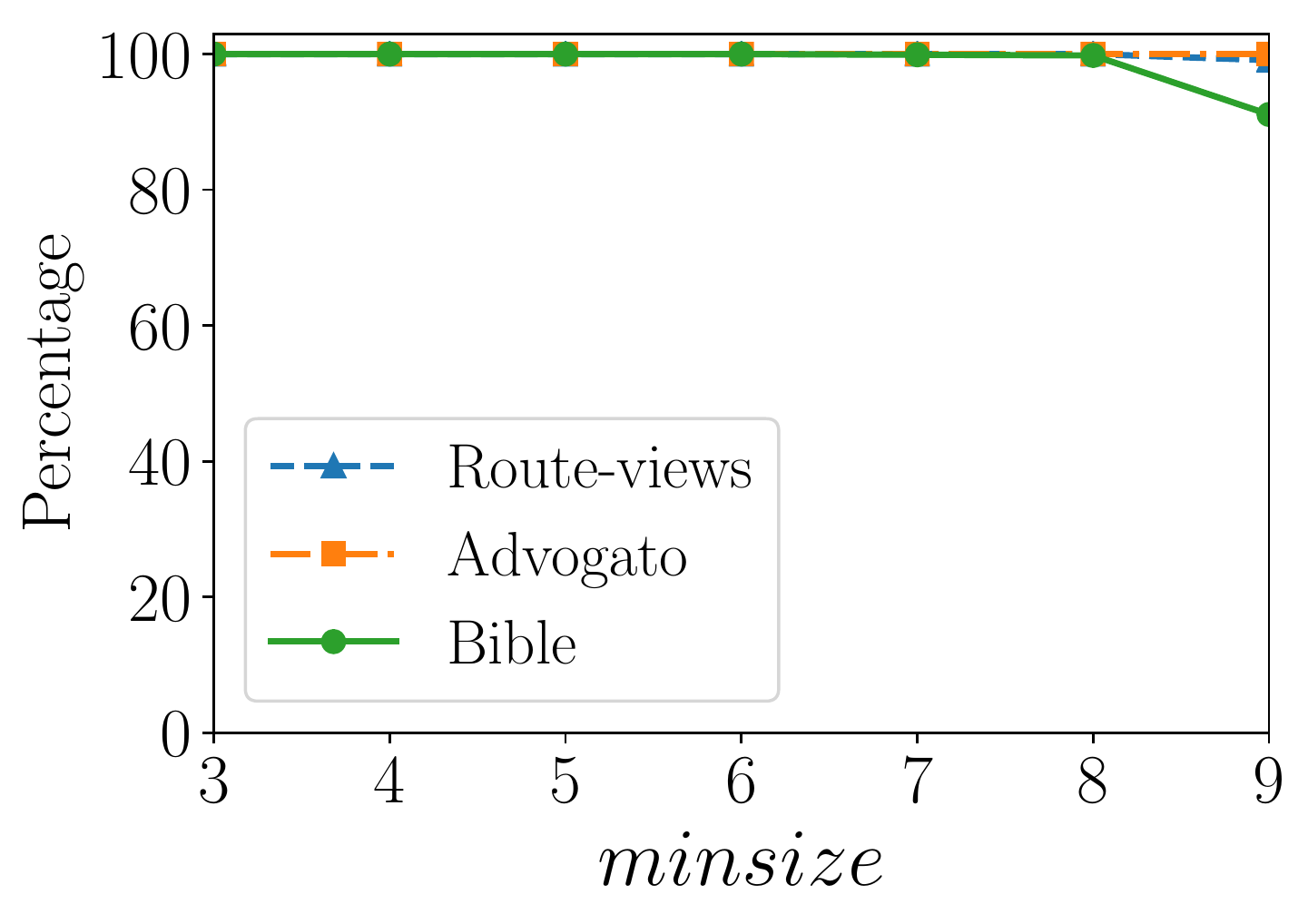}}
	\subfloat[$\gamma'=0.9$]{%
		\includegraphics[width=.25\textwidth] {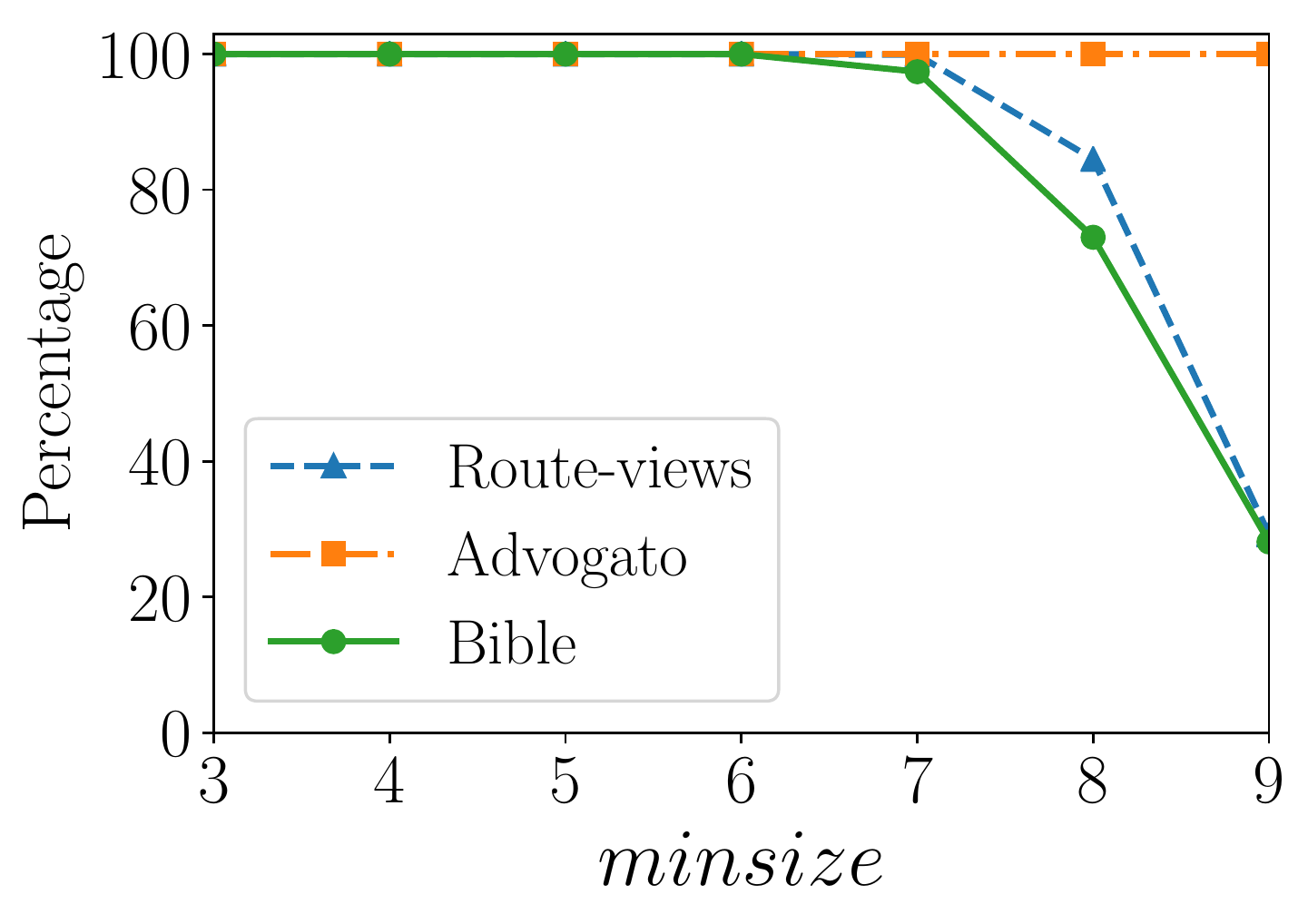}}
	\subfloat[$\gamma'=0.95$]{%
		\includegraphics[width=.25\textwidth] {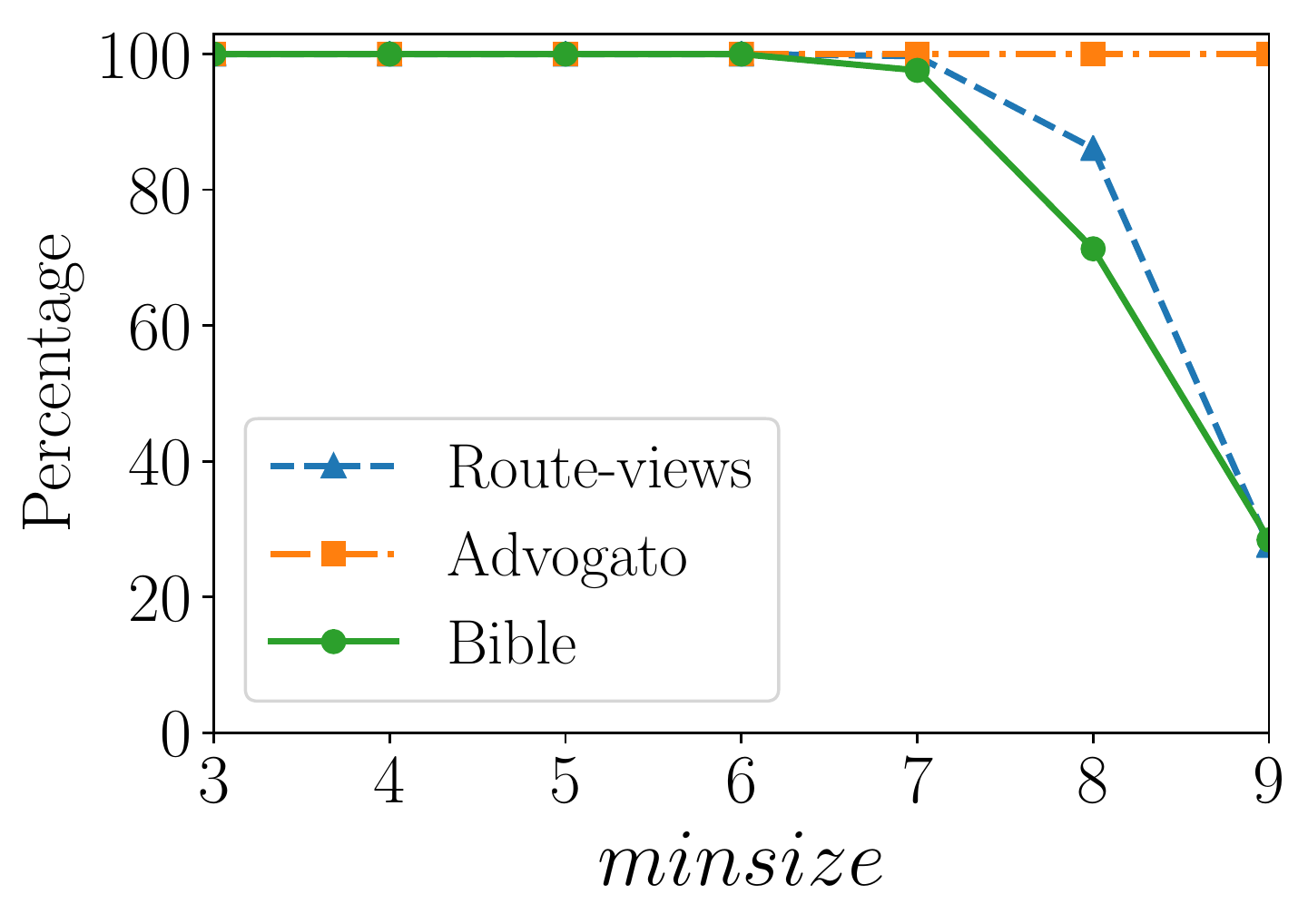}}
	\subfloat[$\gamma'=1.0$]{%
		\includegraphics[width=.25\textwidth] {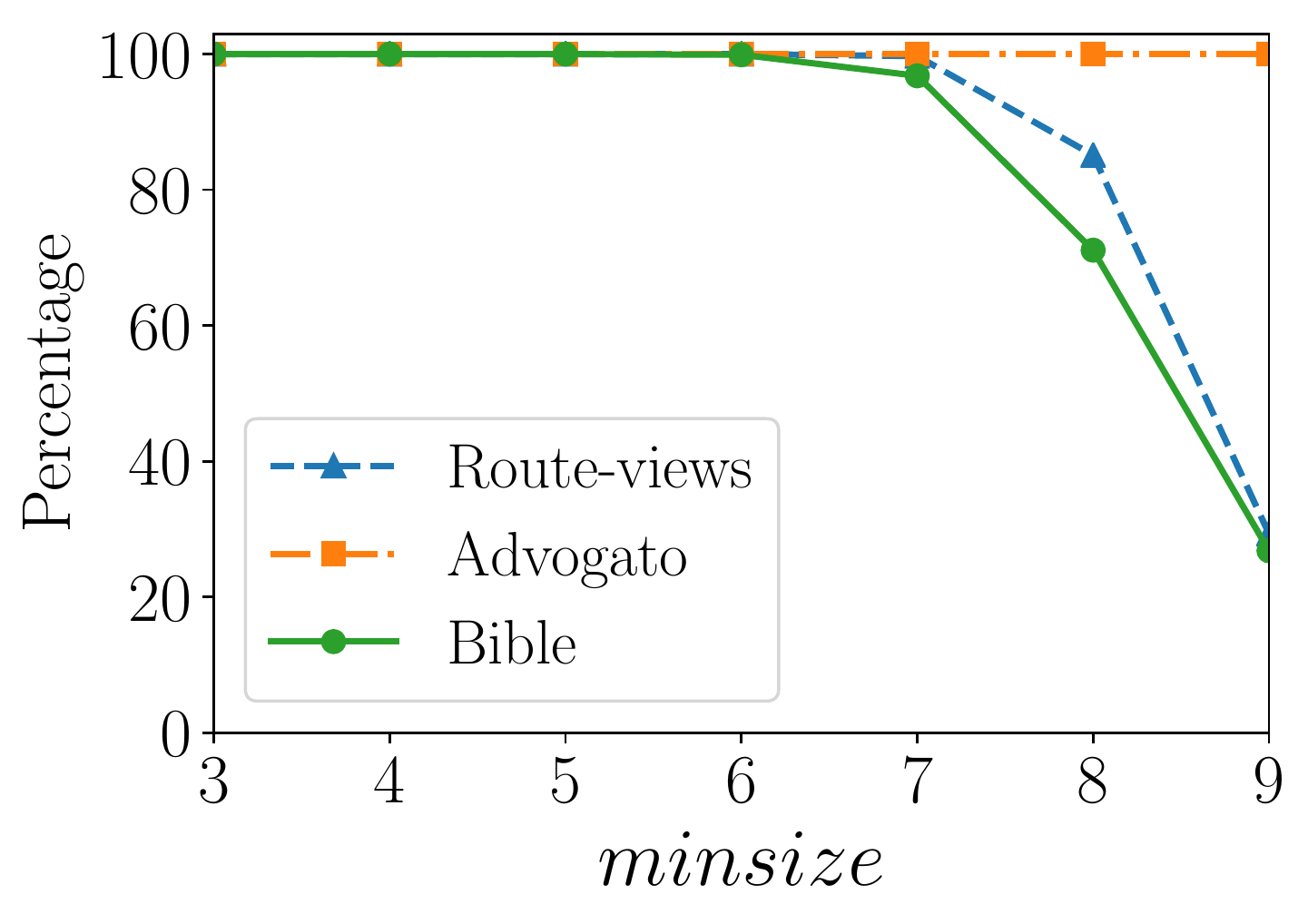}}
	\caption{On the $x$-axis is the size of a $\gamma'$-quasi-clique. The $y$-axis shows what fraction of the sampled $0.8$-quasi-cliques contained a $\gamma'$-quasi-clique of a given size. The results are shown for three different graphs, \advo, \route, and \bible.	
		\label{fig:kernel}}	
	\vspace{-3ex}
\end{figure*}
\section{Algorithm for Top-$k$ QCE}
\label{sec:algorithm}
In this section, we present algorithms for enumerating top-$k$ $\gamma$-quasi-cliques given input graph $G$ parameters $k$ and $\gamma$. 

A straightforward {\bf baseline} algorithm for \topKqce~ is to enumerate all maximal quasi-cliques using the \quick algorithm \cite{LW08}, and then only output the $k$ largest among them. We call this algorithm as \pl. \\



{\bf \kqc Algorithm:} We next present our heuristic algorithm for \topKqce, \kqc. The intuition is as follows. Our observations from experiments on a range of graphs showed that within a dense subgraph, a $\gamma$-quasi-clique for a given $\gamma$, there is usually a smaller, but denser subgraph, i.e. a $\gamma'$-quasi-clique with $\gamma' > \gamma$. We describe some results below.

We considered graphs \advo, \route, and \bible\footnote{These graphs are described in \cref{table:graph-stats}}, and $\gamma=0.8$. We randomly sampled $1000$ $\gamma$-quasi-cliques each from the graphs \advo, \route, and \bible, of size at least $10$ (\mnsz{} $ = 10$)\footnote{We did not consider the graph \slsh{} because the the size of largest \qc{} in this graph is less than $10$.}. Interestingly, we found that {\bf every sampled $0.8$-quasi-clique from  \advo, \route, and \bible~ graphs had, as a subgraph, a $\gamma'$-quasi-clique of size at least $7$, for different values of $\gamma'$, ranging from $0.85$ to $1.0$}. Details are shown in the \cref{fig:kernel}. 

The above suggests that large $\gamma$-quasi-cliques (usually) contain $\gamma'$-quasi-cliques of substantial sizes as subgraphs for $\gamma' > \gamma$. Note that an adversary can form a $\gamma$-quasi-clique without any large  $\gamma'$-quasi-clique contained within. However, our experiments show that this is not the case in real-world networks, and the size of $\gamma'$-quasi-cliques (or ``kernels'') in  $\gamma$-quasi-cliques is relatively large (See \Cref{fig:kernel}).

\begin{figure}[t!]
	\captionsetup[subfigure]{justification=centering}
	\centering
	\includegraphics[width=.3\textwidth] {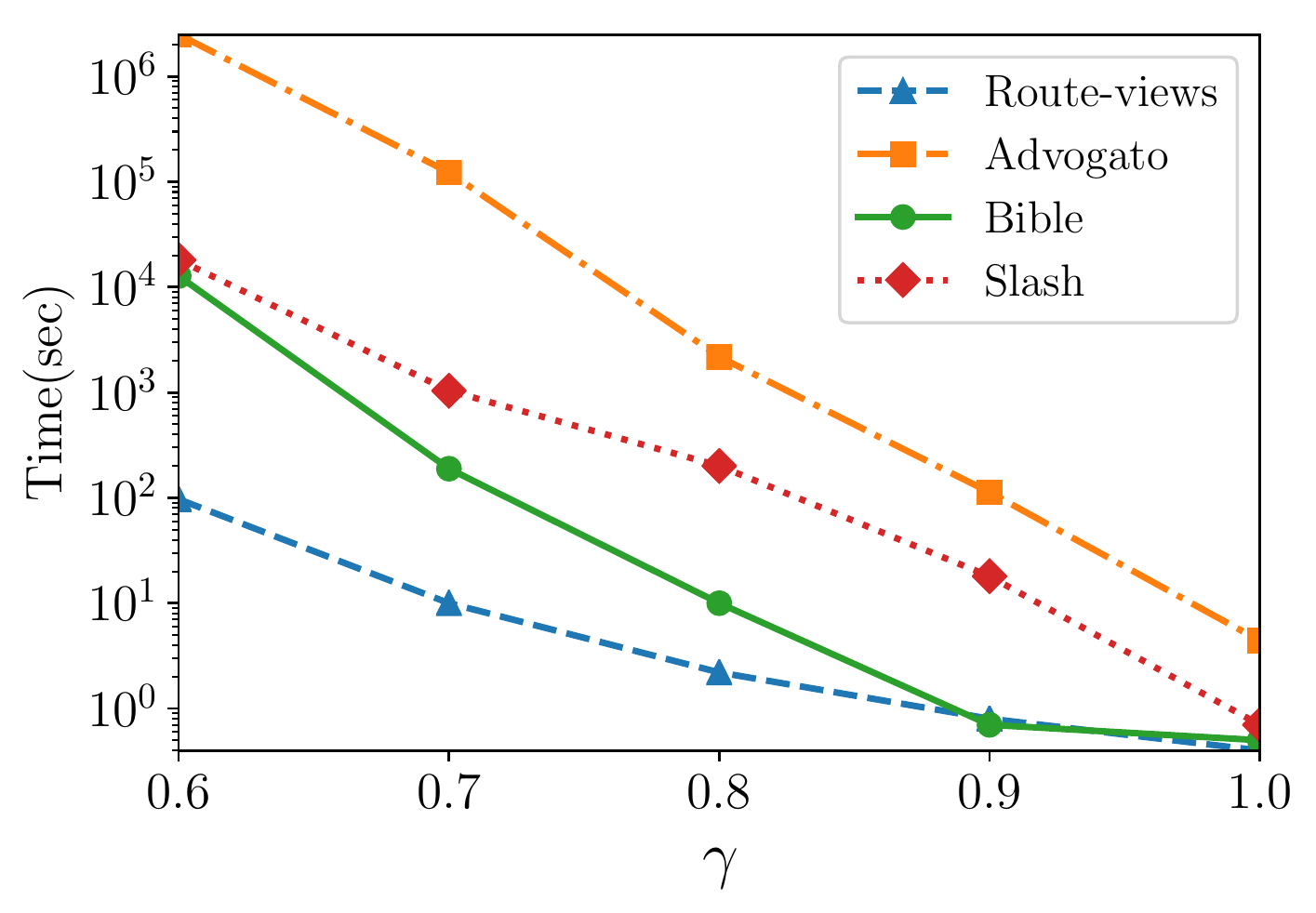}
	\caption{Runtimes of $\gamma$-quasi-clique enumeration for different values of $\gamma$, for \mnsz $= 5$, and for different graphs. }
	\label{fig:time-vs-gamma}
\end{figure}
We further note that as $\gamma$ increases, the complexity of finding $\gamma$-quasi-cliques decreases substantially. To see this, \Cref{fig:time-vs-gamma} shows the computational cost of enumerating $\gamma$-quasi-cliques as $\gamma$ increases. Note that the $y$-axis is in log-scale. The trend is that the cost decreases exponentially as $\gamma$ increases.

Based on the above observations, our algorithm idea is as follows. Given a threshold $0 < \gamma < 1$, we choose $\gamma'$ such that $\gamma < \gamma' \le 1$. We then enumerate the set $Y$ consisting of the largest $k'$ maximal $\gamma'$-quasi-cliques in the graph $G$. These dense subgraphs in $Y$ are considered ``kernels'' that are then further expanded to recover $k$ maximal $\gamma$-quasi-cliques in $G$. Thus, our algorithm has two parts:
\begin{itemize}
\item [(1)]
{\em Kernel Detection:} Find kernels in the graph, i.e. $\gamma'$-quasi-cliques for the chosen value of $\gamma'$. Then, among all kernels, largest $k'$ maximal kernels are extracted.

\item [(2)]
{\em Kernel Expansion:} Expand detected kernels into larger $\gamma$-quasi-cliques. This can be performed by iterating through the enumerated $\gamma'$-quasi-cliques and then using an existing algorithm for QCE, such as \quick \cite{LW08} to enumerate all maximal quasi-cliques that contain each kernel. Next, among all extracted $\gamma$-quasi-cliques, largest-$k$ maximal $\gamma$-quasi-cliques are enumerated.
\end{itemize}

The \kqc algorithm is described in Algorithm \ref{alg:kqc}. In \Cref{line:kqc-getX} of \Cref{alg:kqc}, a modified version of \quick{}, \quickm, is used to extract all $\gamma'$-quasi-cliques. This version does not actually check if a quasi-clique is maximal, before outputting it. For \kqc{}, the quasi-cliques from the subroutine need not be maximal, since \Cref{alg:kmax} sorts quasi-cliques in an ascending order of their sizes and suppresses non-maximal quasi-cliques. By omitting a maximality check, \quickm{} is more efficient than \quick{}.

\begin{lemma}
	\Cref{alg:kmax} returns at most $k$ largest quasi-cliques that are maximal with respect to a set of quasi-cliques $S$ (i.e. not contained within any other quasi-clique in $S$).
	\label{lem:kmax}
\end{lemma}

\begin{IEEEproof}
	A quasi-clique $q$ is added to the set $Q$ in \Cref{line:kmax-added} of \Cref{alg:kmax} if the size of $Q$ is less than $k$ and $q$ is not a subset of any other quasi-clique in $Q$ (Using the {\small \textbf{\texttt{if}}} block in \Cref{line:kmax-if}). Because of the latter condition, all quasi-cliques in $Q$ are maximal with respect to the quasi-cliques of $S$. On the other hand, since all $\gamma$-quasi-cliques in $S$ are sorted in an ascending order of their sizes (\Cref{line:kmax-sort}), $Q$ maintains the largest maximal quasi-cliques from $S$, and the size of $Q$ cannot be larger than $k$. 
\end{IEEEproof}

With \cref{lem:kmax,lem:kqc}, we show that every $\gamma$-quasi-clique, returned by \Cref{alg:kqc}, is maximal and has at least \mnsz vertices. 

\begin{algorithm}[b!]
	\caption{\kqc($G,\gamma,\mnsz,k$)
		\label{alg:kqc}
	}
	\DontPrintSemicolon
	\SetKwInOut{Input}{Input}
	\KwIn{Graph $G = (V, E)$, parameter $0 < \gamma < 1$, size threshold $\mnsz$, and an integer $k$.}
	\KwOut{$k$ maximal $\gamma$-quasi-cliques in $G$ with at least $\mnsz$ vertices in each.}
	
	Choose $\gamma'$ such that $\gamma < \gamma' \le 1$, and $k' \geq k$\label{line:kqc-choose}\;
	$X \gets \quickm(G, \phi, \gamma', \mnsz)$ \Comment{ \clrb{Kernel Detection -- retrieve $\gamma'$-quasi-cliques from $G$.}} \label{line:kqc-getX}\;
	$Y \leftarrow \kmax(X, k')$ \Comment{\clrb{\Cref{alg:kmax}}.}\label{line:kqc-kmaxX}\;
	$Z \gets \emptyset$\;
	\For{$\text{a quasi-clique } q \in Y$ \label{line:kqc-forY}}
	{
		$Z \gets Z \cup \quickm(G, q, \gamma, \mnsz)$ \Comment{ \clrb{Kernel Expansion -- add $\gamma$ quasi-cliques through expanding $q$.}} \label{line:kqc-addZ}\;
	}
	$R \gets \kmax(Z, k)$ \label{line:kqc-kmaxZ} \Comment{\clrb{\Cref{alg:kmax}.}}\;
	
	\Return $R$ \label{line:kqc-return}\;
\end{algorithm}
\begin{lemma}
	The set $R$ in \Cref{alg:kqc} contains at most $k$ $\gamma$-quasi-cliques, where each quasi-clique has at least \mnsz vertices, and is a maximal quasi-clique in the graph. 
	\label{lem:kqc}
\end{lemma}

\begin{IEEEproof}
	In \Cref{line:kqc-getX} of \Cref{alg:kqc}, \quickm extracts all $\gamma'$-quasi-cliques in the graph $G$. Then in \Cref{line:kqc-kmaxX} and by \Cref{lem:kmax}, $X$ contains largest maximal $\gamma'$-quasi-cliques of $G$, where $|X| \leq k'$. 
	In \Cref{line:kqc-forY,line:kqc-addZ}, every $\gamma$-quasi-clique of the set $Z$ contains at least a $\gamma'$-quasi-clique from the set $Y$. It is because we expand every $\gamma'$-quasi-clique in $Y$ by \quickm using the parameter $\gamma$. Therefore, any quasi-clique in the set $Z$ has at least \mnsz vertices. Since $R \subseteq Z$ (\Cref{line:kqc-kmaxZ}), any quasi-clique in $R$ has at least \mnsz vertices. In addition, by \cref{lem:kmax}, we know that $|R| \leq k$.
	
	In the rest, we show that quasi-cliques in $R$ are maximal in $G$. By contradiction, assume that there is a $\gamma$-quasi-clique in $R$ which is \textit{not} maximal in $G$, $i.e.$ suppose that there are two $\gamma$-quasi-cliques $h$ and $h'$ in $G$ s.t. $h \in R$, $h \subset h'$, and $h' \not\in R$. We know that $h$ is discovered by the expansion of a $\gamma'$-quasi-clique $q$ (\Cref{line:kqc-addZ}). Hence, $q \subseteq h$ and $q \subset h'$. On the other hand, \quickm ensures that all $\gamma$-quasi-cliques containing $q$ are enumerated. The enumerated quasi-cliques are added to $Z$ in \Cref{line:kqc-addZ}. Therefore, $h' \in Z$ because $h'$ is a $\gamma$-quasi-clique and contains $q$. In addition, since $h \subset h'$, \Cref{lem:kmax} ensures that $h' \in R$ and $h \not\in R$. This contradicts our assumption that $h \in R$ and $h' \not\in R$.
\end{IEEEproof}

\begin{algorithm}[t!]
	\caption{\kmax($S, k$)
		\label{alg:kmax}
	}
	\DontPrintSemicolon
	\SetKwInOut{Input}{Input}
	\KwIn{Set of quasi-cliques $S$ and an integer $k$.}
	\KwOut{top (largest)-$k$ maximal quasi-cliques from $S$.}
	Sort $S$ in an ascending order of sizes of quasi-cliques \label{line:kmax-sort}\;
	$Q \gets \emptyset$\;
	\For{$\text{a quasi-clique } q \in S$}
	{
		\If{$(|Q| < k) \wedge (\forall q' \in Q, q \not\subseteq q')$ \label{line:kmax-if}}{
			$Q \gets Q \cup q$ \label{line:kmax-added}
		}
	}
	\Return $Q$\;
\end{algorithm}

In \Cref{line:kqc-choose} of \Cref{alg:kqc}, we need to choose two user-defined parameters, $\gamma'$ and $k'$, based on the given values of $\gamma$ and $k$. Here, we discuss the influence of these parameters on the accuracy and runtime of \kqc{}.

\textbf{Dependence on $\gamma'$:} For a given $\gamma$, varying the value $\gamma' \in (\gamma, 1]$ has effect on the runtime of \kqc. Based on our observation in \Cref{fig:time-vs-gamma}, for a high value of $\gamma'$, the kernel detection phase of \kqc{} can extract kernels faster. However, these kernels have relatively smaller sizes, and the expansion phase will take a longer time. Conversely, if $\gamma'$ is small (close to $\gamma$), kernel detection takes more time to extract $\gamma'$-quasi-cliques while the kernel expansion phase requires less time as kernels have relatively larger sizes and less chance to be expanded.

\textbf{Dependence on $k'$:} In \Cref{alg:kqc}, $k'$ decides the number of kernels, which need to be extracted in the kernel detection and then expanded to mine $\gamma$-quasi-cliques. The higher the value of $k'$, the more the kernels are needed to be processed in \kqc{}. Then, the runtime of \kqc{} may increase with a higher value for $k'$. However, the chance to mine larger maximal $\gamma$-quasi-cliques also increases since more kernels need to be enlarged in the kernel expansion phase. Therefore, a higher value of $k'$ can increase both the runtime and the accuracy of \kqc{}. 

In Section~\ref{sec:expts}, we present an empirical sensitivity analysis of parameters $\mnsz$, $\gamma$, $\gamma'$, $k$, and $k'$ on the accuracy and runtime of the algorithm.


\vspace{-1ex}
\begin{figure*}[!t]
	\centering
	\includegraphics[width=0.3\textwidth]{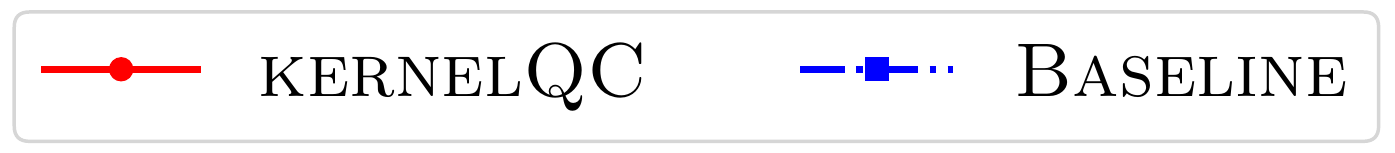}
	\vspace{-5ex}
\end{figure*}
\begin{figure*}[!t]
	\captionsetup[subfigure]{justification=centering}
	\centering
	\subfloat[\advo]{%
		\includegraphics[width=.24\textwidth] {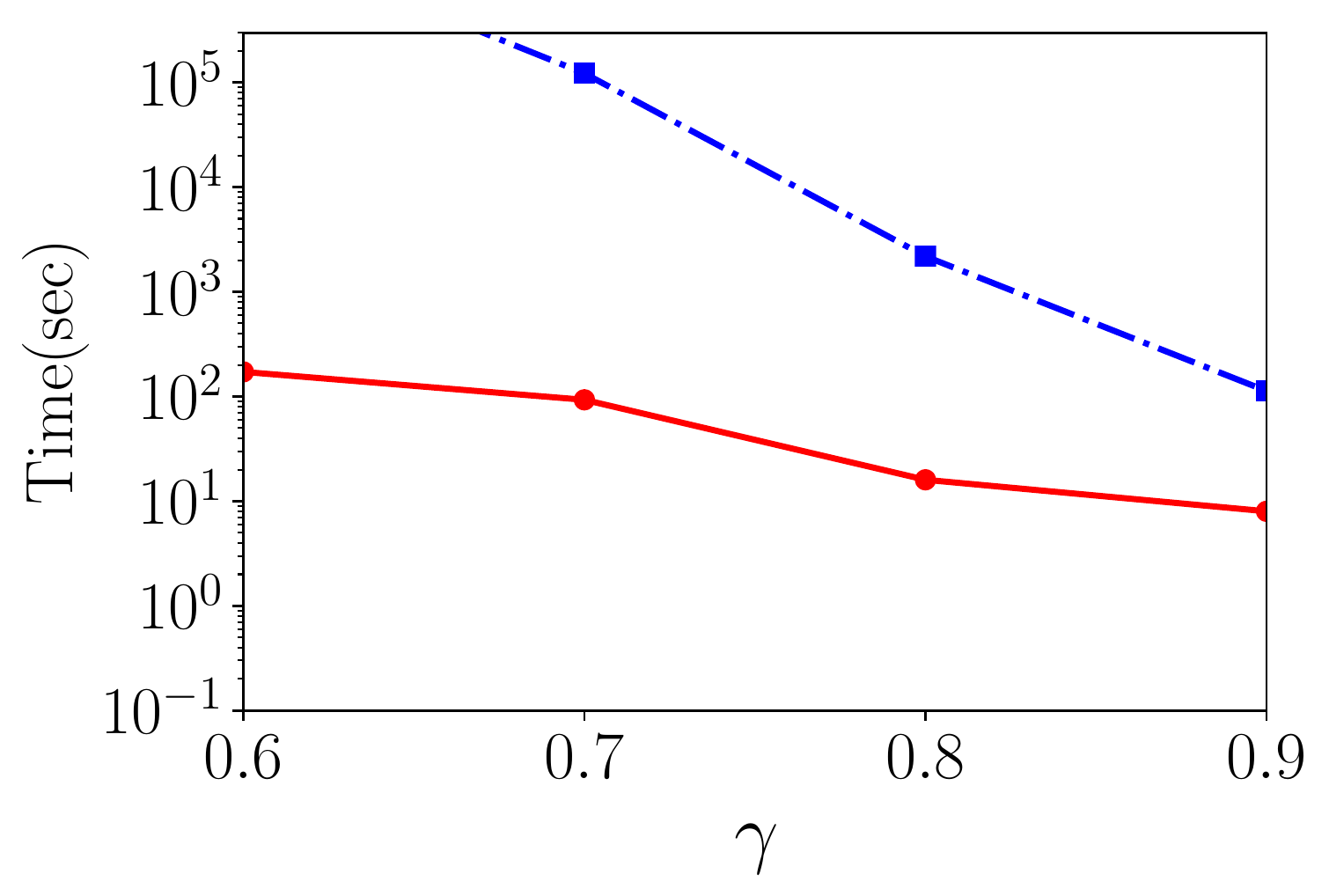}}
	\subfloat[\route]{%
		\includegraphics[width=.24\textwidth]{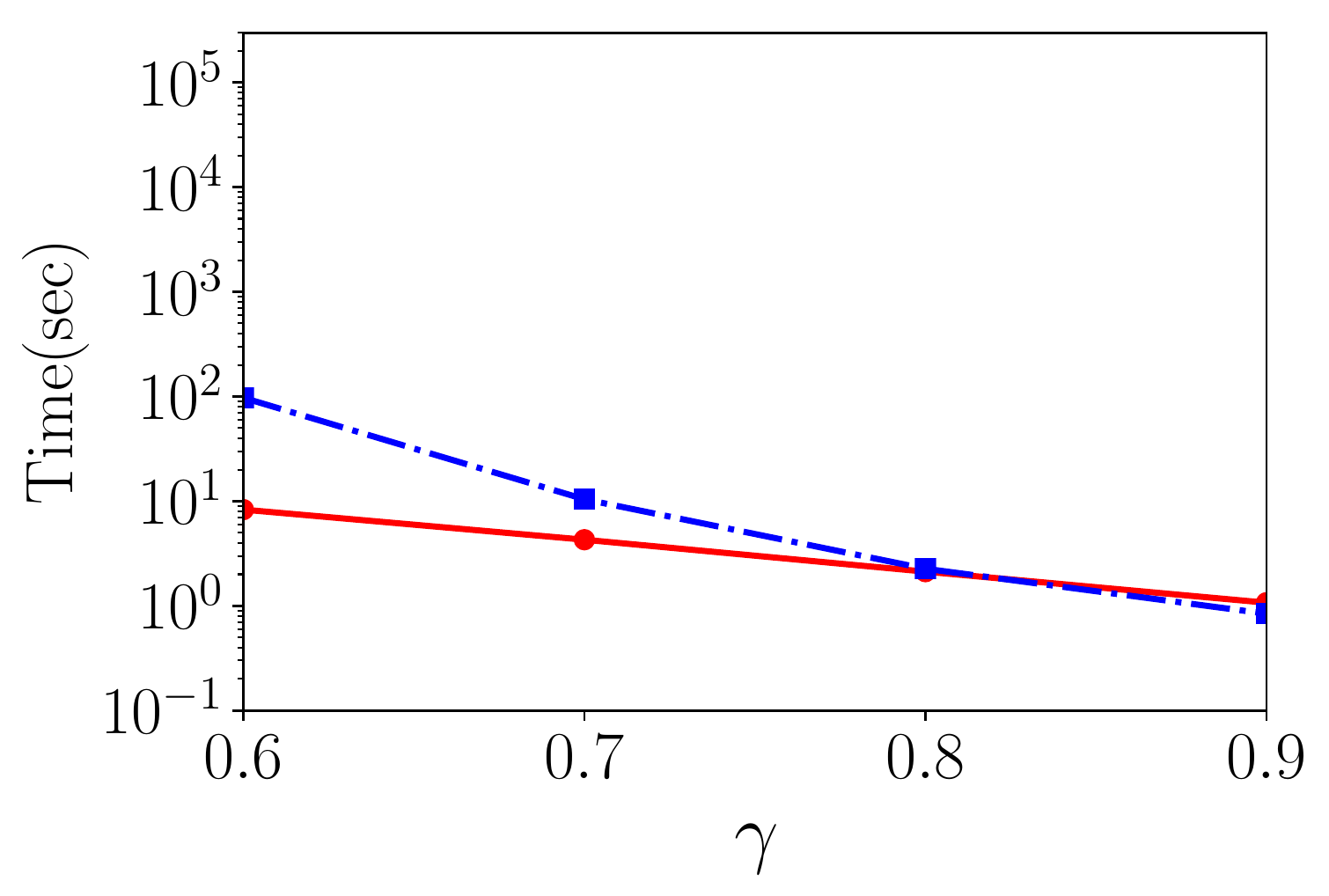}}
	\subfloat[\bible]{%
		\includegraphics[width=.24\textwidth] {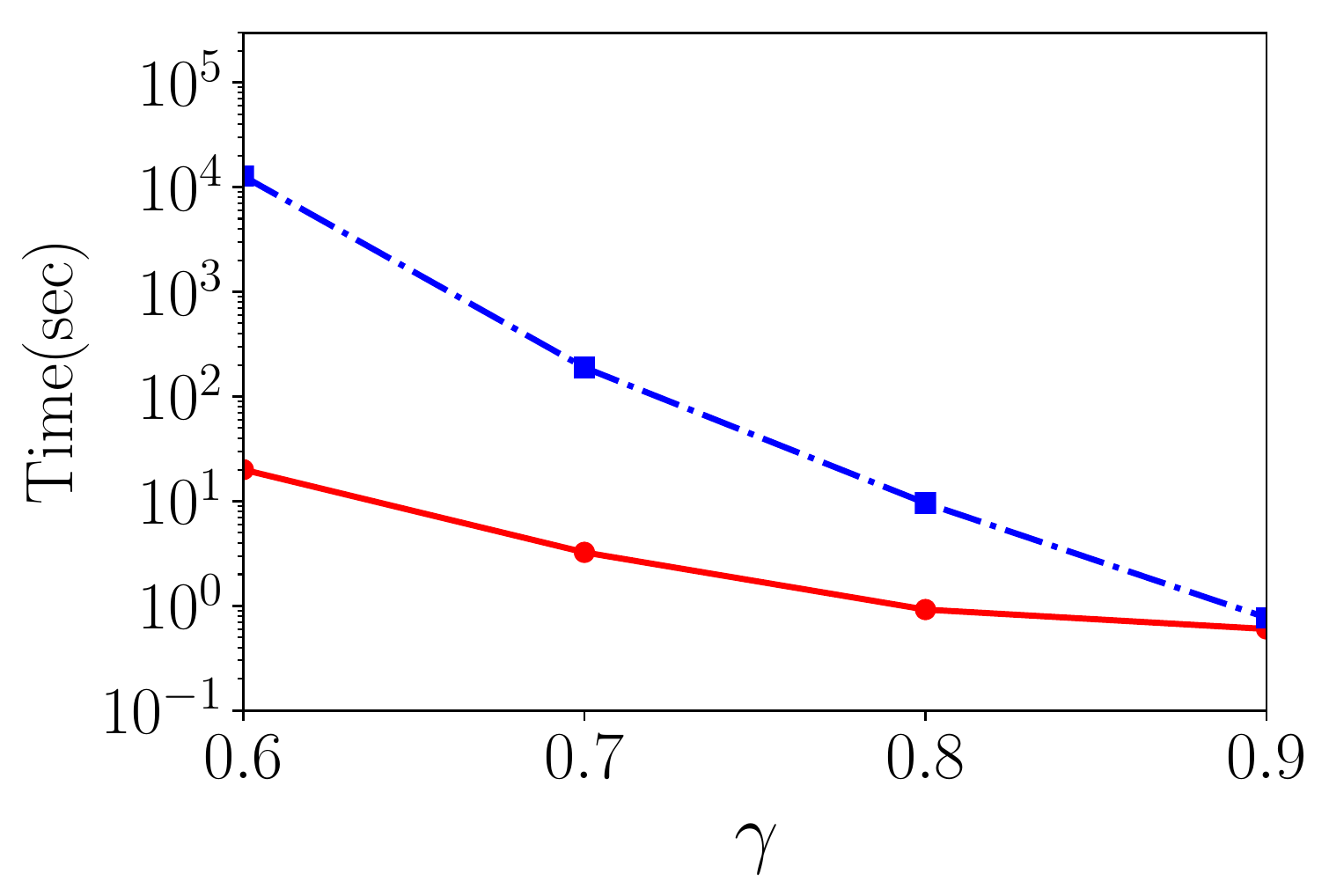}}
	\subfloat[\slsh]{%
		\includegraphics[width=.24\textwidth]{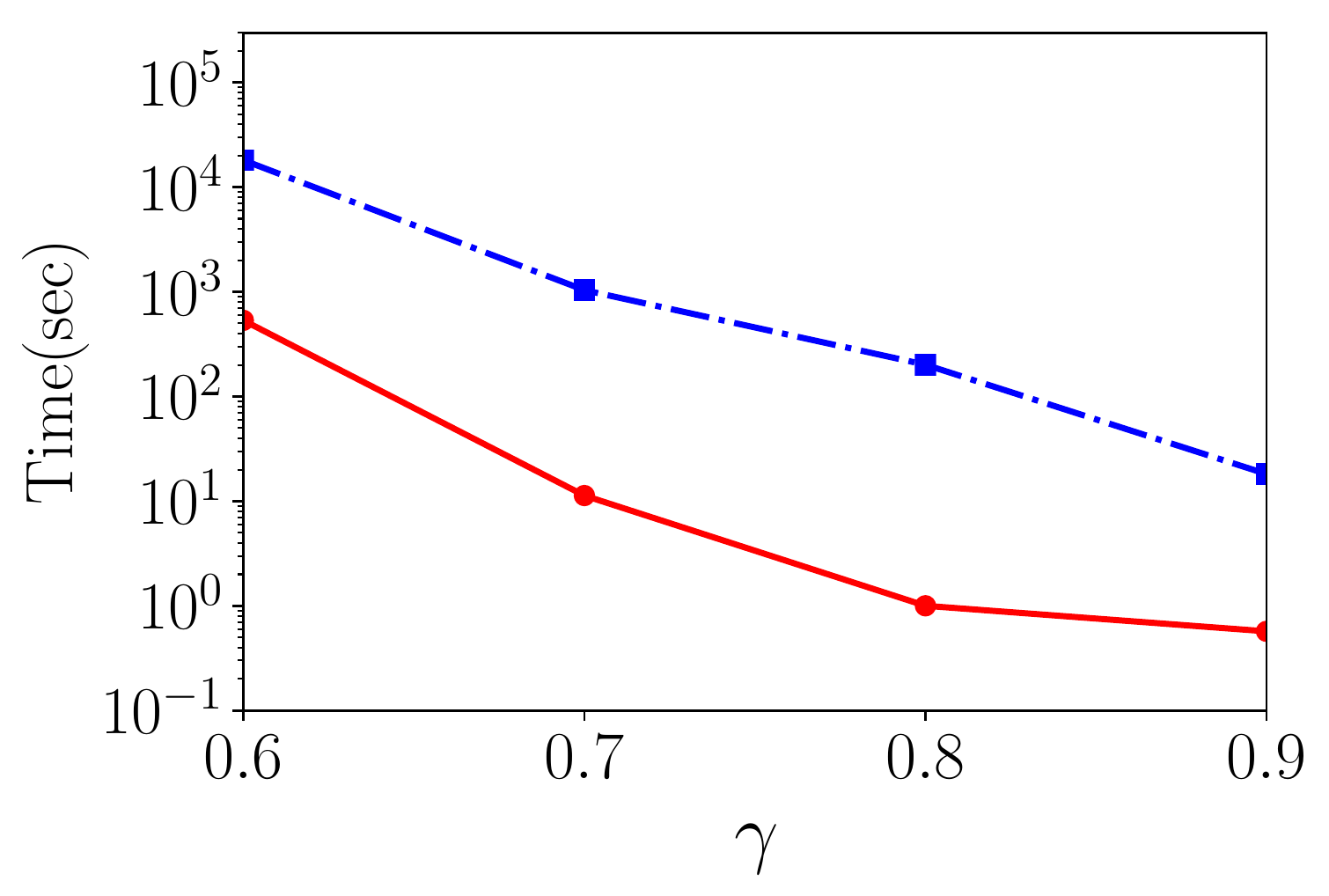}}
	\caption{The runtimes of \kqc{} and \pl{} as a function of $\gamma$.}
	\label{fig:g-vs-time}
	\vspace{-2ex}
\end{figure*}

\newcommand{\tmpl}{\pl{} takes}
\newcommand{\spd}{Speedup factor of \kqc{} over \pl}
\newcommand{\ac}{Error percentage of \kqc}
\section{Experiments}
\label{sec:expts}
\textbf{Networks and Experimental Setup.} We used real-world networks from publicly available repository at KONECT.\footnote{http://konect.uni-koblenz.de/} The networks we used are summarized in \Cref{table:graph-stats} and are converted to simple graphs by removing self-loops and multiple edges. We implemented the \pl{} and \kqc{} algorithms in C++ and compiled with \texttt{g++} compiler with \texttt{-O3} as the optimization level. The experiments are conducted on a cluster of machines equipped with a 2.0 GHz 8-Core Intel E5 2650 and 64.0 GB memory.\\

\textbf{Metrics.} As discussed in \Cref{sec:algorithm}, \kqc{} is a heuristic algorithm for extracting the top-$k$ maximal \qc{}s. There is no guarantee that it will always be correct, i.e. it may not always enumerate the $k$ largest maximal quasi-cliques. On the other hand, \pl{} mines the exact top-$k$ maximal \qc{}s. We need a metric to measure the accuracy of \kqc compared to the \pl algorithm. For this purpose, we use S{\o}ergel similarity, which is as follows. Suppose that $H = \langle h_1,h_2,h_3,\cdots,h_k \rangle$ is an ascending ordered list, maintaining the sizes of $k$ maximal \qc{}s returned by \kqc{}.\footnote{A size of a quasi-clique is the number of vertices in the quasi-clique.} Similarly, suppose that $Z = \langle z_1,z_2,z_3,\cdots,z_k \rangle$ is a list in an ascending order, which contains the sizes of the top-$k$ maximal \qc{}s, returned by \pl{}, the exact algorithm. The S{\o}ergel similarity between two lists $H$ and $Z$ is as follows:
\begin{align}
 \text{S{\o}ergel similarity (H, Z)} = \frac{\sum_{i=1}^{k} |h_i - z_i|}{\sum_{i=1}^{k}\max{(h_i, z_i)}} \times 100
\end{align}
Using a similar method, we can compute the error percentage of a list $H$ compared to a list $Z$. Here, we define how we measure the error percent of a list $H$ from a list $Z$:
{\small
\begin{align}
\text{Error percent (H, Z)} = \left(1 - \frac{\sum_{i=1}^{k} |h_i - z_i|}{\sum_{i=1}^{k}\max{(h_i, z_i)}}\right) \times 100
\label{eq:error-percent}
\end{align}
}

For our purpose, the S{\o}ergel similarity is better suited than other metrics such as Jaccard similarity. In particular, if we used Jaccard similarity to measure the similarity between vertex sets of two quasi-cliques, this will fail to consider the sizes of the quasi-cliques. If two algorithms return sets of quasi-cliques that are of exactly the same sizes, but whose elements are different, then the Jaccard similarity will show a poor match, while the S{\o}ergel similarity will show a perfect match. 

Note that S{\o}ergel similarity shows the similarity of two lists of numbers. Here, we consider the lists of \textit{sizes} of quasi-cliques, obtained by \kqc and \pl algorithms. As mentioned above, the two lists with length $k$, where each list is sorted in an ascending order of sizes of quasi-cliques. Henceforth, when we refer to the error percent of \kqc{}, we compare the returned lists of \kqc{} and \pl{} using \Cref{eq:error-percent}.
\begin{table}[t!]
	\centering
	
	\scalebox{0.84}{
\begin{tabular}{|l|r|r|r|c|}
	\hline
	Graph   & \multicolumn{1}{c|}{\#Vertices} & \multicolumn{1}{c|}{\#Edges} & \multicolumn{1}{c|}{Maximum Degree} & Average Degree \\ \hline
	\advo    & \num{5155}                            & \num{39285}                        & \num{803}                                 & 15.24          \\ \hline
	\route   & \num{6474}                            & \num{12572}                        & \num{1458}                                & 3.88           \\ \hline
	\bible   & \num{1773}                           & \num{9131}                        & \num{364}                                 & 10.30          \\ \hline
	\slsh    & \num{51083}                          & \num{116573}                        & \num{2915}                                 & 4.56           \\ \hline
	\live    & \num{104103}                         & \num{2193083}                      & \num{2980}                                & 42.13          \\ \hline
	\youtube & \num{1134890}                        & \num{2987624}                     & \num{28754}                               & 5.27           \\ \hline
	\hyves   & \num{1402673}                         & \num{2777419}                      & \num{31883}                               & 3.96           \\ \hline
\end{tabular}
	}
\caption{Summary of the input graphs.}
\label{table:graph-stats}
\end{table}
\begin{figure*}[t]
	\centering
	\includegraphics[width=0.6\textwidth]{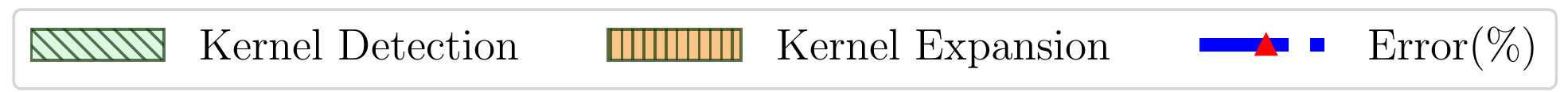}
	\vspace{-5ex}
\end{figure*}
\begin{figure*}[t!]
	\captionsetup[subfigure]{justification=centering}
\centering
  \subfloat[$\gamma=0.6$ did not finish after \num{259200} secs]
  [$\gamma=0.6$. \pl{} did not \\ finish after \num{259200} secs.]{%
  \includegraphics[width=.24\textwidth] {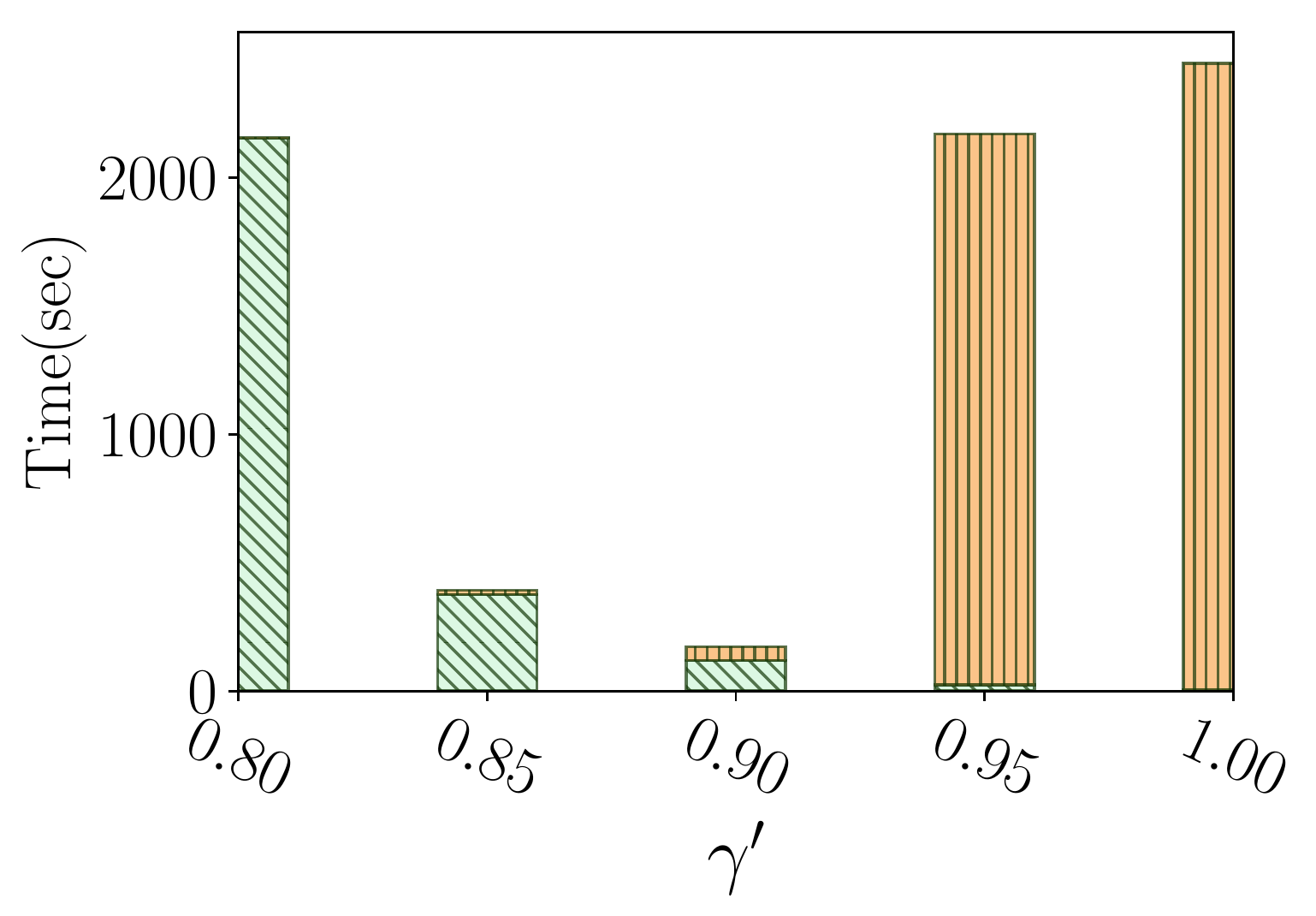} \label{fig:advo-0.6}}
  \subfloat[$\gamma=0.7$, \tmpl{} 122618 secs][$\gamma=0.7$, error is $0\%$. \\ \tmpl{} 122618 secs]{%
  \includegraphics[width=.26\textwidth]{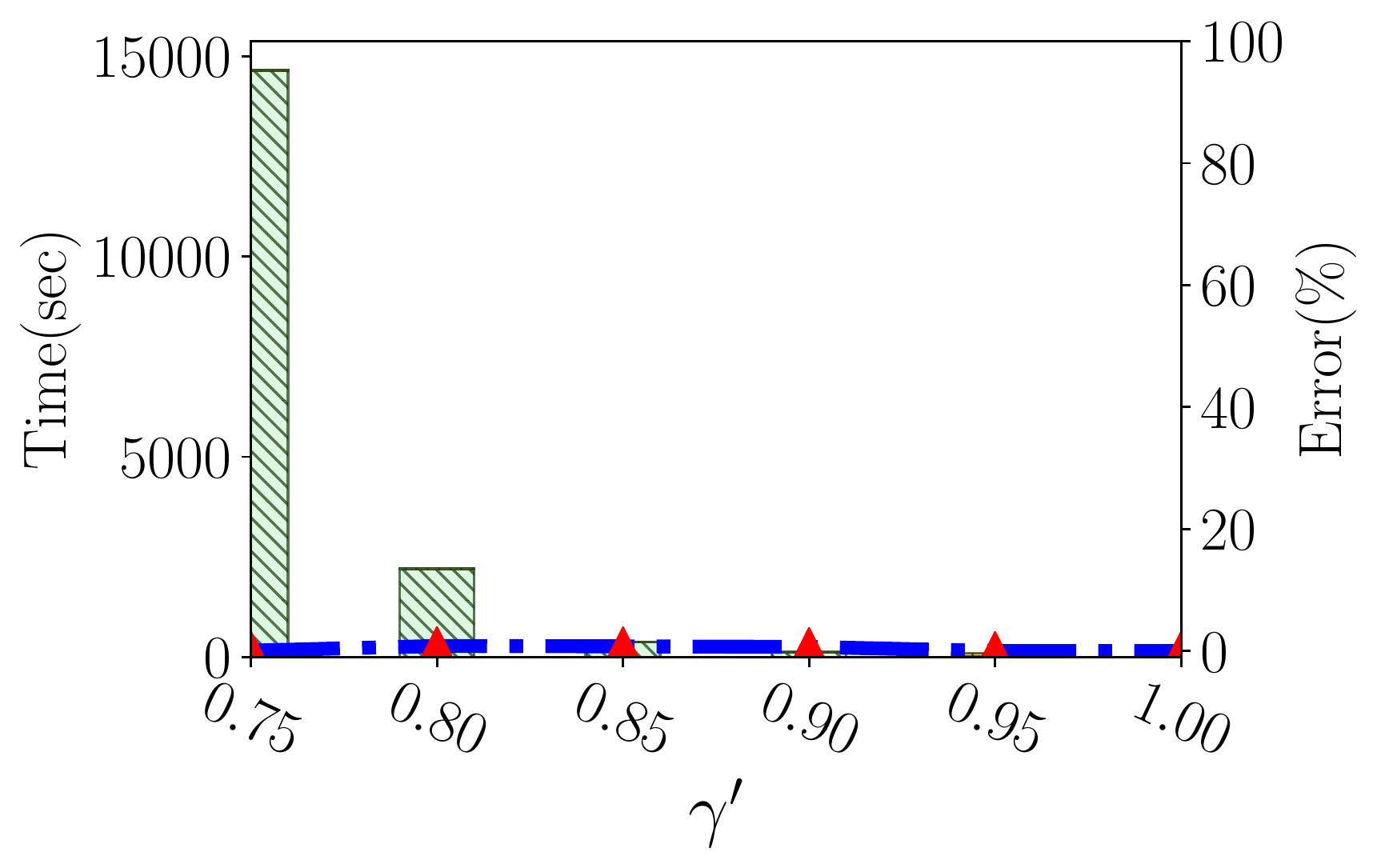} \label{fig:advo-0.7}}
    \subfloat[$\gamma=0.8$, \tmpl{} 2185 secs][$\gamma=0.8$, error is $0\%$. \\ \tmpl{} 2185 secs]{%
  	\includegraphics[width=.23\textwidth] {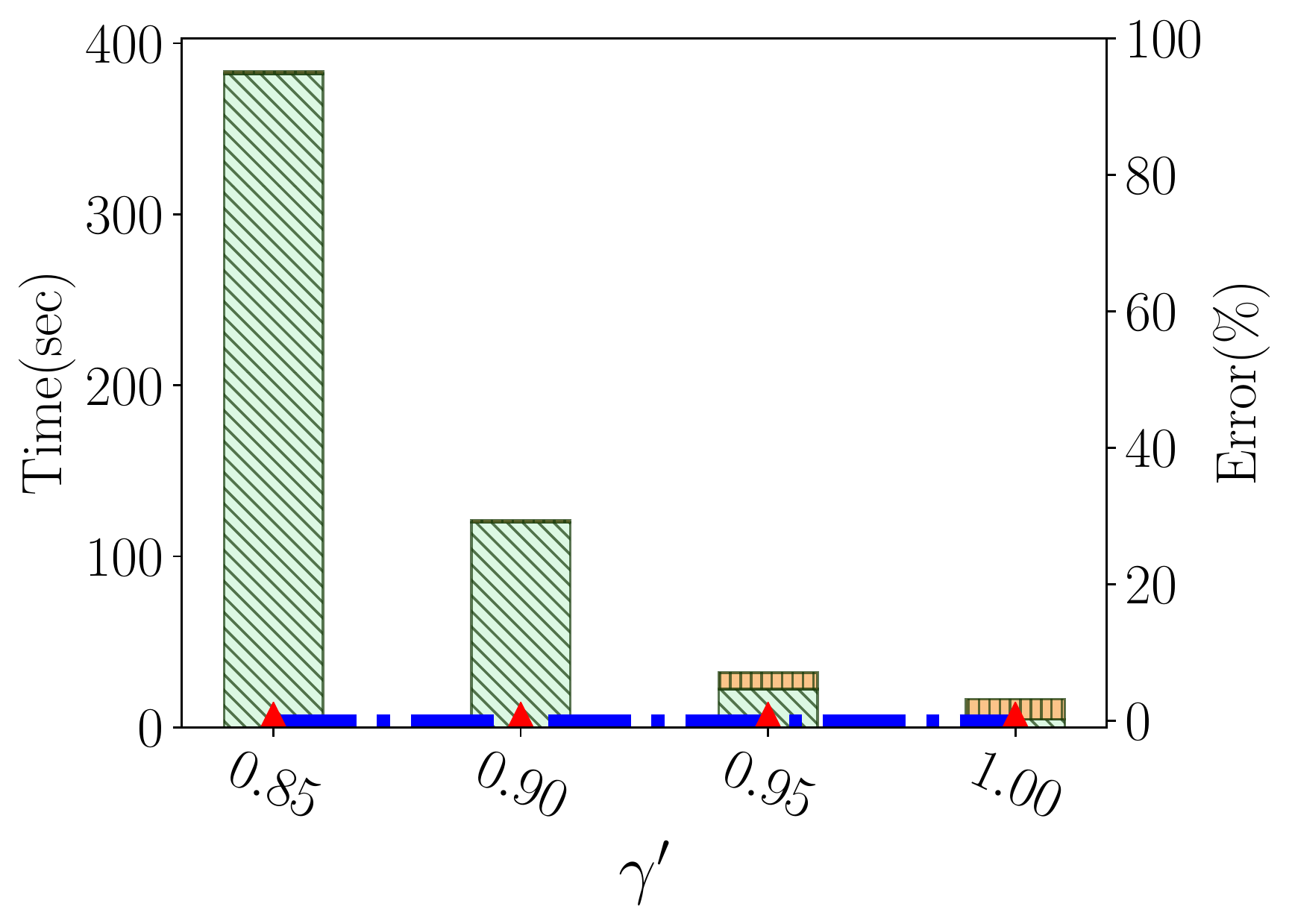}}
  \subfloat[$\gamma=0.9$, \tmpl{} 113 secs][$\gamma=0.9$, error is $0\%$. \\ \tmpl{} 113 secs]{%
  	\includegraphics[width=.23\textwidth]{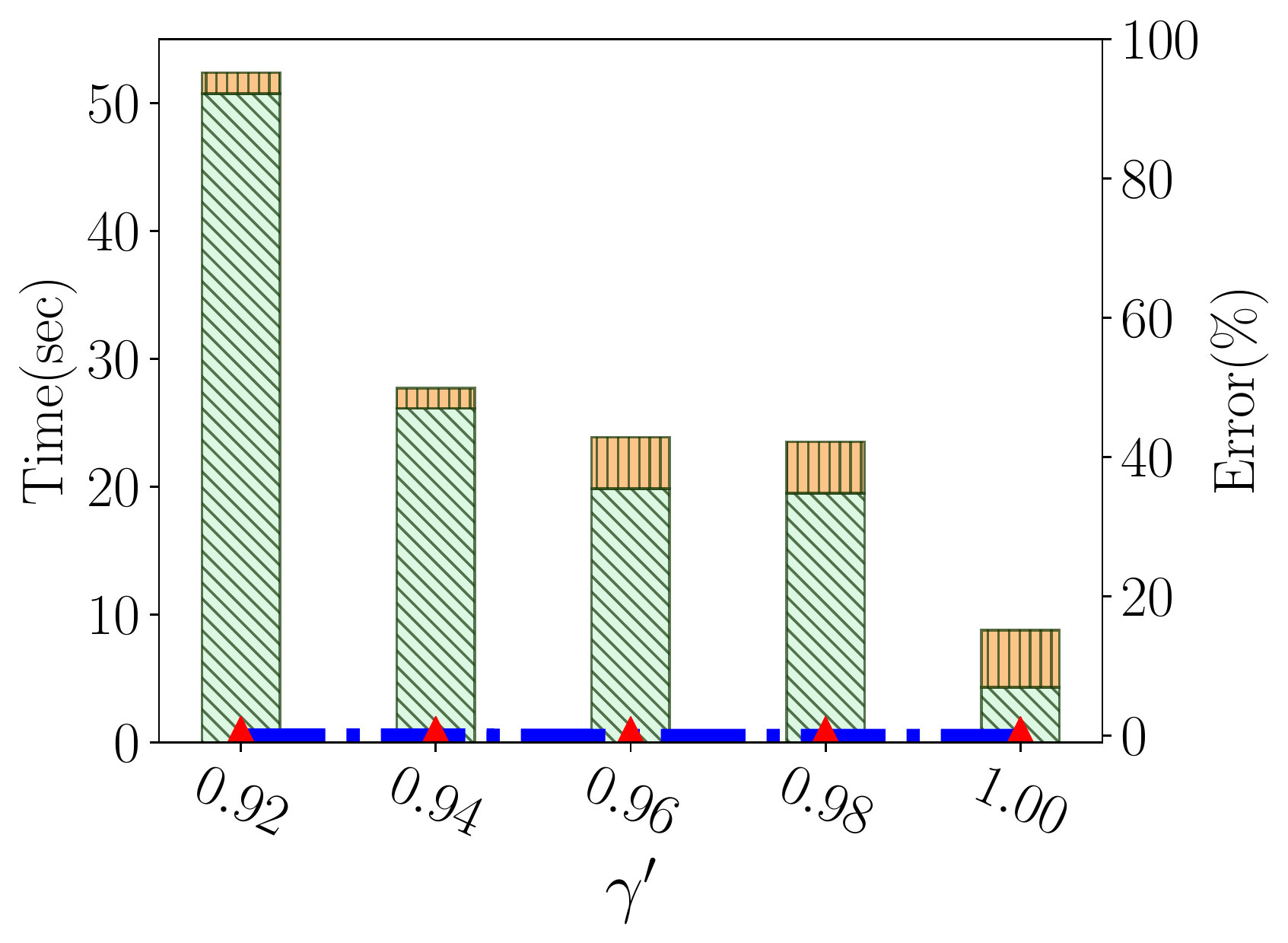}}
  	\caption{\advo, $k = 100$, $k' = 300$, \mnsz{} $= 5$.}
	\label{fig:advo-gp}
	\vspace{-4.5ex}
\end{figure*}
\begin{figure*}[t!]
	\captionsetup[subfigure]{justification=centering}
	\centering
	\subfloat[ $\gamma=0.6$, \tmpl 97 secs][ $\gamma=0.6$, error is $0\%$. \\ \tmpl{} 97 secs]{%
		\includegraphics[width=.26\textwidth] {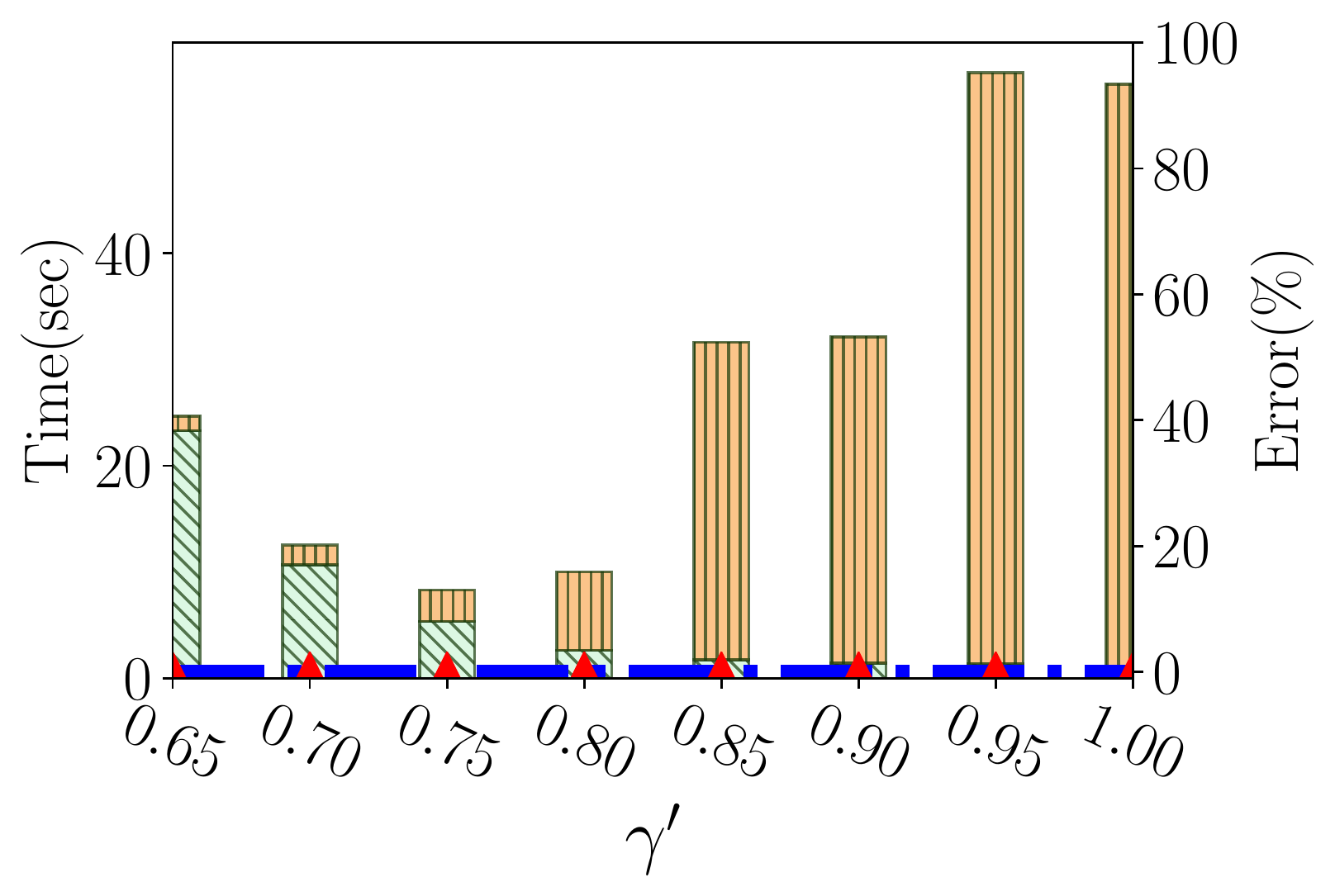}}
	\subfloat[ $\gamma=0.7$, \tmpl 10 secs][ $\gamma=0.7$, error is $0\%$. \\ \tmpl{} 10 secs]{%
		\includegraphics[width=.24\textwidth]{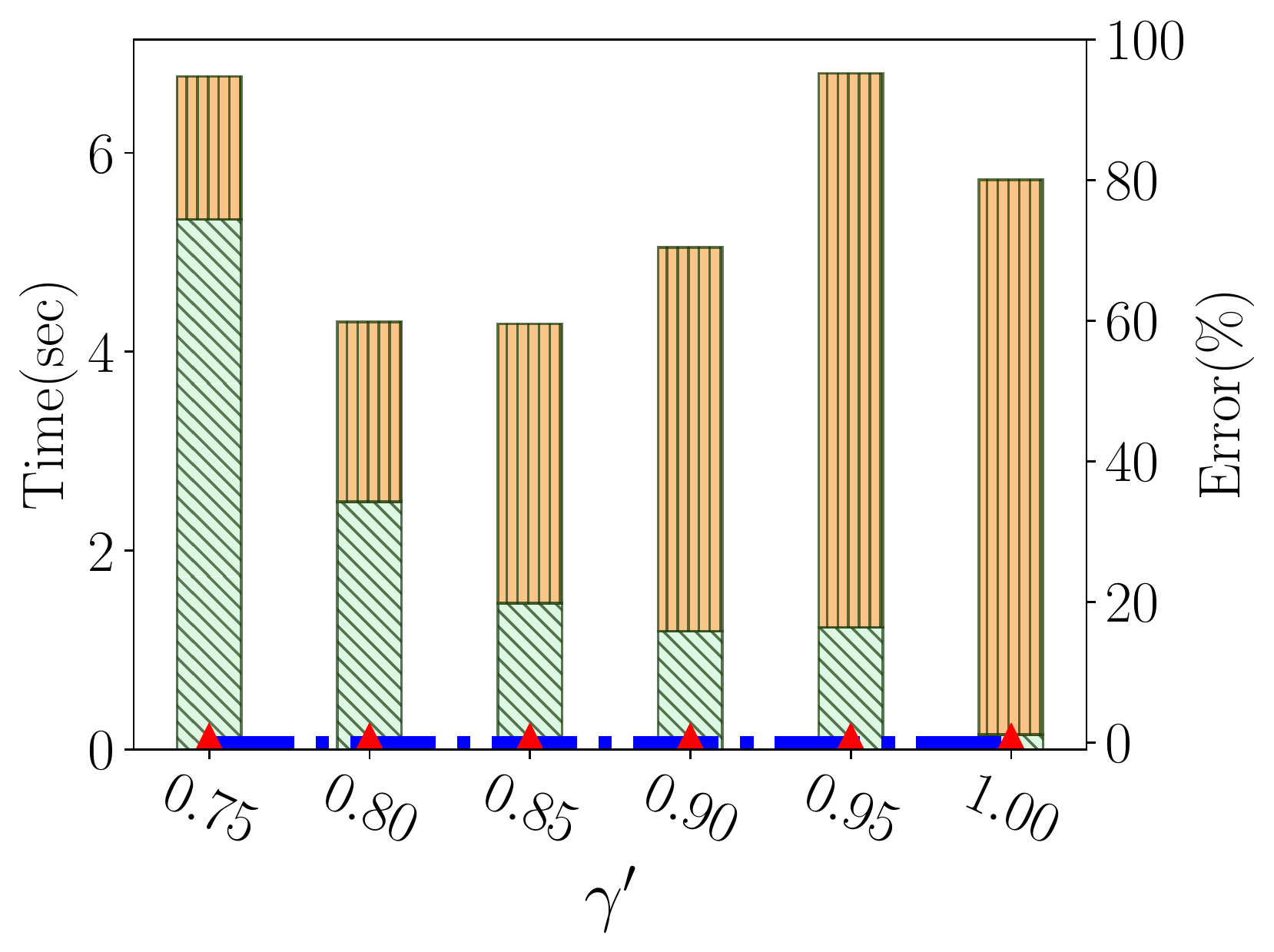}}
	\subfloat[ $\gamma=0.8$, \tmpl 2 secs][ $\gamma=0.8$, error is $0\%$. \\ \tmpl{} 2 secs]{%
		\includegraphics[width=.24\textwidth] {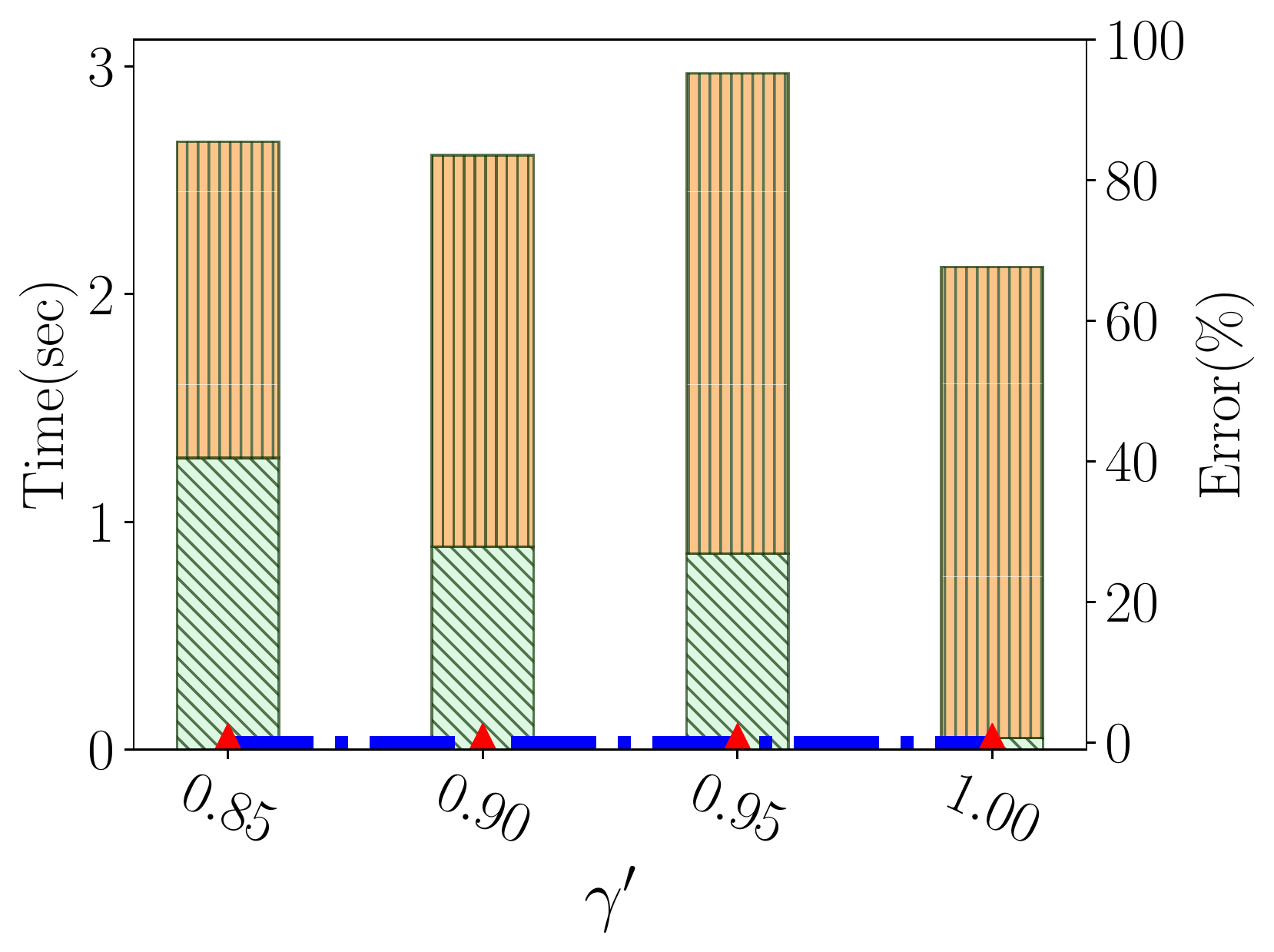}}
	\subfloat[ $\gamma=0.9$, \tmpl 0.9 secs][ $\gamma=0.9$, error is $0\%$. \\ \tmpl{} 0.9 secs]{%
		\includegraphics[width=.24\textwidth]{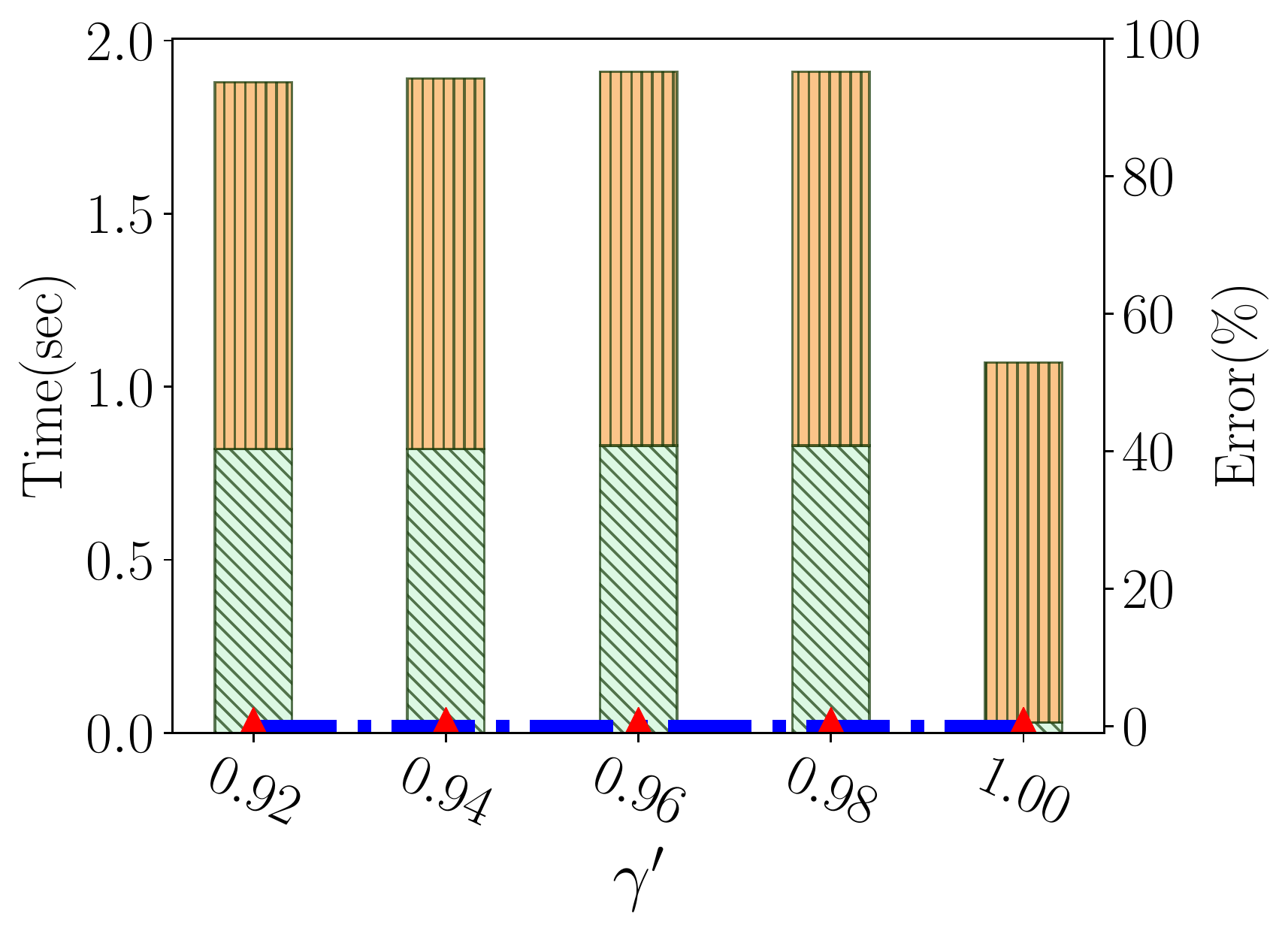}}
	\caption{\route, $k = 100$, $k' = 300$, \mnsz{} $= 5$.}
	\label{fig:route-gp}
	\vspace{-4.5ex}
\end{figure*}
 \begin{figure*}[t!]
 		\captionsetup[subfigure]{justification=centering}
 	\centering  
   \subfloat[$\gamma=0.6$, \tmpl 12760 secs][ $\gamma=0.6$, error is $0\%$. \\ \tmpl{} 12760 secs]{%
 	\includegraphics[width=.24\textwidth] {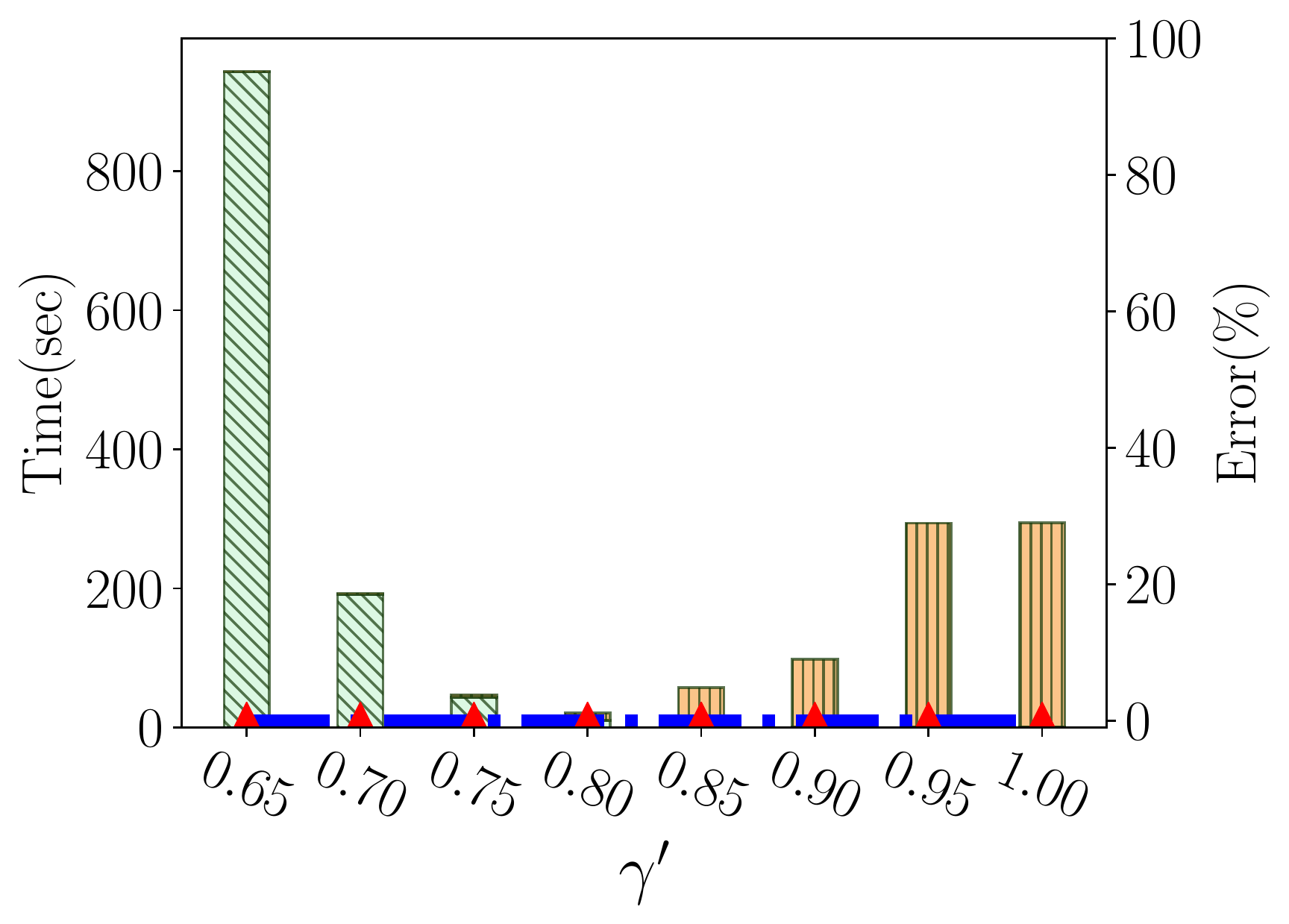} \label{fig:bible-0.6}}
 \subfloat[ $\gamma=0.7$, \tmpl 189 secs][ $\gamma=0.7$, error is at most $0.1\%$. \\ \tmpl{} 189 secs]{%
 	\includegraphics[width=.24\textwidth]{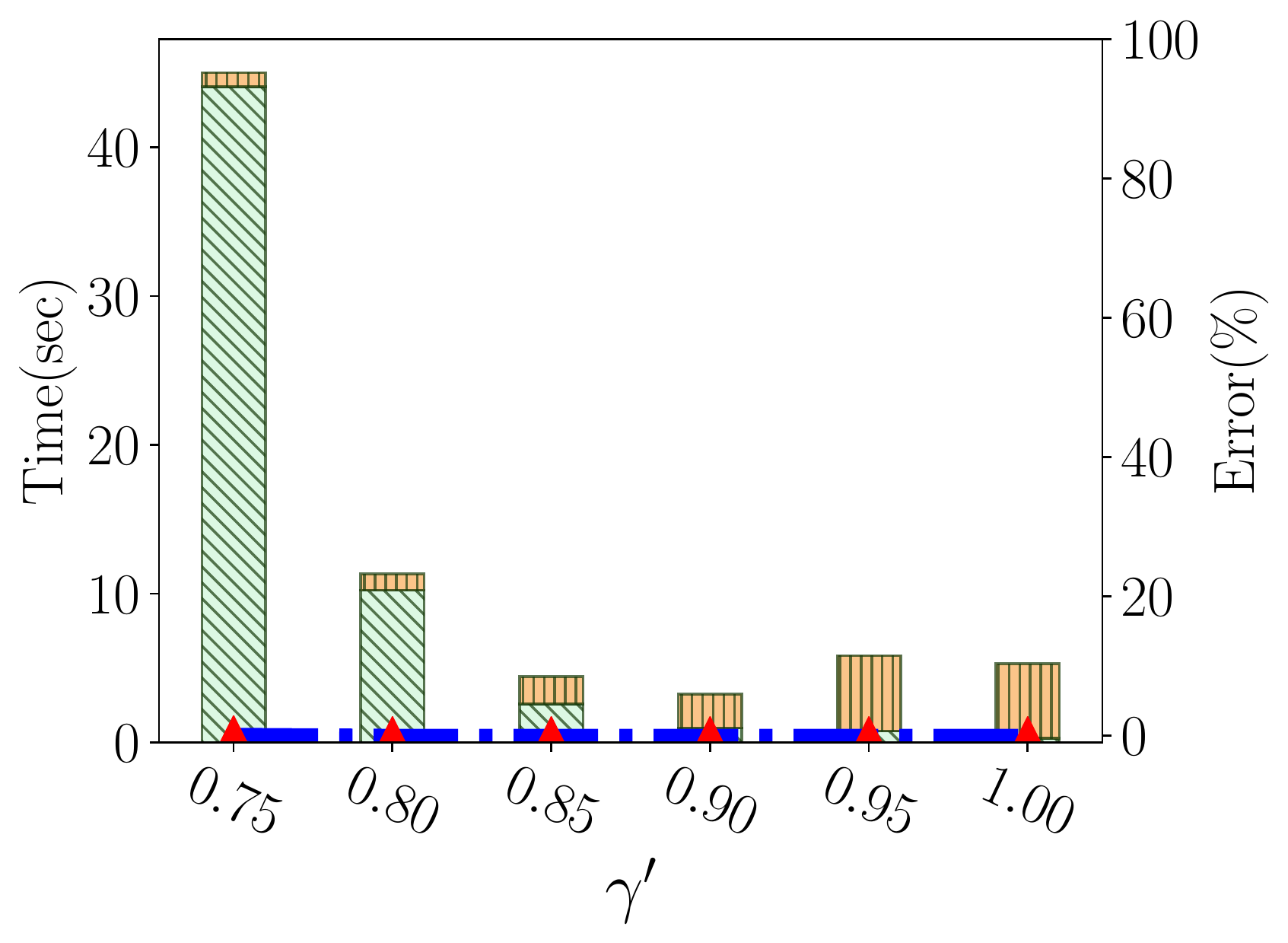}\label{fig:bible-0.7}}
 \subfloat[ $\gamma=0.8$, \tmpl 10 secs][ $\gamma=0.8$, error is at most $0.8\%$. \\ \tmpl{} 10 secs]{%
 	\includegraphics[width=.24\textwidth] {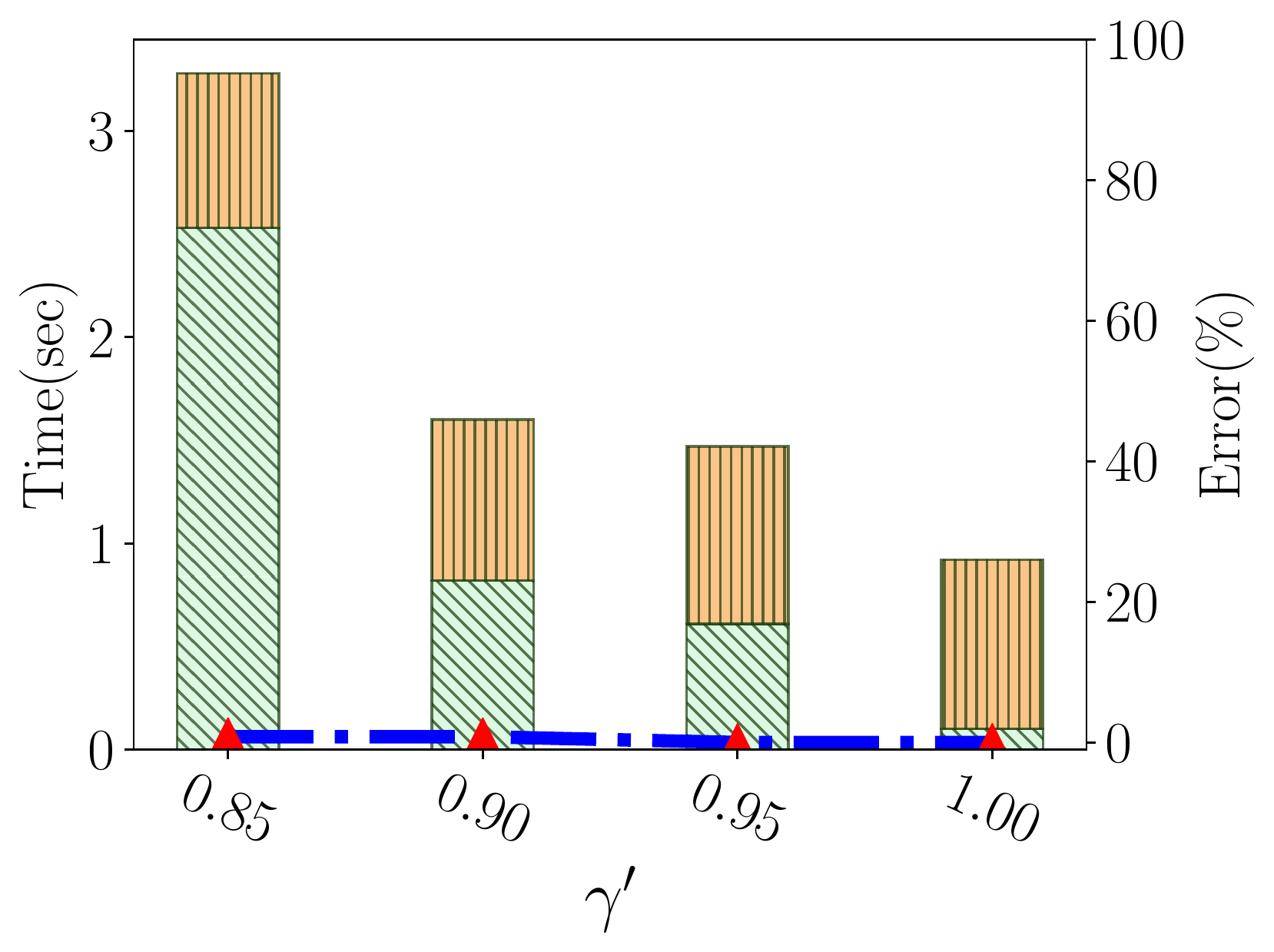}}
 \subfloat[ $\gamma=0.9$, \tmpl 0.8 secs][ $\gamma=0.9$, error is at most $0.6\%$. \\ \tmpl{} 0.8 secs]{%
 	\includegraphics[width=.26\textwidth]{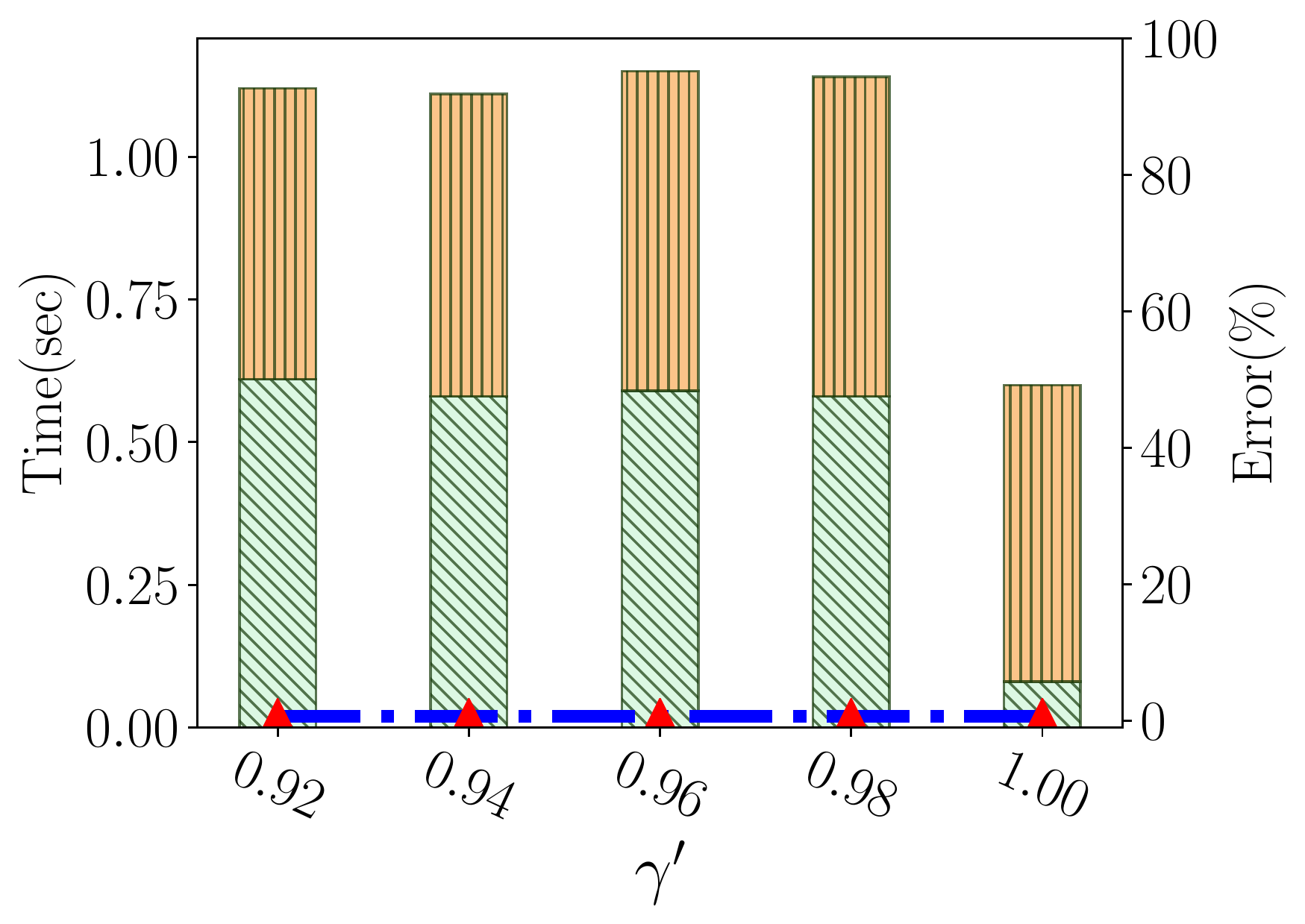}}
  	\caption{\bible, $k = 100$, $k' = 300$, \mnsz{} $= 5$.}
	\label{fig:bible-gp}
	\vspace{-4.5ex}
\end{figure*}
\begin{figure*}[t]
	\centering
	\includegraphics[width=0.6\textwidth]{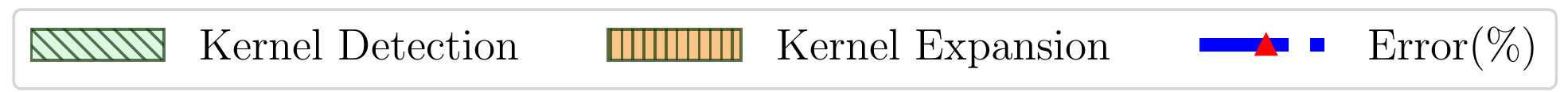}
	\vspace{-5ex}
\end{figure*}
\begin{figure*}[t!]
 		\captionsetup[subfigure]{justification=centering}
 	\centering
   \subfloat[ $\gamma=0.6$, \tmpl 18152 secs][ $\gamma=0.6$, error is $0\%$. \\ \tmpl{} 18152 secs]{%
 	\includegraphics[width=.24\textwidth] {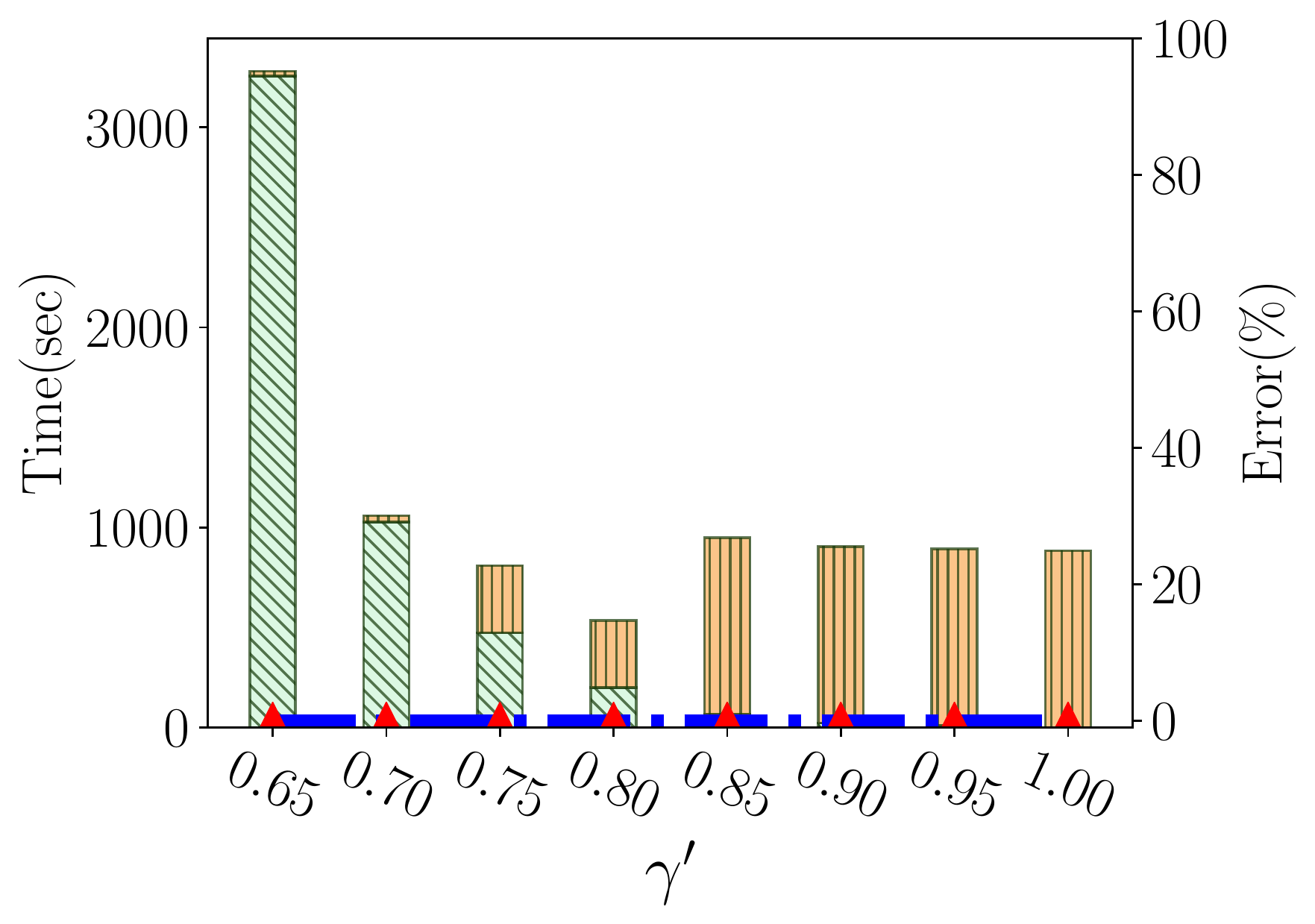} \label{fig:slash-0.6}}
 \subfloat[ $\gamma=0.7$, \tmpl 1039 secs][ $\gamma=0.7$, error is $0\%$. \\ \tmpl{} 1039 secs]{%
 	\includegraphics[width=.24\textwidth]{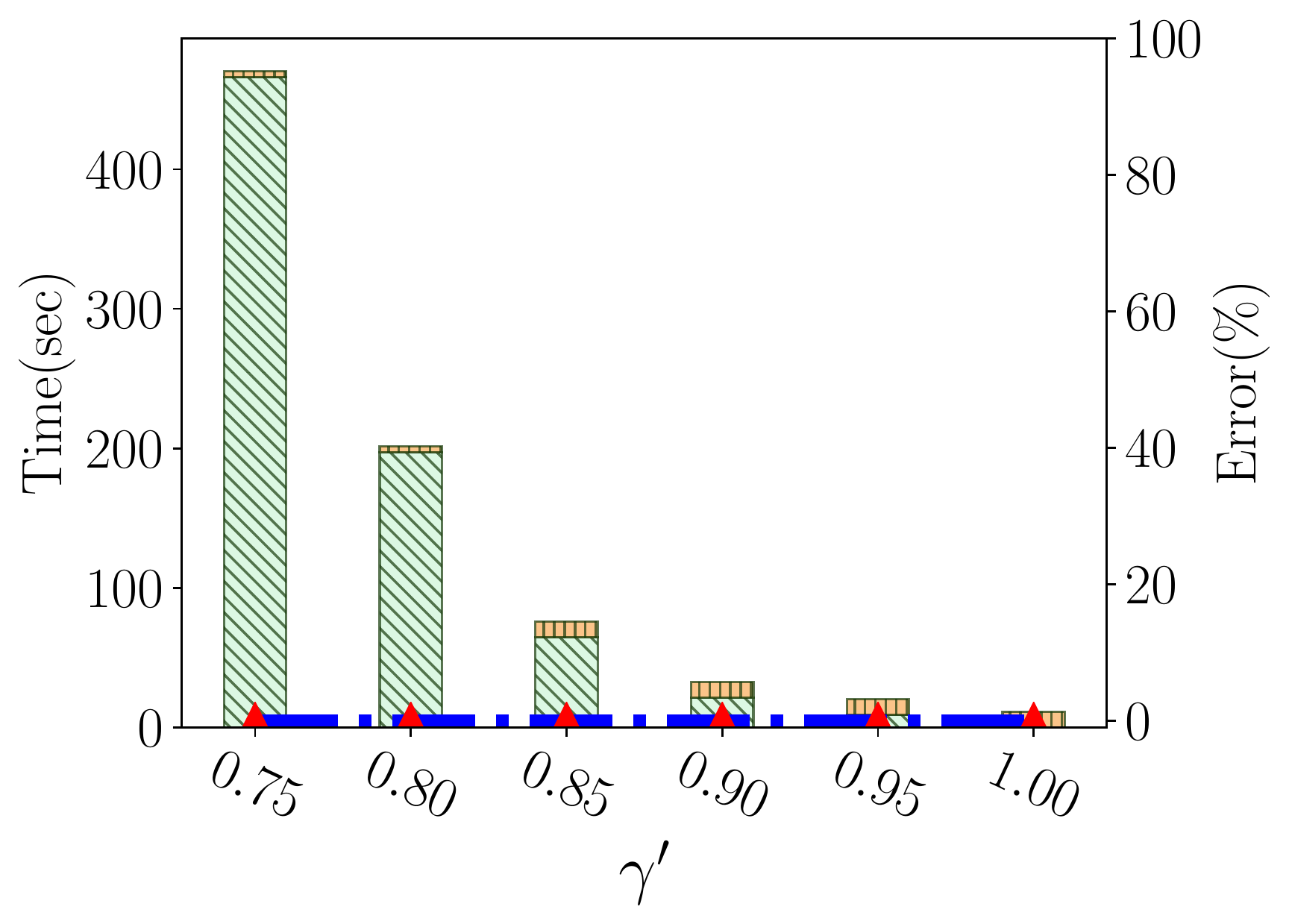}}
 \subfloat[ $\gamma=0.8$, \tmpl 201 secs][ $\gamma=0.8$, error is at most $2.1\%$. \\ \tmpl{} 201 secs]{%
 	\includegraphics[width=.24\textwidth] {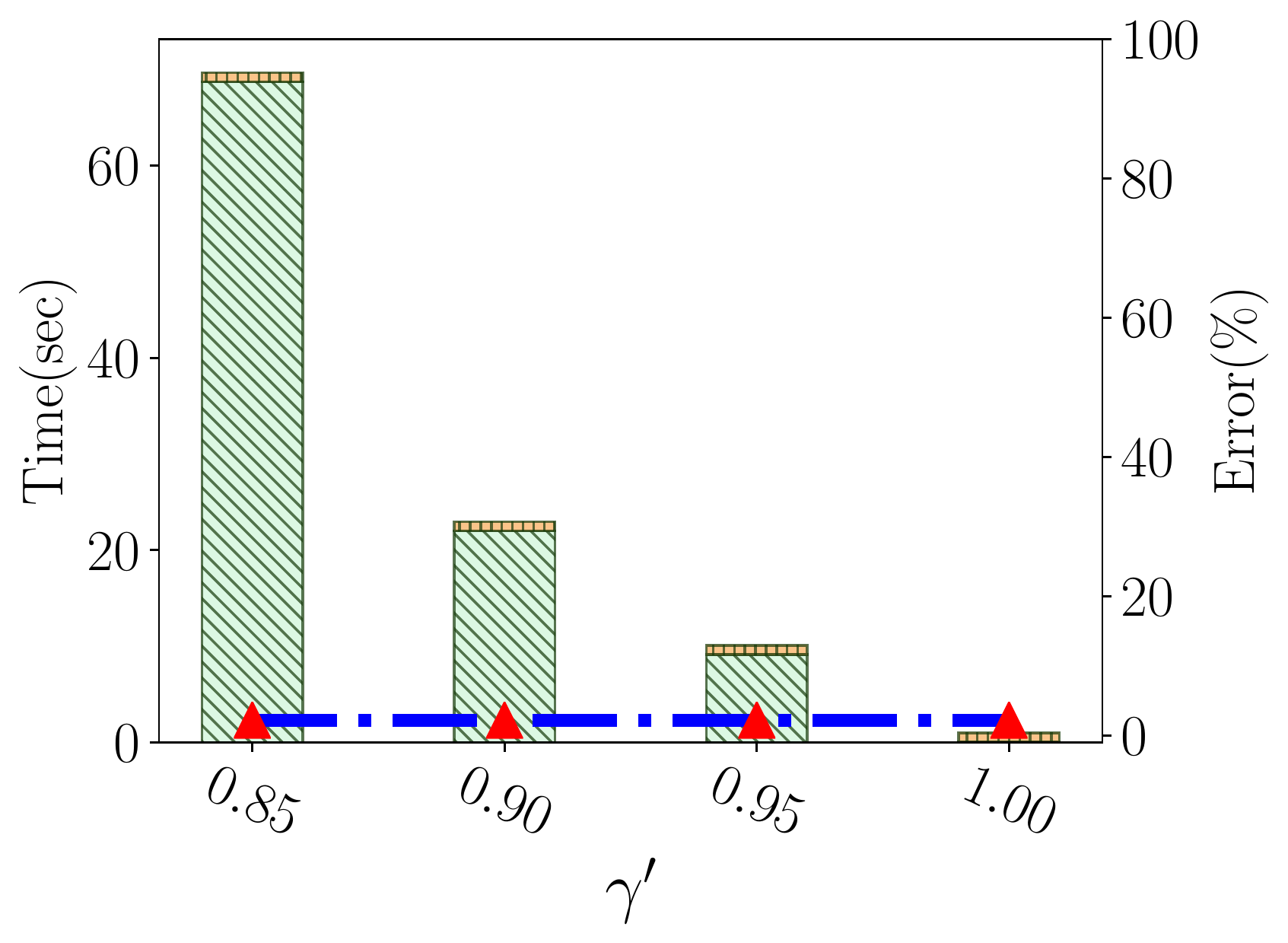} \label{fig:slash-0.8}}
 \subfloat[ $\gamma=0.9$, \tmpl 18 secs][ $\gamma=0.9$, error is $0\%$. \\ \tmpl{} 18 secs]{%
 	\includegraphics[width=.24\textwidth]{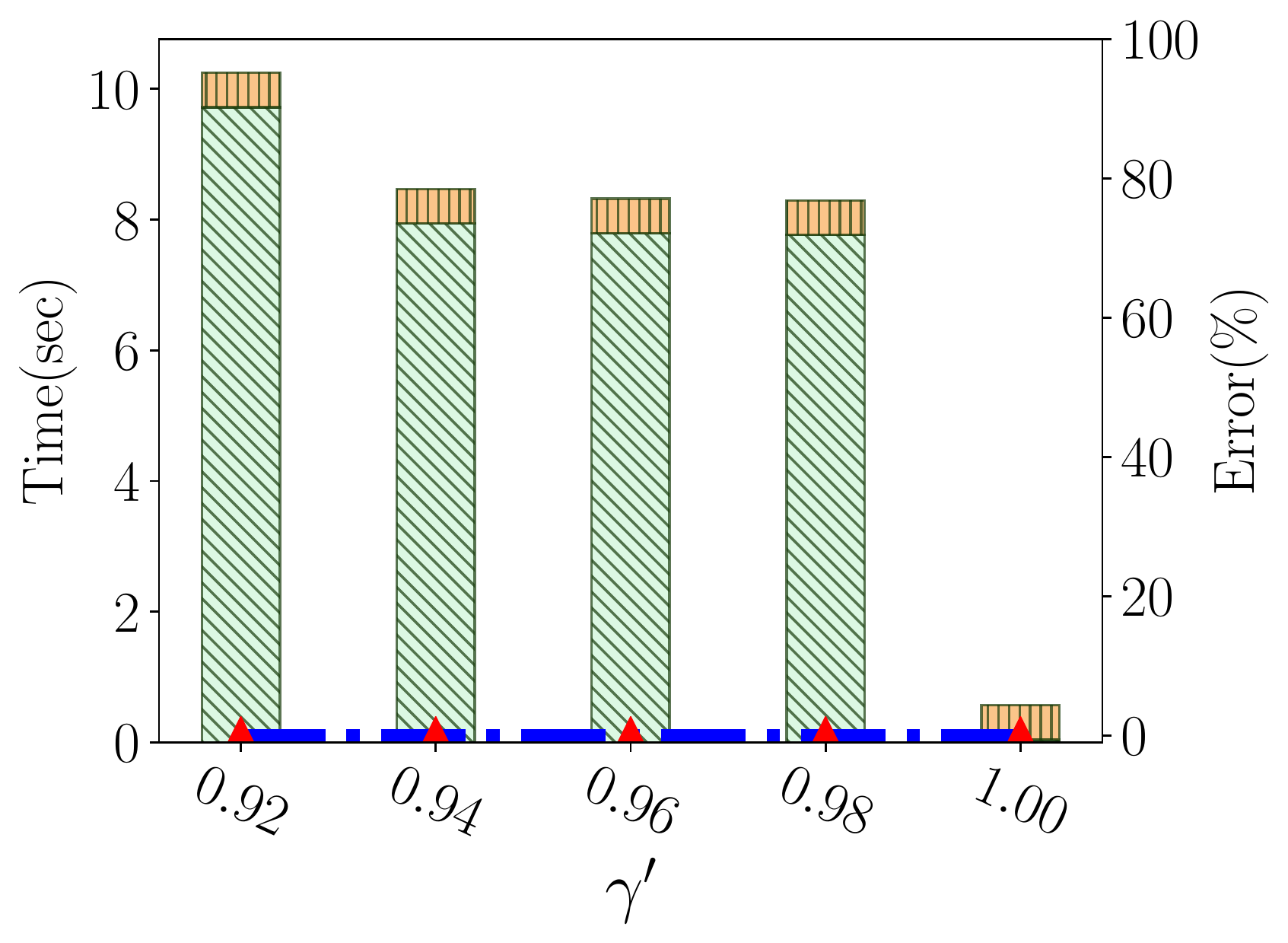}}
  \caption{\slsh, $k = 100$, $k' = 300$, \mnsz{} $= 5$.}
\label{fig:slash-gp}
\end{figure*}
\begin{figure*}[t]
	\centering
	\includegraphics[width=0.5\textwidth]{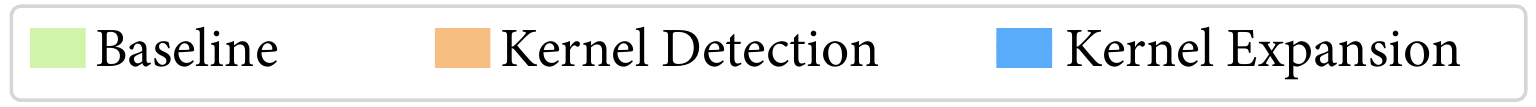}
	\label{fig:mnszt-legend}
	\vspace{-5ex}
\end{figure*}
\begin{figure*}[t]
	\captionsetup[subfigure]{justification=centering}
	\centering
	\subfloat[\advo]{%
		\includegraphics[width=.25\textwidth] {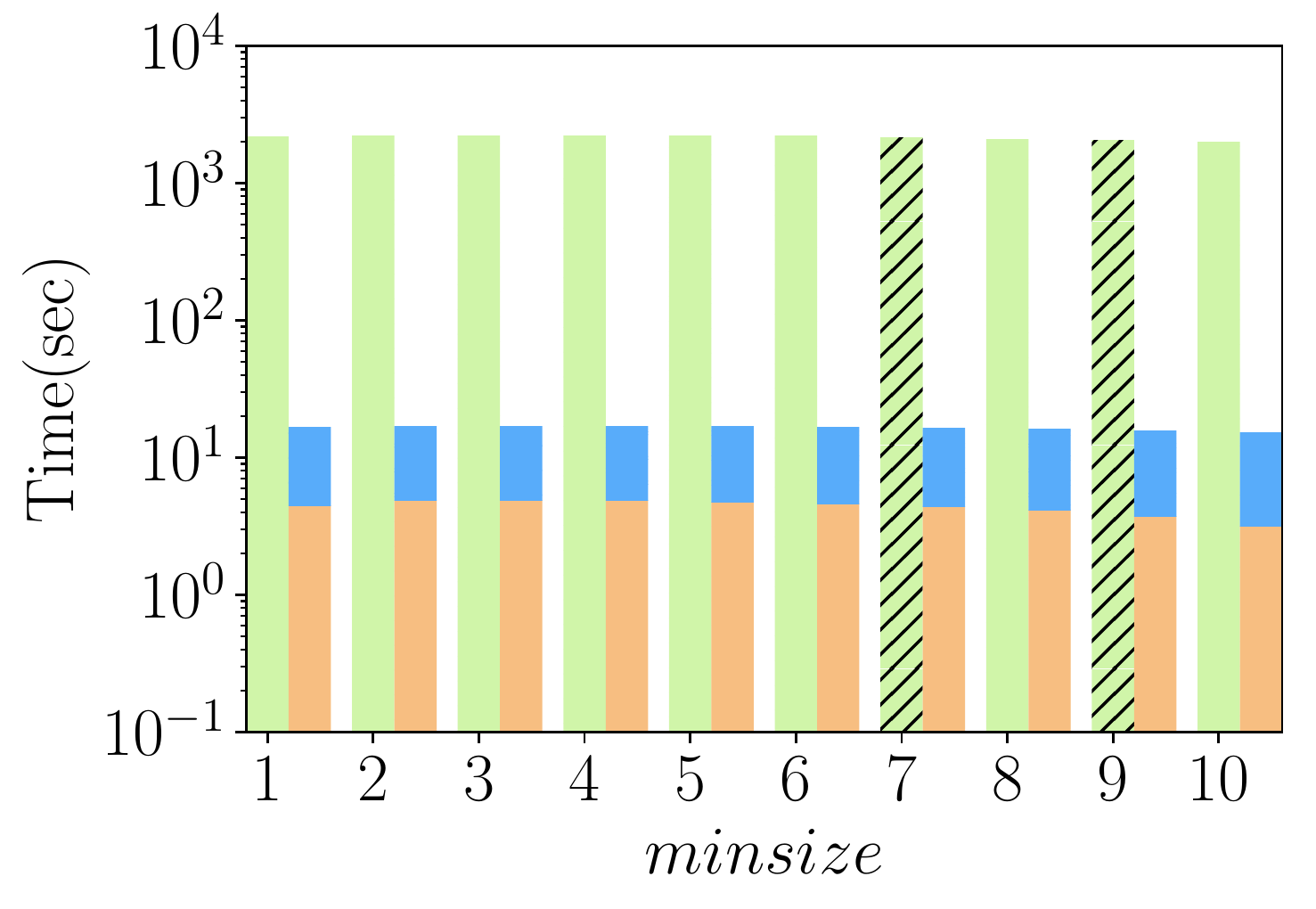} \label{fig:mnszt-advo}}
	\subfloat[\route]{%
		\includegraphics[width=.25\textwidth]{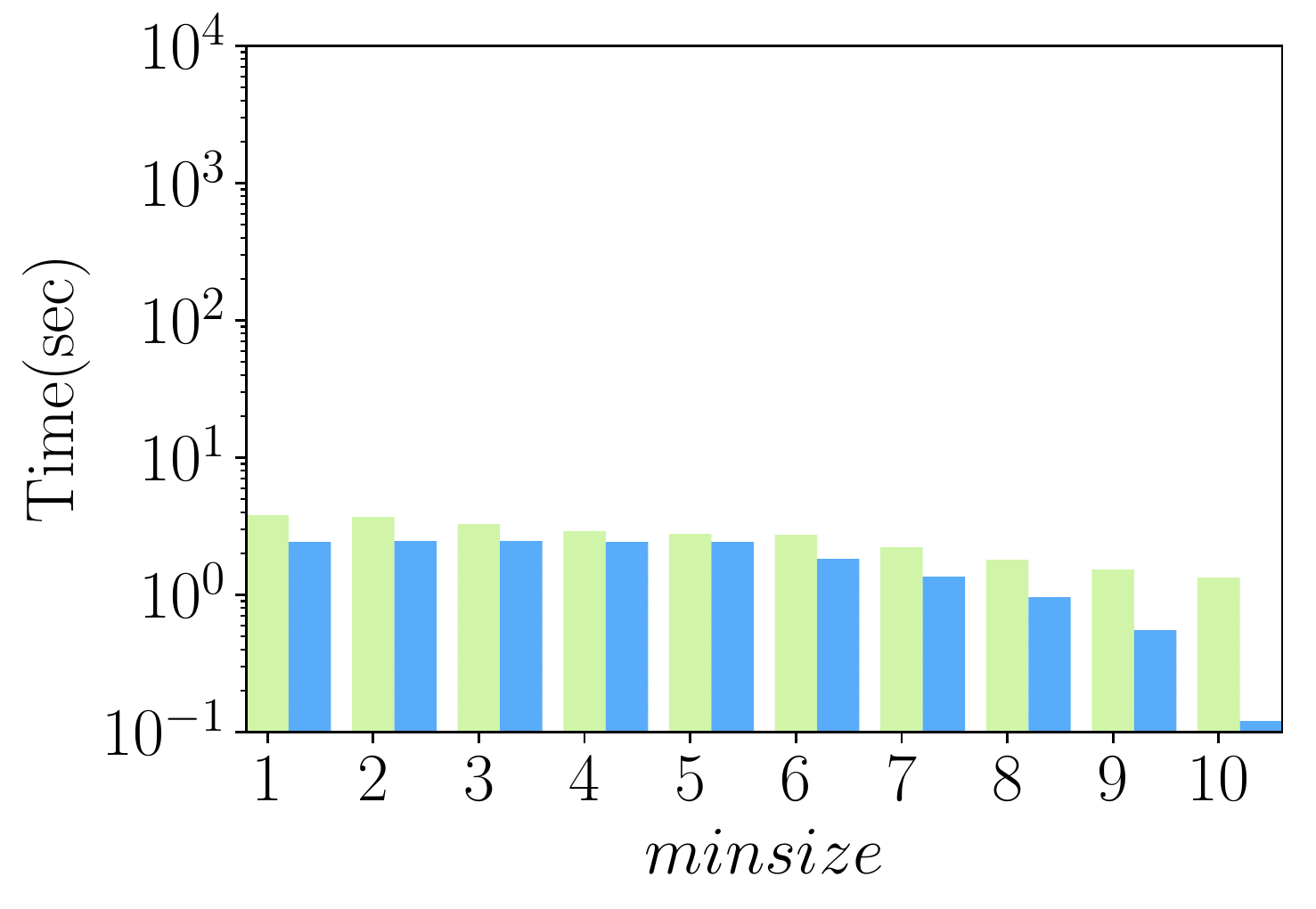}  \label{fig:mnszt-route}}  
	\subfloat[\bible]{%
		\includegraphics[width=.25\textwidth]{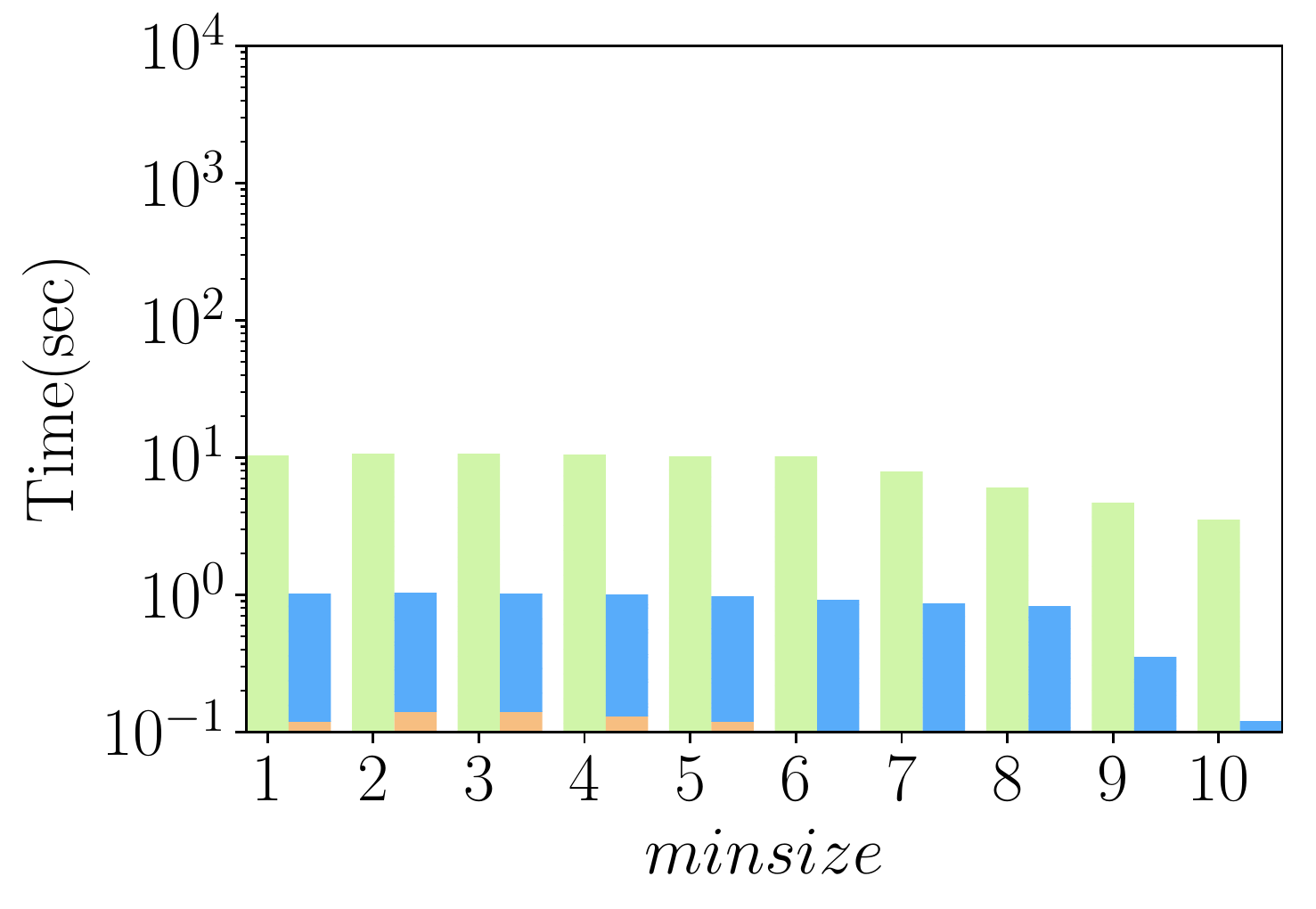}, \label{fig:mnszt-bible}}
	\subfloat[\slsh]{%
		\includegraphics[width=.25\textwidth] {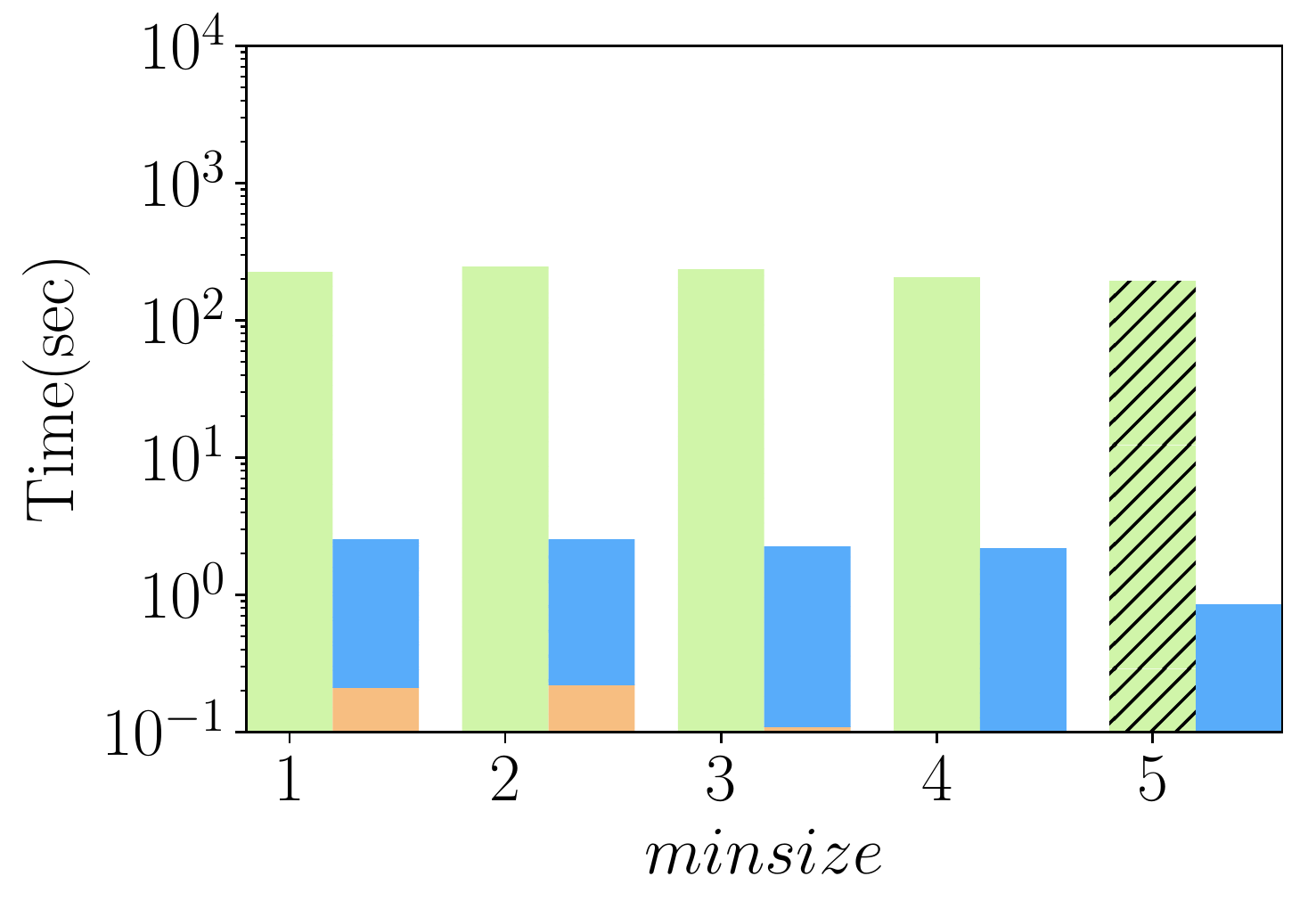} \label{fig:mnszt-slash}}
	\hfill 
	\caption{ {The runtimes of \kqc and \pl as a function of \mnsz, for $\gamma=0.8$, $\gamma' = 1.0$, $k = 100$, and $k' = 300$.}}
	\label{fig:mnszt}
	\vspace{-4ex}
\end{figure*}

\begin{figure*}[t]
	\captionsetup[subfigure]{justification=centering}
	\centering
	\subfloat[\advo]{%
		\includegraphics[width=.25\textwidth]{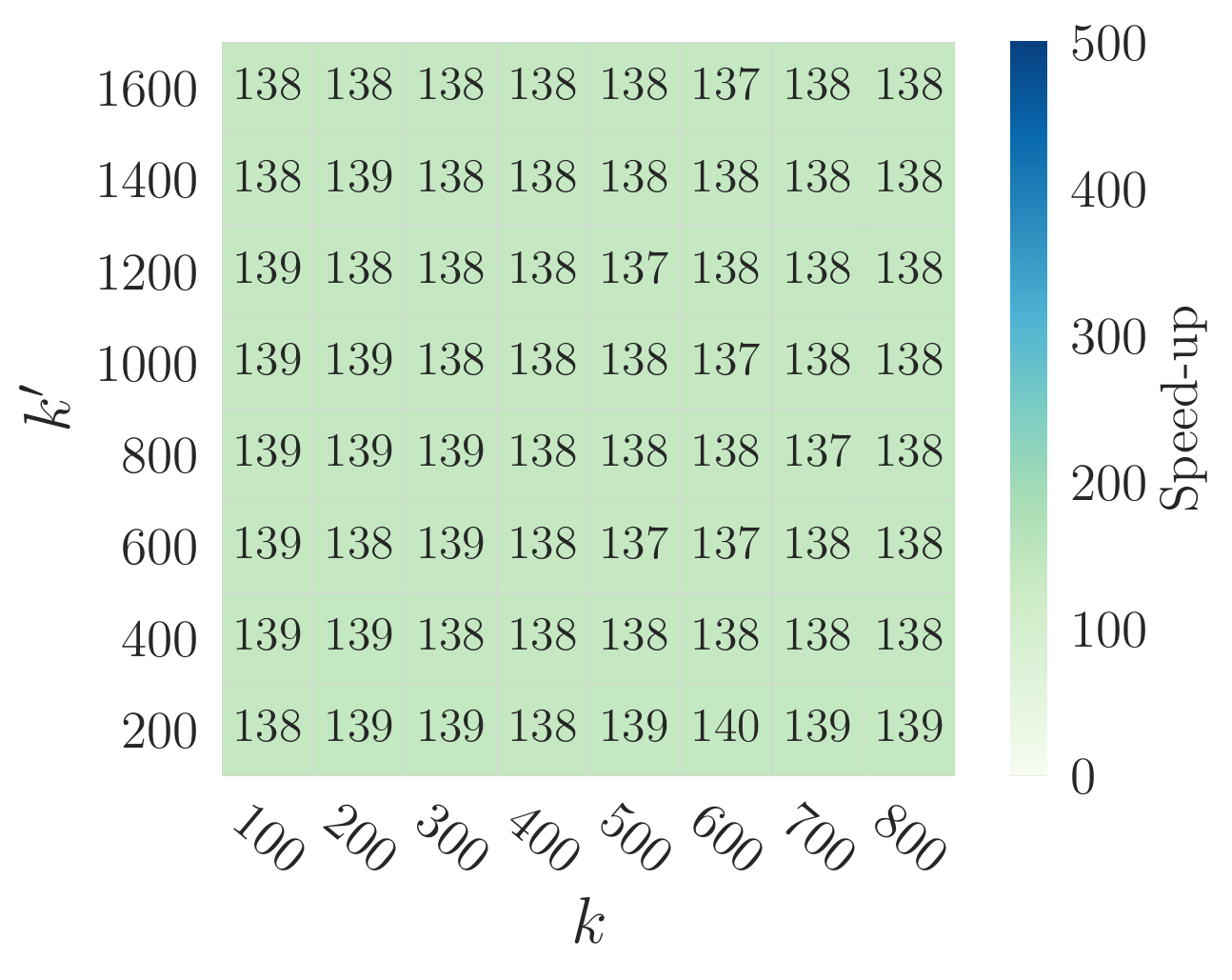} \label{fig:kac-advo-sp}}
	\subfloat[\route]{%
		\includegraphics[width=.25\textwidth]{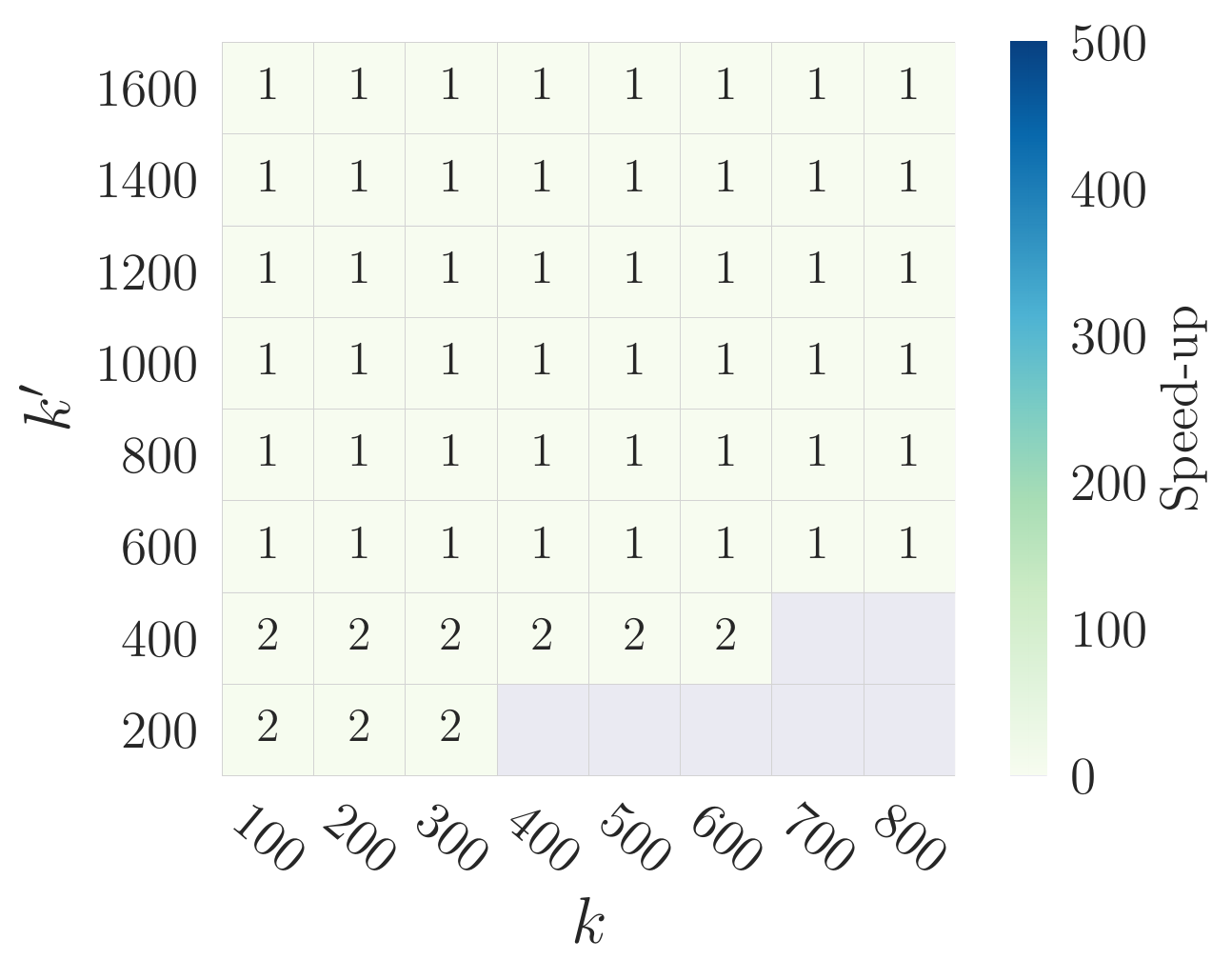} \label{fig:kac-route-sp}}
	\subfloat[\bible]{%
		\includegraphics[width=.25\textwidth]{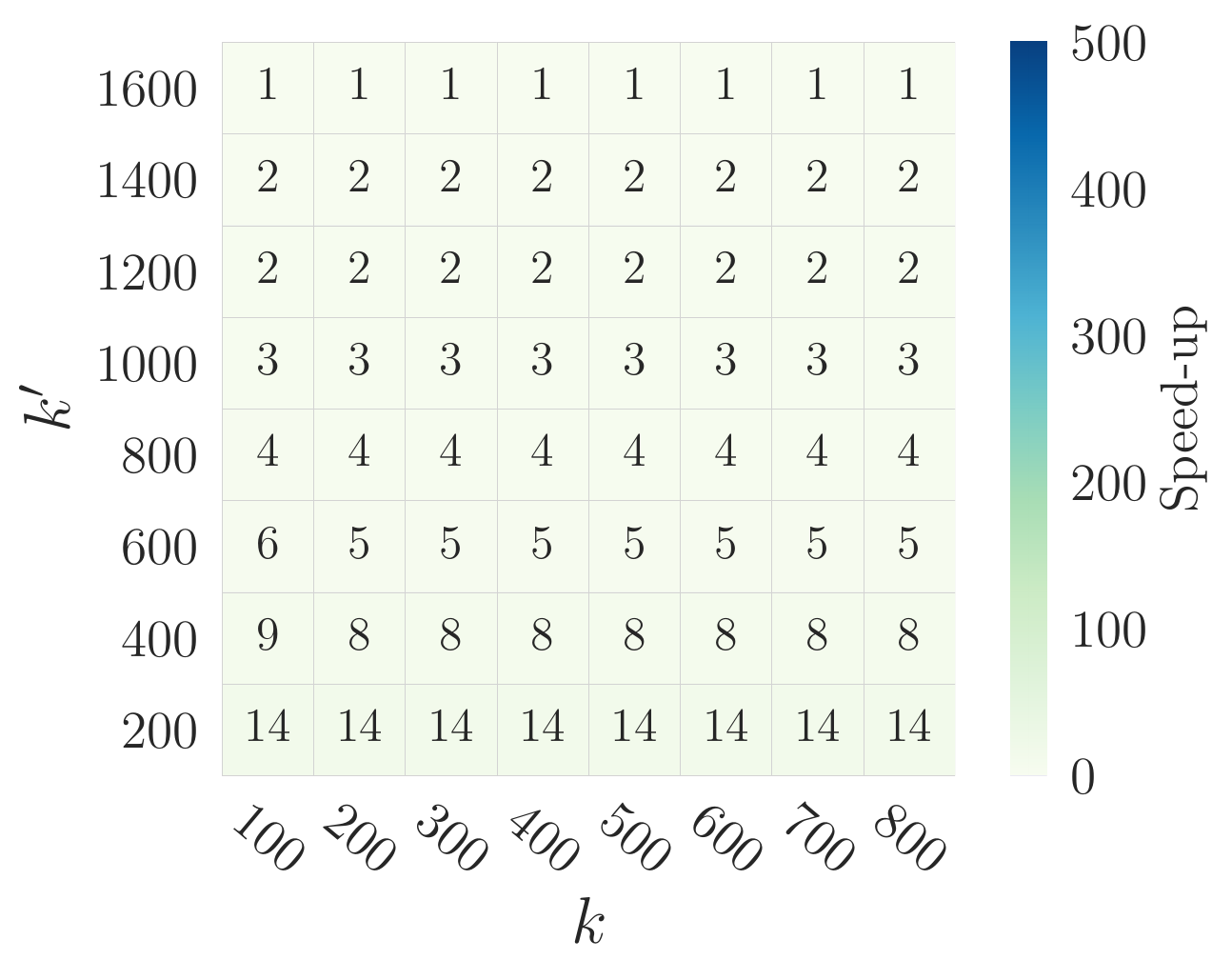} \label{fig:kac-bible-sp}}
	\subfloat[\slsh]{%
		\includegraphics[width=.25\textwidth]{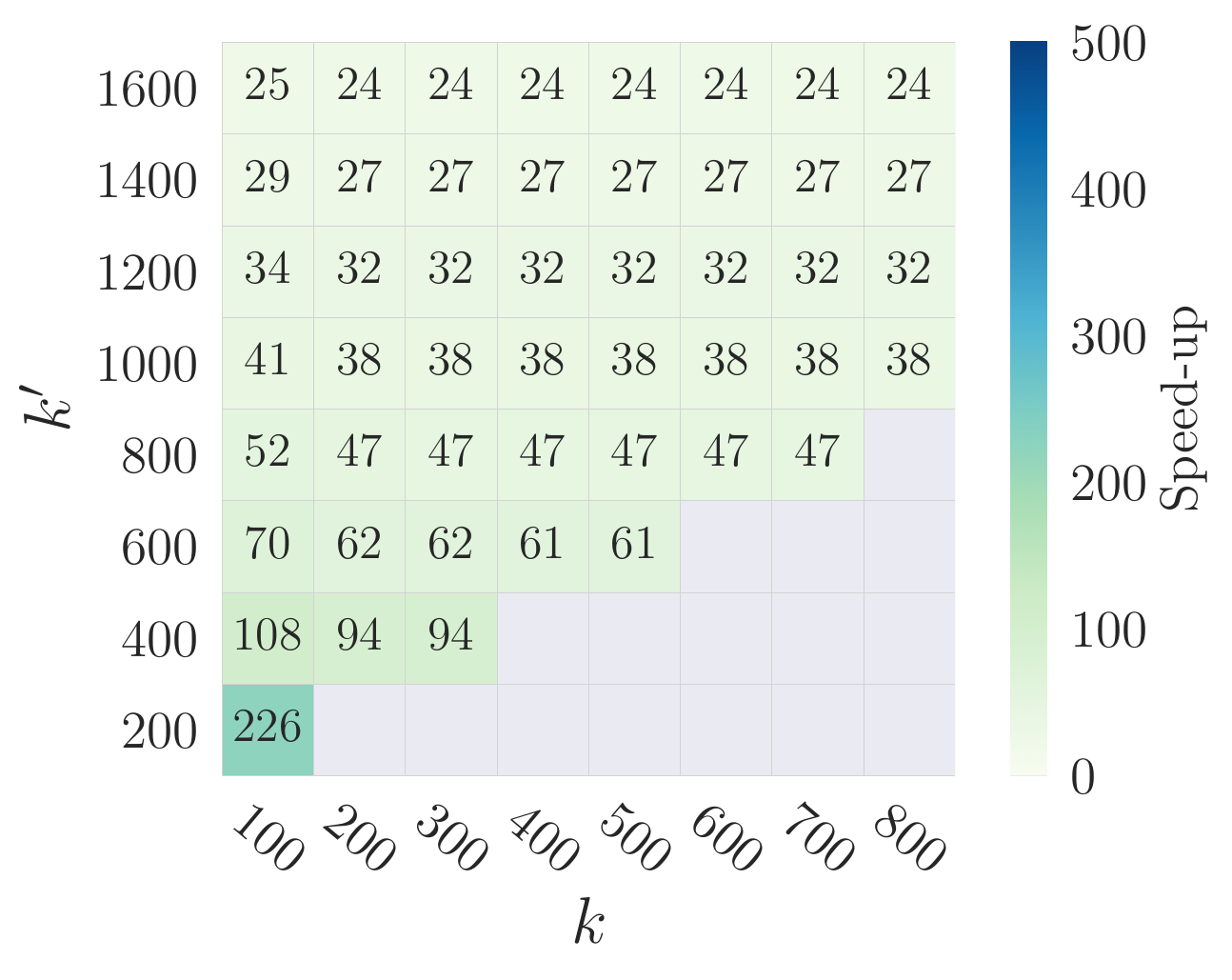} \label{fig:kac-slash-sp}}
	\hfill 
	\caption{\spd, for $\gamma=0.8$, $\gamma' = 1.0$, and \mnsz{}$=3$.}
	\vspace{-4ex}
\end{figure*}
\begin{figure*}[t]
	\captionsetup[subfigure]{justification=centering}
	\centering
	\subfloat[\advo]{%
		\includegraphics[width=.25\textwidth] {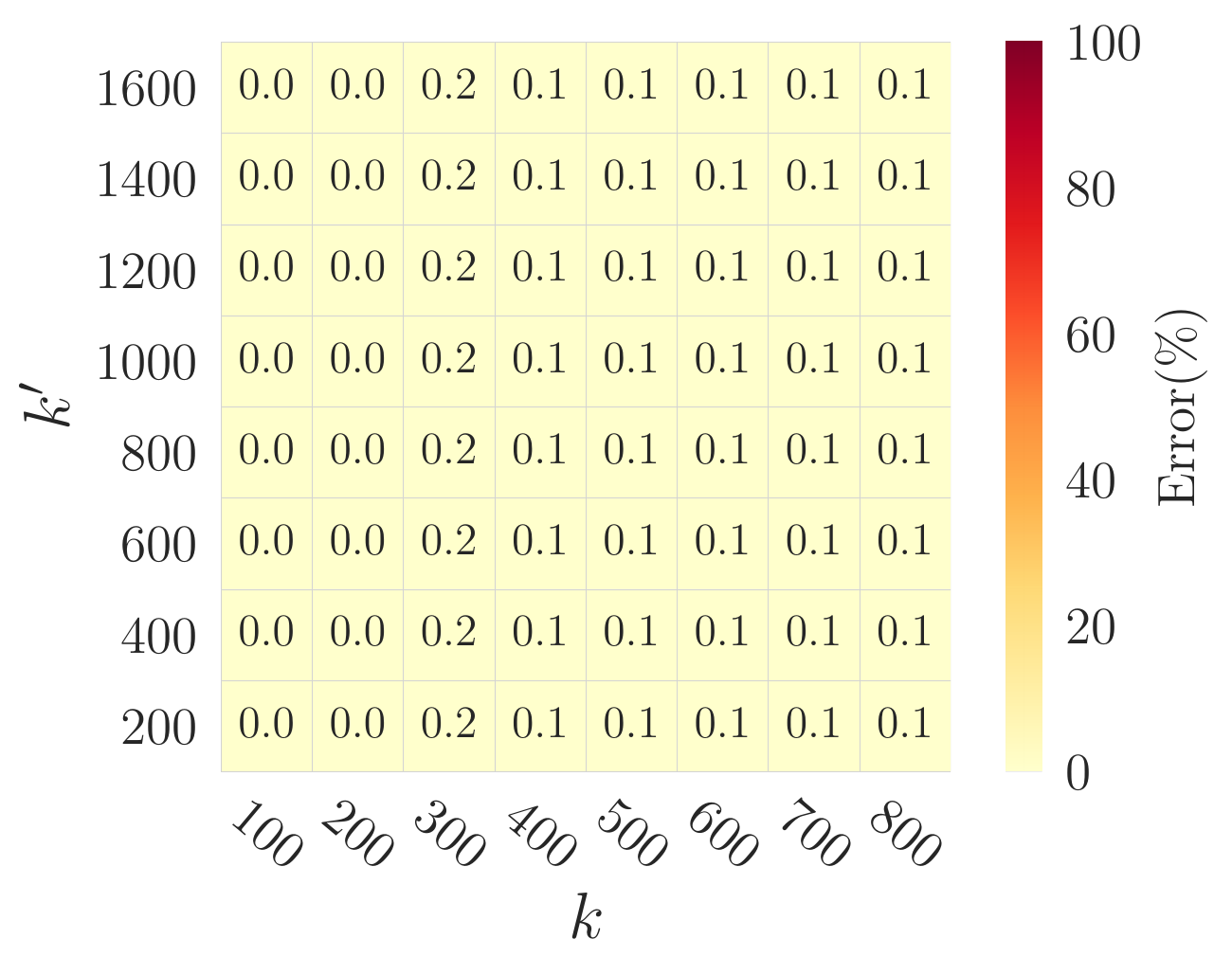} \label{fig:kac-advo-ac}}
	\subfloat[\route]{%
		\includegraphics[width=.25\textwidth] {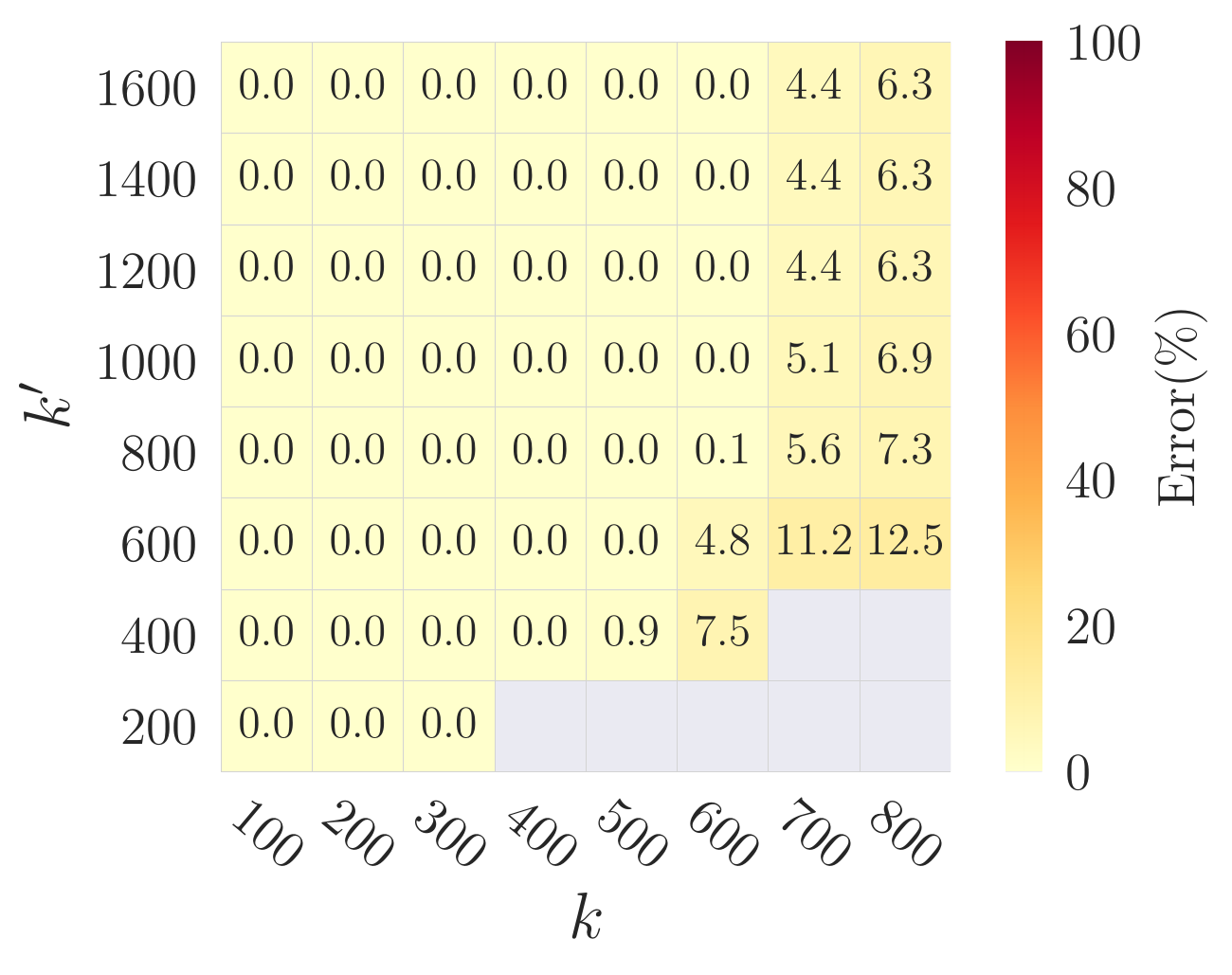} \label{fig:kac-route-ac}}
	\subfloat[\bible]{%
		\includegraphics[width=.25\textwidth] {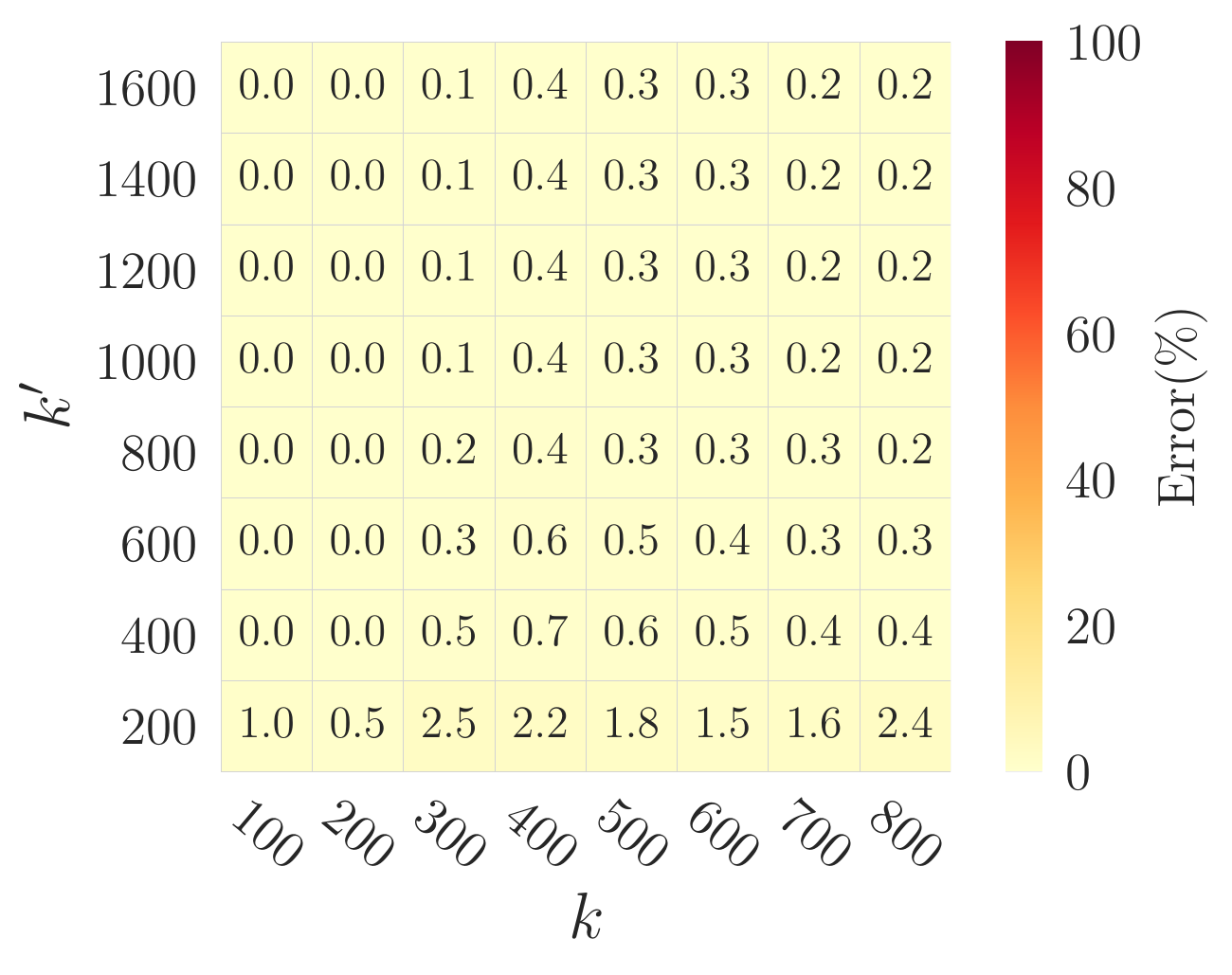} \label{fig:kac-bible-ac}}
	\subfloat[\slsh]{%
		\includegraphics[width=.25\textwidth] {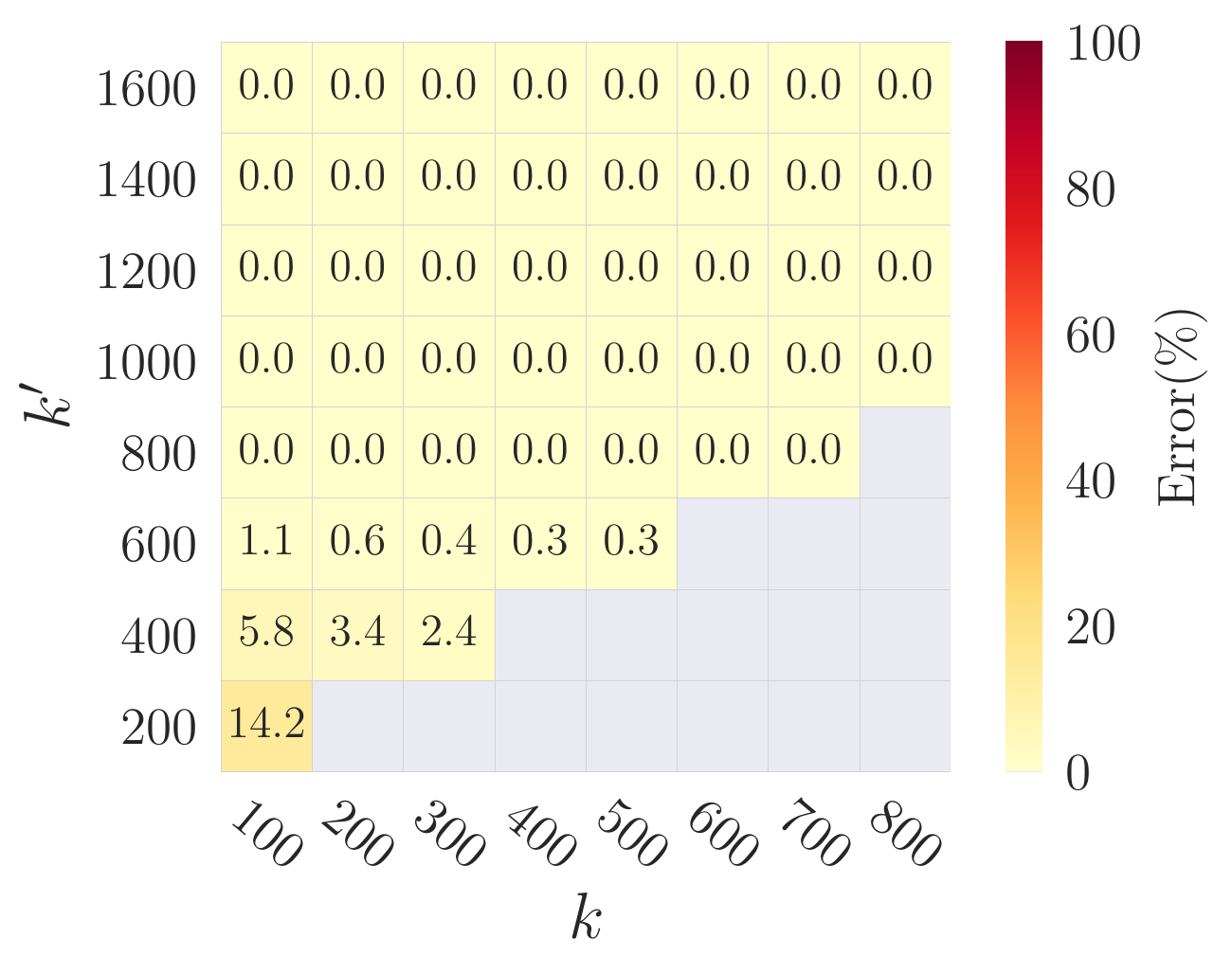} \label{fig:kac-slash-ac}}
	\hfill 
	\caption{\ac, for $\gamma=0.8$, $\gamma' = 1.0$, and \mnsz{}$=3$.}
	\vspace{-3ex}
\end{figure*}

\textbf{Runtime Compared to \pl{}:} The experiments show that \kqc{} yields a significant speedup over the \pl{}, for enumerating Top-$k$-quasi-cliques. For example, in \Cref{fig:advo-0.7} in graph \advo{} with $\gamma = 0.7$, $k = 100$, and $k' = 300$, when we set $\gamma' = 0.9$, \kqc{} yields a speedup of $984$x over \pl{}. For $\gamma=0.6$, the speedup is even more since \pl{} did not finish after $259$K secs while \kqc{} took only $172$ secs.\footnote{K stands for thousands} For the graph \bible{}, with $\gamma = 0.6$ and $\gamma' = 0.8$, \kqc{}  yields a $34$x speedup over \pl{} (\Cref{fig:bible-0.6}). In \slsh{} and the same values for $\gamma$ and $\gamma'$, \kqc{} yields a $638$x speedup over \pl{} (\Cref{fig:slash-0.6}). On the \route{} graph, the speedup is not high, especially for large values of $\gamma$. The reason is that this graph is not very dense, and even the \pl{} had a small runtime ($<100$ secs). For this graph, obtaining high speedups is not as important.


\textbf{Error Rate:} The results show that \kqc{} has a high accuracy for different graphs and various parameter settings while achieving a huge speedup over \pl{}. As shown in \Cref{fig:advo-gp,fig:bible-gp,fig:slash-gp,fig:route-gp}, error percentage for most graphs and different values of $\gamma$ and $\gamma'$ is less than $0.9\%$, and is often zero (i.e. exactly matches with the output of \pl{}). The highest error among all experiments is $2.1\%$ and belongs to the graph \slsh{} where $\gamma$ is $0.8$ (\Cref{fig:slash-0.8}). We did not report the error percent of \Cref{fig:advo-0.6} because the \pl{}  did not finish after $259$K secs.
\subsection{Dependence on $\gamma'$}
In \Cref{fig:advo-gp,fig:bible-gp,fig:slash-gp,fig:route-gp}, we ran \kqc{}  for different values of $\gamma'$ and $\gamma$ on four graphs \advo, \route, \bible, and \slsh. The purpose of these experiments is to understand the effects of the user-defined parameters $\gamma$ and $\gamma'$ on \kqc{} in terms of accuracy and time-efficiency. These experiments are helpful to choose an optimum value for $\gamma'$ which can result in a good balance between accuracy and runtime. For all graphs, we set $k = 100$, $k' = 300$ to find Top-$k$ quasi-cliques, and \mnsz{} $ = 5$ which is the minimum size threshold of quasi-cliques. 

\textbf{Performance of kernel detection and expansion:} Here, we show how different values of $\gamma$ and $\gamma'$ can have effect on two parts of \kqc{}. When $\gamma'$ is close to $\gamma$ kernel detection is slower than kernel expansion (\Cref{fig:advo-gp,fig:bible-gp,fig:slash-gp,fig:route-gp}). This is because kernel detection needs to extract all $\gamma'$-quasi-cliques. As shown in \Cref{fig:g-vs-time,fig:time-vs-gamma}, it requires greater computation time to mine all $\gamma'$-quasi-cliques when $\gamma'$ is smaller. On the other hand, if $\gamma'$ is larger, the $\gamma'$-quasi-cliques found by kernel detection are smaller. This fact keeps the size of $\gamma'$-quasi-cliques small, at least in the graphs we used. Therefore, since the kernel expansion phase starts with smaller kernels, it requires more time to explore larger \qc{}s. The ideal choice of $\gamma'$ should balance between the costs of the two phases, kernel expansion and kernel detection.

\subsection{Impact of the minimum size threshold (\mnsz)} 
\Cref{fig:mnszt} represents the runtimes of \kqc{} and \pl{} for different values of \mnsz{}. Based on the runtimes of \kqc{}, one can see that \mnsz{} does not have a major impact on the runtime, for the most part. However, in some cases, such as \Cref{fig:mnszt-bible} and for \mnsz{} $ = 10$, the runtime of \kqc{} decreases drastically. This is because the size of largest \qc{} in the graph \bible{} is $12$. Setting \mnsz{}$ = 10$ can considerably reduce the search space for \kqc{}, which leads to a decrease in runtime.

\mnsz{} is a user-defined parameter. When there is no knowledge about the given graph and \mnsz{} is set to a high value, it is possible that the graph does not contain $k$ \qc{}s with the size at least \mnsz{}, let alone top-$k$ maximal \qc{}s. This can cause us to miss large \qc{}s which could be good candidates to be placed in top-$k$ maximal \qc{}s. One way to handle this is to start with a high value of \mnsz{} and decrease it if enough quasi-cliques are not found with prior settings.


\remove{
\begin{figure*}[t]
	\includegraphics[width=0.3\textwidth]{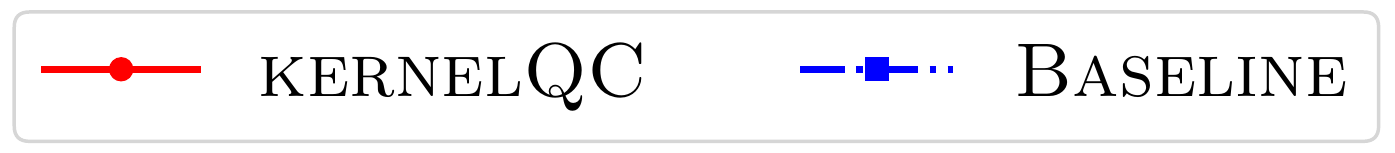}
	\vspace{-4ex}
\end{figure*}
\begin{figure*}[t]
	\captionsetup[subfigure]{justification=centering}
	\centering
	
	\subfloat[\advo]{%
		\includegraphics[width=.25\textwidth] {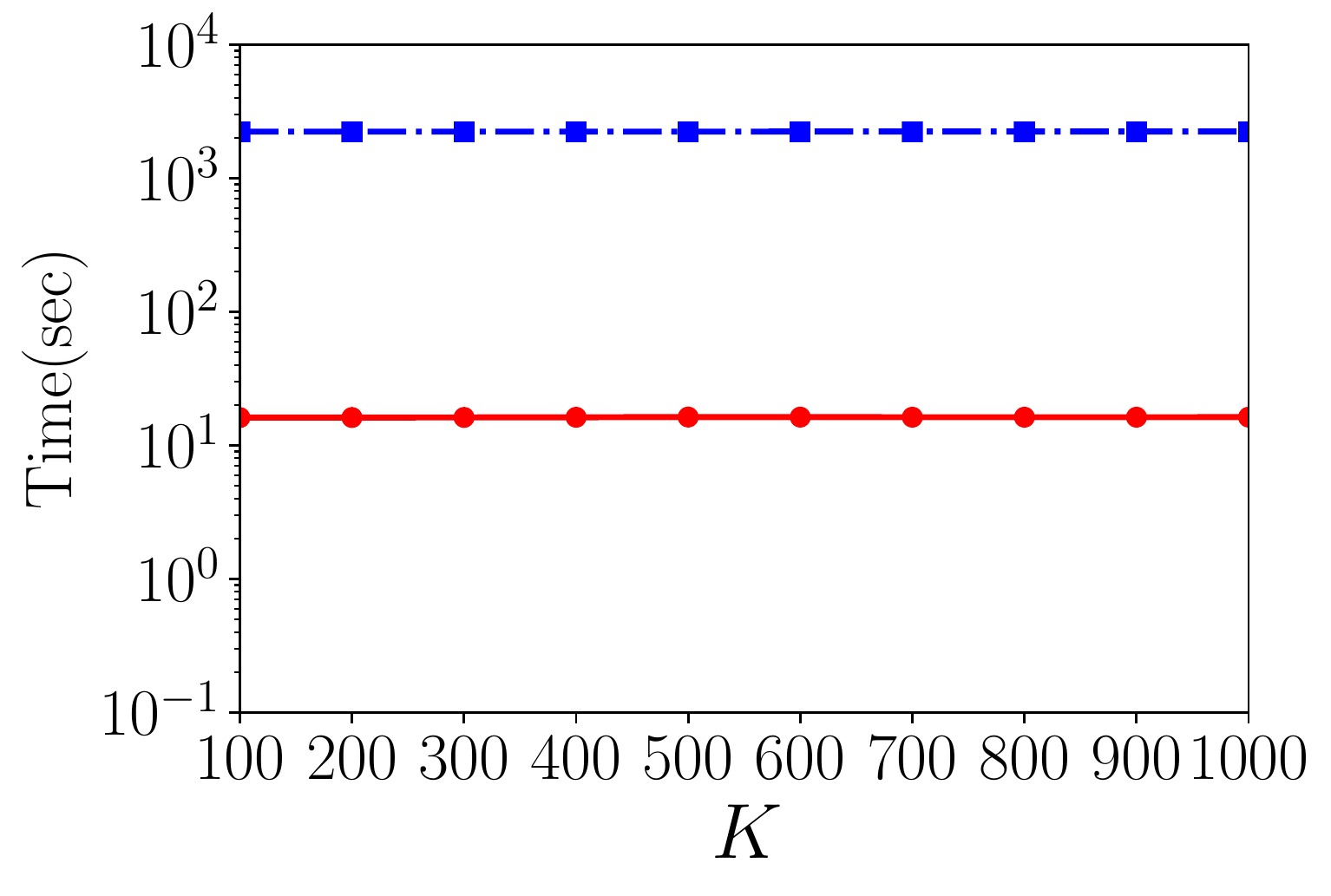}}
	\subfloat[\route]{%
		\includegraphics[width=.25\textwidth]{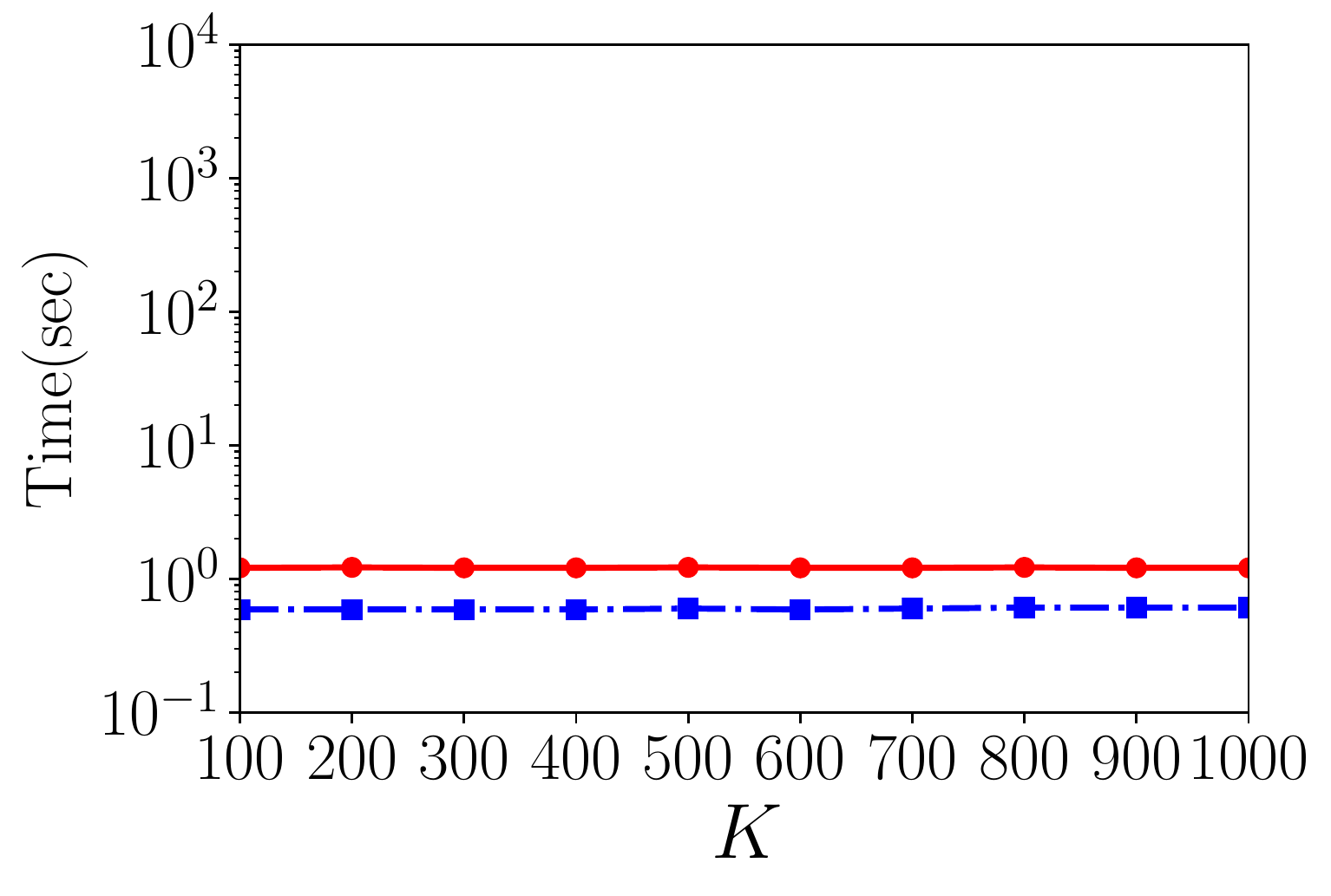}}
	\subfloat[\bible]{%
		\includegraphics[width=.25\textwidth] {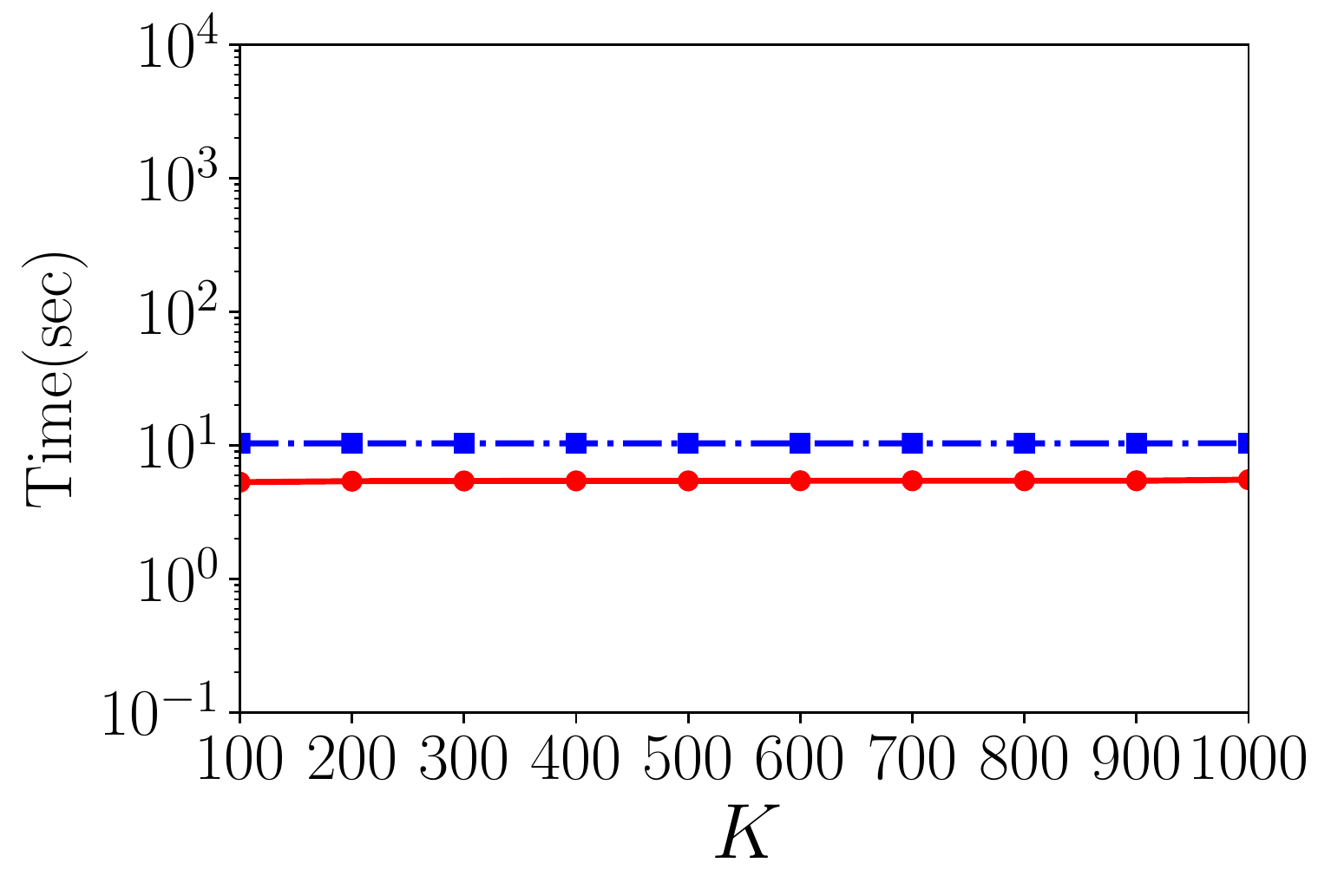}}
	\subfloat[\slsh]{%
		\includegraphics[width=.25\textwidth]{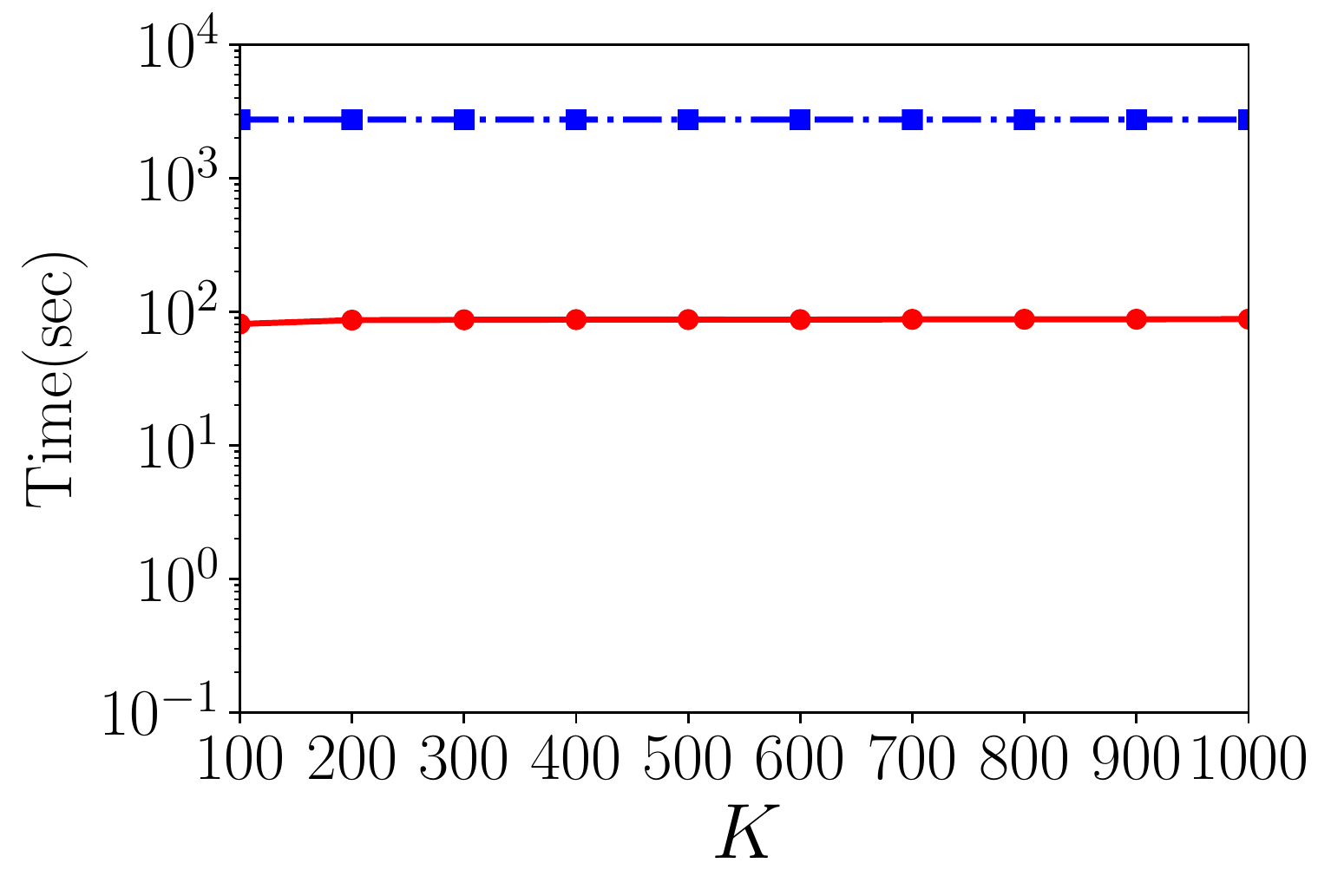}}
	\caption{The runtimes of \kqc{} and \pl{} for different values of $k$ where $k' = 1200$}
	\label{fig:k-vs-time}
\end{figure*}
}

\subsection{Dependence on $k$ and $k'$}
We consider different values of $k$ and $k'$. \Cref{fig:kac-advo-ac,fig:kac-route-ac,fig:kac-bible-ac,fig:kac-slash-ac} show the error percent of \kqc{}. More specifically, each cell shows the error percent for corresponding values of $k$ and $k'$. 
In addition, \Cref{fig:kac-advo-sp,fig:kac-route-sp,fig:kac-bible-sp,fig:kac-slash-sp} represent the speedup factor of \kqc{} over \pl{}. Similarly, each cell in these figures represent a speedup factor of \kqc{} over \pl{}. There are also some empty cells (for example in \Cref{fig:kac-slash-ac,fig:kac-slash-sp}). The empty cells indicate that in some graphs \kqc{} could not extract $k$ maximal \qc{s} with a given value of $k'$ due to a very few number of maximal quasi-cliques. Therefore, we did not report the error percent and speedup factors for those cases.

\textbf{Speedup compared to \pl{}:}
\Cref{fig:kac-advo-sp,fig:kac-route-sp,fig:kac-bible-sp,fig:kac-slash-sp} represent the speedup factor of \kqc{} over \pl{}. Based on the results, an increase in value of $k'$ makes \kqc{} slower compared to \pl{}. For example, in \Cref{fig:kac-slash-sp}, for $k = 100$ and $k' = 200$, the speedup of \kqc over \pl is $226$x while for the same value of $k$ and $k' = 400$, it is reduced to $108$x. The reason is that a higher value of $k'$ in \kqc{} means more number of kernels. Therefore, the kernel expansion phase needs to expand more kernels, which increases the overall runtime.

\textbf{Error rate:}
Here, we describe the effect of parameter $k'$ on the accuracy of \kqc{}.  As shown in \Cref{fig:kac-advo-ac,fig:kac-route-ac,fig:kac-bible-ac,fig:kac-slash-ac}, the higher value of $k'$ we set, the lower error we obtain. For example, in \Cref{fig:kac-slash-ac}, for $k = 100$ and $k' = 200$, \kqc{} results in $14.2\%$ error while increasing the value of $k'$ to $800$ can yield zero percent error. The reason is that by setting higher values for $k'$, we retrieve more $\gamma'$-quasi-cliques in kernel detection of \kqc{}, and there are more kernels to be expanded by the kernel expansion of \kqc{}. In other words, a high value for $k'$ can increase the chance of \kqc{} to unearth very large \qc{s}. For a fixed value of $k'$, the error percent of different values of $k$ fluctuates slightly in most cases. Here, we give an example why an increase in value of $k$ can results in both lower and higher error percent. Let assume the size of \qc{}s returned by \kqc{} is $H = \langle10,10,9\rangle$ (for $k = 3$), and the size of \qc{}s returned by the exact algorithm (\pl{}) is $Z = \langle12,10,10\rangle$. Based on the error metric we used (See \Cref{eq:error-percent}), the error percent of \kqc{} in this case is $9.3\%$. For $k = 4$, suppose that the returned list by \pl{} is $Z = \langle12,10,10,9\rangle$. The error percent can be lowered if \kqc{} returns $H = \langle10,10,9,9\rangle$, where the error is $7.3\%$. It can be also greater if \kqc{} returns $H = \langle10,10,9,8\rangle$, where the error is $9.7\%$.

\begin{table}[t!]
	\begin{tabular}{|l|c|c|c|c|c|}
		\hline
		\multicolumn{1}{|c|}{\multirow{2}{*}{Graph}} & \multicolumn{4}{c|}{\kqc}       & \pl                       \\ \cline{2-6} 
		\multicolumn{1}{|c|}{}                       & $\gamma$    & $\gamma'$   & {\footnotesize {\texttt{Avg sz}}} & Time(sec) & Time(sec)                      \\ \hline
		\live                                  & 0.85 & 1.0  & 24.1 & 843       & $>$ 259K secs \\ \hline
		\youtube                                      & 0.8  & 1.0  & 27.2 & 3130      &  $>$ 259K secs \\ \hline
		\hyves                                        & 0.75 & 0.95 & 33.4 & 9026      &  $>$ 259K secs \\ \hline
	\end{tabular}
	\caption{{\small Performance of \kqc on large graphs.  $k=100, k'=300, \mnsz=5$. {\footnotesize {\texttt{Avg sz}}} shows the average size of $k$ quasi-cliques. $>259$K means \pl did not finish in $72$ hours (note this happens with every graph).}}
	\label{table:lrg-graphs}
\end{table}

\subsection{Performance of \kqc on Large graphs}
Our method can handle larger graphs. As shown in \Cref{table:lrg-graphs}, \kqc is able to retrieve large maximal quasi-cliques on the graphs with millions of edges and vertices. For example, \kqc lists $100$ maximal quasi-cliques in \num{3130} and \num{9026} secs respectively for the graphs \youtube and \hyves while \pl does not finish after $259$K secs (72 hours). The speedup is even higher in the graph \live, where \kqc takes only $843$ secs while \pl did not finish in $72$ hours. This is because \live has a higher average degree, hence denser than other graphs (See \cref{table:graph-stats} for more details). Therefore, the search space for \pl in this graph can be huge while \kqc quickly enumerates kernels of \live and then expands them to obtain $k$ maximal $\gamma$-quasi-cliques.

\vspace{-1ex}
\section{Conclusions}
\label{sec:concl}
Quasi-clique enumeration is an important problem in the area of dense subgraph enumeration. We considered the problem of enumerating top-$k$ maximal degree-based quasi-cliques from a graph. We first showed that it is NP-hard to even determine whether a given (degree-based) quasi-clique is maximal. We then presented a novel heuristic algorithm \kqc for enumerating top-$k$ maximal quasi-cliques, based on an idea of finding dense kernels, followed by expanding them into larger quasi-cliques. Our experiments showed that \kqc can often lead to a speedup of three orders of magnitude, when compared with a state-of-the-art baseline algorithm. This implies that it may be possible to mine quasi-cliques from larger graphs than was possible earlier. Many directions remain to be explored, including the following: (1) Can the idea of detecting and expanding kernels be applied to other incomplete dense structures, such as quasi-bicliques? (2) Can the algorithms for quasi-cliques be parallelized effectively?
\vspace{-2ex}

\bibliographystyle{IEEEtran}
\bibliography{IEEEabrv,ref}

\end{document}